\documentclass{sig-alternate}

\usepackage{enumitem}\setitemize{leftmargin=12pt}\setenumerate{leftmargin=12pt}

\usepackage{times}
\usepackage{theorem} 
\usepackage{stmaryrd} 
\usepackage{epsfig} 
\usepackage{misc}
\usepackage{amsmath,amssymb}
\usepackage{bbm}
\usepackage{bm}

\usepackage{boxedminipage}
\usepackage{pifont}
\usepackage{graphicx}
\usepackage{xspace} 
\usepackage{color}
\usepackage{url}

\newcommand{\MDD}{MDDlog\xspace}

\newcommand{\FGDD}{frontier-guarded DDlog\xspace}
\newcommand{\DD}{DDlog\xspace}
\newcommand{\adom}{\ensuremath{\mathsf{adom}}}
\newcommand{\coNP}{\textnormal{\sc coNP}\xspace}
\newcommand{\forb}{\mn{Forb}}

 \newtheorem{definition}{Definition}
 \newtheorem{theorem}{Theorem}
 \newtheorem{example}{Example}
 
 \newtheorem{lemma}{Lemma}
 \newtheorem{proposition}{Proposition}

\begin{document}

\conferenceinfo{PODS'13,} {June 22--27, 2013, New York, New York, USA.}
\CopyrightYear{2013}
\crdata{978-1-4503-2066-5/13/06}
\clubpenalty=10000
\widowpenalty = 10000

\title{Ontology-based Data Access:\\ A Study through Disjunctive
  Datalog, CSP, and MMSNP}

\numberofauthors{4} 
\addtolength{\auwidth}{-6mm}

\author{
\alignauthor
Meghyn Bienvenu\\
       \affaddr{\mbox{CNRS \& Universit\'{e} Paris Sud}}\\
       \affaddr{Orsay, France}
\alignauthor
Balder ten Cate\\
       \affaddr{UC Santa Cruz}\\
       \affaddr{Santa Cruz, CA, USA}
\alignauthor
Carsten Lutz\\
       \affaddr{University of Bremen}\\
       \affaddr{Bremen, Germany}
\alignauthor Frank Wolter\\
       \affaddr{University of Liverpool}\\
             \affaddr{Liverpool, UK} 
}       
       
 \maketitle   

\begin{abstract}
  \emph{Ontology-based data access} is concerned with querying
  incomplete data sources in the presence of domain-specific knowledge
  provided by an ontology.  A central notion in this setting is that
  of an \emph{ontology-mediated query}, which is a database query
  coupled with an ontology.
  In this paper, we study 
  several classes of
  ontology-mediated queries, where the database queries are given as
  some form of conjunctive query  and the
  ontologies are formulated in description logics or other relevant fragments
  of first-order logic, such as the guarded fragment and the
  unary-negation fragment.
 The contributions of the paper are three-fold. First, we
 characterize the expressive power of ontology-mediated queries in terms
 of fragments of disjunctive datalog.  
 Second, we establish intimate
 connections between ontology-mediated queries and constraint
 satisfaction problems (CSPs) and their logical generalization,
 MMSNP formulas. Third, we exploit these connections to obtain new
 results regarding (i) first-order rewritability and
 datalog-rewritability of ontology-mediated queries, (ii) P/NP
 dichotomies for ontology-mediated queries, and (iii) the query
 containment problem for ontology-mediated queries.
 


\end{abstract}

\section{Introduction}

\noindent
%
Ontologies are logical theories that formalize domain-specific
knowledge, thereby making it available for machine processing.  Recent
years have seen an increasing interest in using ontologies in
data-intensive applications, especially in the context of intelligent
systems, the semantic web, and in data integration.  A much studied
scenario is that of answering queries over 
an incomplete database under the open world semantics,
taking into account
knowledge provided by an ontology
\cite{DBLP:conf/pods/CalvaneseGL98,Romans,DBLP:journals/ai/CaliGP12}. We
refer to this
as \emph{ontology-based data access (OBDA)}.

There are several important use cases for
OBDA. 
A classical one is to enrich an incomplete data source with background
knowledge, in order to obtain a more complete set
of 
answers to a query. For example, if a medical patient database
contains the facts that patient1 has finding Erythema Migrans and
patient2 has finding Lyme disease, and the ontology provides the
background knowledge that a finding of Erythema Migrans is sufficient
for diagnosing Lyme disease, then both patient1 and patient2 can be
returned when querying for patients that have the diagnosis Lyme
disease.  This use of ontologies is also central to query answering in
the semantic web.
OBDA can also be used to enrich the data schema (that is, the relation
symbols used in the presentation of the data) with additional symbols
to be used in a query. 
For example,
a patient database may contain facts such as patient1 has diagnosis
Lyme disease and patient2 has diagnosis Listeriosis, and an ontology
could add the knowledge that Lyme disease and Listeriosis are both
bacterial infections, thus enabling queries such as ``return all
patients with a bacterial infection'' despite the fact
that 
%
the data schema does not include a relation or attribute explicitly referring to bacterial infections.
Especially in the bio-medical domain, applications of this kind are
fueled by the availability of comprehensive
professional ontologies such as {\sc
  Snomed CT} and {\sc FMA}.
%
A third prominent application of OBDA is in data integration, where an
ontology can be used to provide a uniform view on multiple data
sources \cite{DBLP:journals/jods/PoggiLCGLR08}. 
This typically 
involves mappings from the source schemas to the schema 
of the ontology, 
which we will not explicitly consider here.

We may
view 
the actual database query and the ontology as two components of one
composite query, which we call an \emph{ontology-mediated query}.
OBDA can then be described as the problem of answering
ontology-mediated queries. The database queries used in OBDA are
typically unions of conjunctive queries, while the ontologies are
typically specified in an ontology language that is either a
description logic, or, more generally, a suitable fragment of
first-order logic. For popular choices of ontology languages, the data
complexity of ontology-mediated queries can be \coNP-complete, which
has resulted in extensive research on finding tractable classes of
ontology-mediated queries, as well as on finding classes of
ontology-mediated queries that are amenable to efficient query
answering techniques
\cite{DBLP:conf/kr/CalvaneseGLLR06,DBLP:journals/jar/HustadtMS07,DBLP:conf/dlog/KrisnadhiL07}.
In particular, relevant classes of ontology-mediated queries have been
identified that admit an FO-rewriting 
(i.e., that are equivalent to a first-order query),
or, alternatively, admit a datalog-rewriting.  FO-rewritings make
it possible to answer ontology-based queries using traditional
database management systems. This is considered one of the most
promising approaches for OBDA, and is currently the subject of
significant research activity, see for example
\cite{Romans,DBLP:conf/kr/GottlobS12,DBLP:conf/icalp/KikotKPZ12,DBLP:conf/kr/KontchakovLTWZ10,PleaseCiteMyPaper}. 
%
%

The main aims of this paper are (i)~to characterize the expressive
power of ontology-mediated queries, both in terms of more traditional
database query languages and from a descriptive complexity
perspective and (ii)~to make progress towards complete and decidable classifications of
ontology-mediated queries, with respect to their data complexity, as well as
with respect to FO-rewritability and datalog-rewritability.

We take an ontology-mediated query to be a triple $(\Sbf,\Omc,q)$
where \Sbf is a \emph{data schema}, \Omc an ontology
, and $q$ a query. Here, the data schema \Sbf fixes the set of
relation symbols than can occur in the data and the ontology $\Omc$ is
a logical theory that may use the relation symbols from \Sbf as well
as additional symbols.  The query $q$ can use any relation symbol that
occurs in \Sbf or \Omc.  As ontology languages, we consider a range of
standard description logics (DLs) and several fragments of first-order
logic that embed ontology languages such as Datalog$^\pm$
\cite{Cali2009}, namely the guarded fragment (GF), the unary negation
fragment (UNFO), and the guarded negation fragment (GNFO).
As query languages for $q$, we focus on unions of conjunctive queries
(UCQs) and unary atomic queries (AQs). The latter are of the form
$A(x)$, with $A$ a unary relation symbol, and correspond to what are
traditionally called \emph{instance queries} in the OBDA literature
(note that $A$ may be a relation symbol from $\Omc$ that is not part
of the data schema). These two query languages are among the most used
query languages in OBDA.
In the
following, we use $(\Lmc,\Qmc)$ to denote the query language that
consists of all ontology-mediated queries $(\Sbf,\Omc,q)$ with \Omc
specified in the ontology language \Lmc and $q$ specified in the query language
\Qmc. For example, (GF,UCQ) refers to ontology-mediated queries in
which \Omc is a GF-ontology and $q$ is a UCQ. We refer to such query
languages $(\Lmc,\Qmc)$ as \emph{ontology-mediated query languages
  (or, OBDA languages)}. \looseness=-1

In Section~\ref{sect:alcUCQ}, we characterize
the expressive power of 
OBDA languages in terms of natural fragments of (negation-free) disjunctive datalog.
We first consider the basic description logic \ALC. 
We show that (\ALC,UCQ) has the
same expressive power as monadic disjunctive datalog (abbreviated
\MDD) and that (\ALC,AQ) has the same expressive power as unary
queries defined in a syntactic fragment of \MDD that we call
connected simple {\MDD}. 
Similar results hold for various description logics that extend \ALC
with, for example, 
inverse roles, role hierarchies, and the universal role, all of which
are standard operators included in the W3C-standardized ontology
language OWL2 DL.
%
%
Turning to other
fragments of first-order logic, we then show that (UNFO,UCQ) also has the same
expressive power as \MDD, while (GF,UCQ) and (GNFO,UCQ) are strictly
more expressive and coincide in expressive power with frontier-guarded
disjunctive datalog, which is the \DD\ fragment  given by programs in 
which, for every atom
$\alpha$ in the head of a rule, there is an atom $\beta$ in the rule
body that contains all variables from~$\alpha$.

In Sections~\ref{sect:obdammsnp} and~\ref{sect:OBDAtoCSP}, we study
ontology-mediated queries from a \emph{descriptive complexity}
perspective. In particular, we establish an intimate connection
between OBDA query languages, constraint satisfaction problems, and
MMSNP. Recall that constraint satisfaction problems (CSPs) form a
subclass of the complexity class \NP that, although it contains
\NP-hard problems, is in certain ways more computationally
well-behaved. The widely known Feder-Vardi conjecture
\cite{FederVardi} states that there is a dichotomy between \PTime and
\NP for the class of all CSPs, that is, each CSP is either in \PTime
or \NP-hard. In other words, the conjecture asserts that there
are no CSPs which are \NP-intermediate in the sense of Ladner's
theorem. Monotone monadic strict \NP without inequality (abbreviated
MMSNP) was introduced by Feder and Vardi as a logical generalization
of CSP that enjoys similar computational properties
\cite{FederVardi}. In particular, it was shown in
\cite{FederVardi,kun-derand} that there is a dichotomy between \PTime
and \NP for MMSNP sentences if and only if the Feder-Vardi conjecture
holds.

In Section~\ref{sect:obdammsnp}, we observe that (\ALC,UCQ) and many
other OBDA languages based on UCQs have the same expressive power as
the query language coMMSNP,
consisting of all queries whose complement
is definable by an MMSNP formula with free variables.
In the spirit of descriptive complexity theory, we say that (\ALC,UCQ)
\emph{captures} coMMSNP. In fact, this result is a consequence of the
results in Section~\ref{sect:alcUCQ} and the observation that \MDD has
the same expressive power as coMMSNP.  It has fundamental consequences
regarding the data complexity of ontology-mediated queries
and the containment problem for such queries, which we describe next. 

 First, we obtain that there is a
dichotomy between \PTime and \coNP for 
ontology-mediated queries from (\ALC,UCQ) if and only if the
Feder-Vardi conjecture holds, and similarly for many other OBDA
languages based on UCQs. To appreciate this result, recall that 
considerable effort has been directed towards identifying tractable
classes of ontology-mediated queries.
%
%
Ideally, one would like to 
classify 
the data complexity of every ontology-mediated query within a
given OBDA language such as (\ALC,UCQ). 
Our aforementioned result ties this task to
proving the Feder-Vardi conjecture.
Significant progress has been made in understanding the complexity of
CSPs and MMSNPs
\cite{DBLP:conf/csr/Bulatov11,DBLP:journals/siamdm/BodirskyCF12,DBLP:journals/ejc/KunN08},
and the connection established in this paper facilitates the transfer
of techniques and results from CSP and MMSNP in order to analyze the data
complexity of query evaluation in (\ALC,UCQ).  We also consider the
standard extension \ALCF of \ALC with functional roles and note that,
for query evaluation in (\ALCF,AQ), there is no dichotomy between
\PTime and \coNP unless \PTime = \NP.  

To establish a counterpart of (GF,UCQ) and (GNFO,UCQ) in the
MMSNP world, we introduce guarded monotone strict NP (abbreviated
GMSNP) as a generalization of MMSNP; specifically, GMSNP is obtained
from MMSNP by allowing guarded second-order quantification in the
place of monadic second-order quantification,
similarly as in the transition from \MDD to frontier-guarded
disjunctive datalog. The resulting query language coGMSNP has the same
expressive power as frontier-guarded disjunctive datalog, and
therefore, in particular, (GF,UCQ) and (GNFO,UCQ) capture coGMSNP.  We
observe that GMSNP has the same expressive power as the extension
MMSNP$_2$ of MMSNP proposed in \cite{Madeleine}.  It follows from our
results in Section~\ref{sect:alcUCQ} that GMSNP (and thus MMSNP$_2$)
is strictly more expressive than MMSNP, closing an open problem from
\cite{Madeleine}. We leave it as an open problem whether GMSNP is
computationally as well-behaved as MMSNP, that is, whether there is a
dichotomy between \PTime and \NP if the Feder-Vardi conjecture holds.

The second application of the connection between OBDA  and MMSNP
concerns query containment.
It was shown in \cite{FederVardi} that containment between MMSNP sentences is
decidable. We use this result to prove that query containment is
decidable for many OBDA languages based on UCQs, including (\ALC,UCQ)
and (GF,UCQ). Note that this refers to a very general form of query
containment in OBDA, as recently introduced and studied in
\cite{KR12-cont}.  For (\ALCF,AQ), this problem (and every other decision problem
discussed below) turns out to be undecidable.

In Section~\ref{sect:OBDAtoCSP}, we consider OBDA languages based on
atomic queries and establish a tight connection to (certain
generalizations of) CSPs. This connection is most easily stated for
\emph{Boolean} atomic queries (BAQs): we prove that ($\ALC$,BAQ)
captures the query language that consists of all Boolean queries
definable as the complement of a CSP.  Similarly ($\ALC$,AQ) extended
with the universal role captures the query language that consists of
all unary queries definable as the complement of a \emph{generalized
  CSP}, which is given by a finite collection of structures enriched
with a constant symbol.  We then proceed to transfer results from the
CSP literature to the ontology-mediated query languages ($\ALC$, BAQ)
and ($\ALC$, AQ).  First we immediately obtain that the existence of a
\PTime/\coNP dichotomy for these ontology-mediated query languages is
equivalent to the Feder-Vardi conjecture.
Then we show that 
query containment is not only decidable (as we could already conclude from 
the connection with coMMSNP described in
Section~\ref{sect:obdammsnp}), but, in fact, \NExpTime-complete.
Finally, taking advantage of recent results for CSPs \cite{DBLP:journals/lmcs/LaroseLT07,FreeseEtAl,BulatovWidth}, we are
able to show that FO-rewritability and datalog-rewritability, as
properties of ontology-mediated queries, are decidable and \NExpTime-complete
for (\ALC, AQ) and (\ALC,BAQ). 

The results in Sections~\ref{sect:obdammsnp} and~\ref{sect:OBDAtoCSP}
just summarized are actually proved not only for \ALC, but also for several
of its extensions. This relies on the equivalences
%
between DL-based OBDA-languages established in
Section~\ref{sect:alcUCQ}.
%

\medskip\par\noindent\textbf{Related Work} A connection between query
answering in DLs and the negation-free fragment of disjunctive datalog
was first discovered and utilized in the influential
\cite{DBLP:phd/de/Motik2006,DBLP:journals/jar/HustadtMS07}, see also
\cite{DBLP:journals/corr/abs-1202-0914}.  This research is
concerned with answer-preserving translations of ontology-mediated
queries into disjunctive datalog. In contrast to the current paper,
it does not consider the expressive power of ontology-mediated
queries, nor their descriptive complexity.  A connection between DL-based OBDA
 and CSPs was first found and exploited in
\cite{KR12-csp}, in a setup that is different from the one studied in
this paper. In particular, instead of focusing on ontology-mediated
queries that consist of a data schema, an ontology, and a database
query, \cite{KR12-csp} concentrates on ontologies while quantifying
universally over all database queries and without fixing a data
schema.  It establishes links to the Feder-Vardi conjecture that
are incomparable to the ones found in this paper, and does not
consider the expressive power and descriptive complexity of queries
used in OBDA. \looseness=-1

\section{Preliminaries}
\label{sect:prelims}

\noindent\textbf{Schemas, Instances, and Queries.}  A \emph{schema} is a finite
collection $\Sbf=(S_1,\dots,S_k)$ of relation symbols with associated
arity. A \emph{fact}
over $\Sbf$ is an expression of the form $S(a_1,\ldots, a_n)$ where
$S\in\Sbf$ is an $n$-ary relation symbol, and $a_1, \ldots, a_n$ are
elements of some fixed, countably infinite set $\mn{const}$ of
\emph{constants}.  An \emph{instance} $\Dmf$ over $\Sbf$ is a finite
set of facts over $\Sbf$.  The \emph{active domain} $\adom(\Dmf)$
of~\Dmf is the set of all constants that occur in the facts
of~\Dmf.  
We will frequently use boldface notation for tuples, such as in $\abf=(a_1,
\ldots, a_n)$, and we denote by $()$ the empty tuple.

A \emph{query over} \Sbf is semantically defined as a mapping $q$ that
associates with every instance \Dmf over \Sbf a set of \emph{answers}
$q(\Dmf) \subseteq \adom(\Dmf)^{n}$, where $n\geq 0$ is the
\emph{arity} of $q$. If $n=0$, then we say that $q$ is a \emph{Boolean
  query}, and we write $q(\Dmf)=1$ if $()\in q(\Dmf)$ and
\mbox{$q(\Dmf)=0$} otherwise.


A prominent way of specifying queries is by means of first-order logic
(FO). Specifically, each schema $\Sbf$ and domain-independent
FO-formula $\vp(x_1,\dots,x_n)$ that uses only relation names from
$\Sbf$ (and, possibly, equality) give rise to the $n$-ary query
$q_{\varphi,\Sbf}$, defined by setting for all $\Sbf$-instances~$\Dmf$,
$$
q_{\varphi,\Sbf}(\Dmf) = \{ (a_{1},\ldots,a_{n})\in \adom(\Dmf)^{n} \mid \Dmf\models \varphi[a_{1},\ldots,a_{n}]\}.
$$
To simplify exposition, we assume that FO-queries do not contain
constants.  We use FOQ to denote the set of all first-order queries,
as defined above. Similarly, we use CQ and UCQ to refer to the class
of conjunctive queries and unions of conjunctive queries, defined as
usual and allowing the use of equality.  AQ denotes the set of
\emph{atomic queries}, which are of the form $A(x)$ with $A$ a unary
relation symbol. Each of these is called a \emph{query language},
which is defined abstractly as a set of queries.
Besides FOQ, CQ, UCQ, and AQ, we consider various other query
languages introduced later, including ontology-mediated ones and
variants of datalog.

Two queries $q_1$ and $q_2$ over \Sbf are \emph{equivalent}, written
$q_1 \equiv q_2$, if for every $\Sbf$-instance \Dmf, we have
$q_1(\Dmf)=q_2(\Dmf)$.  We say that query language $\Qmc_2$ is
\emph{at least as expressive as} query language $\Qmc_1$, written
$\Qmc_1 \preceq \Qmc_2$, if for every query $q_1 \in \Qmc_1$ over some
schema \Sbf, there is a query $q_2 \in \Qmc_2$ over \Sbf with $q_1
\equiv q_2$.  $\Qmc_1$ and $\Qmc_2$ \emph{have the same expressive
  power} if $\Qmc_1 \preceq \Qmc_2 \preceq \Qmc_1$.
%


%

\smallskip\noindent\textbf{Ontology-Mediated Queries.}  We introduce
the fundamentals of ontology-based data access.  An \emph{ontology
  language}~\Lmc is a fragment of first-order logic (i.e., a set of FO
sentences), and an \emph{\Lmc-ontology} \Omc is a finite set of
sentences from 
~\Lmc.  We introduce various ontology
languages throughout the paper, including descriptions logics 
and the guarded fragment.

An \emph{ontology-mediated query} over a schema $\Sbf$ is a triple
$(\Sbf,\Omc,q)$, where \Omc is an ontology and $q$ a query over $\Sbf
\cup \mn{sig}(\Omc)$, with $\mn{sig}(\Omc)$ the set of relation symbols
used in \Omc. Here, we call $\Sbf$ the \emph{data schema}.  Note that
the ontology can introduce symbols that are not in the data schema.
As explained in the introduction, this allows the ontology to enrich
the schema of the query $q$.  Of course, we do not require
that every relation of the data schema needs to occur in the ontology.
%
%
 We have explicitly included $\Sbf$
in the specification of the ontology-mediated query to emphasize that the
ontology-mediated query is 
interpreted as a query over~$\Sbf$.


The semantics of an ontology-mediated query is given in terms of
\emph{certain answers}, defined next. A \emph{finite relational
  structure} over a schema \Sbf is a pair $\Bmf=(\mn{dom},\Dmf)$ where
$\mn{dom}$ is a non-empty finite set called the \emph{domain} of \Bmf and \Dmf
is an instance over \Sbf with $\mn{adom}(\Dmf) \subseteq \mn{dom}$.
When \Sbf is understood, we use $\mn{Mod}(\Omc)$ to denote the set of
all finite relational structures $\Bmf$ over $\Sbf \cup
\mn{sig}(\Omc)$ such that $\Bmf \models\Omc$.
Let $(\Sbf,\Omc,q)$ be an ontology-mediated query with $q$ of arity~$n$.
The \emph{certain answers to $q$ on an \Sbf-instance \Dmf given} \Omc
is the set $\mn{cert}_{q,\Omc}(\Dmf)$ of  tuples
 $\abf \in
\mn{adom}(\Dmf)^n$ such that for all $(\mn{dom},\Dmf') \in
\mn{Mod}(\Omc)$ with $\Dmf \subseteq \Dmf'$ (that is, all models of
\Omc that extend \Dmf), we have $\abf \in q(\Dmf')$.
%
%
%
%

Note that all ontology languages considered in this paper enjoy finite
controllability, meaning that finite relational structures can be replaced
with unrestricted ones without changing the certain answers 
to unions of conjunctive queries
\cite{DBLP:conf/lics/BaranyGO10,DBLP:journals/pvldb/BaranyCO12}. 

Every ontology-mediated query $Q=(\Sbf,\Omc, q)$ can
be semantically interpreted as a query $q_Q$ over $\Sbf$ by setting
$q_Q(\Dmf) = \mn{cert}_{q,\Omc}(\Dmf)$ for all $\Sbf$-instances~\Dmf.
Taking this view 
one step further, each choice of an ontology
language \Lmc and query language \Qmc gives rise to a
 query language,
denoted  $(\Lmc,\Qmc)$, defined as the set of queries
%
$
 q_{(\Sbf,\Omc, q)}$ with \Sbf a schema,  $\Omc$ an $\Lmc$-ontology,
 and $q \in \Qmc$ a  query over $\Sbf\cup\mn{sig}(\Omc)$.
%
We refer to such query languages
$(\Lmc,\Qmc)$ as \emph{ontology-mediated query languages} (or, \emph{OBDA
  languages}). 
%
%
%
\begin{table*}[t]\small
  \centering
  \begin{boxedminipage}{\textwidth}
    \vspace*{-2mm}
\[\begin{array}{lr@{\ }c@{\ }l}
  \forall x ( \  \exists y (\mn{finding}(x,y) \wedge \mn{ErythemaMigrans}(y)) &  \exists \mn{finding} . \mn{ErythemaMigrans} & \sqsubseteq &  \exists \mn{diagnosis}
  . \mn{LymeDisease} \\[.5mm]
  \qquad\qquad\qquad \rightarrow  \exists y (\mn{diagnosis}(x,y) \wedge \mn{LymeDisease}(y))\ )
 \\[1mm]
\forall x (\ (\mn{LymeDisease}(x) \vee \mn{Listeriosis}(x))\rightarrow\mn{BacterialInfection}(x)\ ) &
\mn{LymeDisease} \sqcup \mn{Listeriosis} &\sqsubseteq& \mn{BacterialInfection} \\[1mm]
\forall x(\ \exists y .(\mn{HereditaryDisposition}(y) \wedge \mn{parent}(x,y))\to  \mn{HereditaryDisposition}(y)) \ ) & 
\exists \mn{parent} . \mn{HereditaryDisposition} &\sqsubseteq& 
 \mn{HereditaryDisposition} 
%
\end{array}\]
\end{boxedminipage}
\caption{Example ontology, presented in (the guarded fragment of) first-order logic and the DL \ALC}
\label{tab:ontology}
\end{table*}
\vspace*{-2mm}
\begin{example}
\label{ex:first}
  The left-hand side of Table~\ref{tab:ontology} shows an ontology~\Omc that is formulated in the guarded fragment of FO.
  Consider the ontology-mediated query $(\Sbf,\Omc,q)$ with 
  data schema and query
  $$
  \begin{array}{r@{\ }c@{\ }l}
\Sbf&=&\{
  \mn{ErythemaMigrans}, \mn{LymeDisease}, \\[1mm]
  && \ \ 
\mn{HereditaryPredisposition}, 
\mn{finding}, \mn{diagnosis}, \mn{parent} \}
  \\[1mm]
  q(x) &= &\exists y (\ \mn{diagnosis}(x,y) \wedge \mn{BacterialInfection}(y) \ ).
  \end{array}
  $$
  %
  For the instance \Dmf over \Sbf that consists of the facts
  $$
  \begin{array}{l@{~~~~}l}
  \mn{finding}(\mn{pat1},\mn{jan12find1}) & \mn{ErythemaMigrans}(\mn{jan12find1}) \\[0.5mm]
  \mn{diagnosis}(\mn{pat2},\mn{may7diag2}) &
  \mn{Listeriosis}(\mn{may7diag2})
  \end{array}
  $$
  we have $\mn{cert}_{q,\Omc}(\Dmf) = \{ \mn{pat1}, \mn{pat2} \}$.
\end{example}

\medskip\noindent\textbf{Description Logics for Specifying
  Ontologies.}  In description logic, schemas are generally
restricted to relations of arity one and two, called \emph{concept
  names} and \emph{role names}, respectively. For brevity, we speak of
\emph{binary schemas}. We briefly review 
the basic description
logic \ALC. Relevant extensions of \ALC will be introduced later on in
the paper.

An \emph{\ALC-concept} is formed according to the syntax rule
$$
C,D ::= A \mid \top \mid \bot \mid 
\neg C \mid C \sqcap D \mid C \sqcup D \mid \exists R . C \mid
\forall R .C
$$
where $A$ ranges over concept names and $R$ over role names. An
\emph{\ALC-ontology} \Omc is a finite set of \emph{concept inclusions}
$C \sqsubseteq D$, with $C$ and $D$ \ALC-concepts.  We define the
semantics of \ALC-concepts by translation to FO-formulas with one free
variable, as shown in Table~\ref{fig:ALCsem}.
\begin{table}[t] \small
  \centering
  \begin{boxedminipage}{\columnwidth}
    \vspace*{-2mm}
  $$
    \begin{array}{@{}r@{\;}c@{\;}llr@{\;}c@{\;}l}
    \top^*(x) &=& \top &&   (C \sqcap D)^*(x) &=& C^*(x) \wedge D^*(x)
    \\[1mm]
    \bot^*(x) &=& \bot &&   (C \sqcup D)^*(x) &=& C^*(x) \vee D^*(x) 
    \\[1mm]
    A^*(x) &=& A(x) && 
        (\exists R . C)^*(x) &=& \exists y \, R(x,y) \wedge C^*(y)  
    \\[1mm]
    (\neg C)^*(x) &=& \neg C^*(x) &&
    (\forall R . C)^*(x) &=& \forall y \, R(x,y) \rightarrow C^*(y) 
  \end{array}
  $$
  \end{boxedminipage}
  \caption{First-order translation of \ALC-concepts}
  \label{fig:ALCsem}
\end{table}
An \ALC-ontology \Omc then translates into the set of FO-sentences
$\Omc^* = \{ \forall x . (C^*(x) \rightarrow D^*(x))\mid C \sqsubseteq
D \in \Omc\}$.  On the right-hand side of Table~\ref{tab:ontology}, we
show the \ALC-version of the guarded fragment ontology displayed on
the left-hand side. Note that, although the translation is equivalence-preserving
in this case, in general, the guarded fragment is a more
expressive ontology language than \ALC. Throughout the paper, we do
not explicitly distinguish between a DL ontology and its translation
into FO.

We remark that, from a DL perspective, the above definitions of
instances and certain answers correspond to making the \emph{standard
  name assumption (SNA)} in ABoxes, which in particular implies the
\emph{unique name assumption}. We make the SNA only to facilitate uniform
presentation; the SNA is inessential for the results presented in this
paper.

%
\begin{example}
\label{ex:second}
  Let \Omc and \Sbf be as in Example~\ref{ex:first}. For $q_1(x) =
  \mn{BacterialInfection}(x)$, the ontology-mediated query
  $(\Sbf,\Omc,q_1)$ is equivalent to the union of conjunctive queries $\mn{LymeDisease}(x) \vee
  \mn{Listeriosis}(x)$. For $q_2(x)=\mn{HereditaryDisposition}(x)$, the ontology-mediated
  query $(\Sbf,\Omc,q_2)$ is equivalent to the query defined by the datalog program
  $$
  \begin{array}{rclcrcl}
    P(x) &\leftarrow& \mn{HereditaryDisposition}(x) & &   \mn{goal}(x) &\leftarrow& P(x)
 \\[0.5mm]
    P(x) &\leftarrow& \mn{parent}(y,x) \wedge P(y)
  \end{array}
  $$
  but not to any first-order query.
\end{example}

\section{OBDA and Disjunctive Datalog}
\label{sect:alcUCQ}

We show that for many OBDA languages, there is a natural fragment of
disjunctive datalog with 
exactly the same expressive power.

%
%
A \emph{disjunctive datalog rule} $\rho$ has the form
\[S_1(\xbf_1) \vee \cdots \vee S_m(\xbf_m) \leftarrow R_1(\ybf_1)\land
\cdots\land R_n(\ybf_n)\] with $m\geq 0$ and $n > 0$.
We refer to
$S_1(\xbf_1) \vee \cdots \vee S_m(\xbf_m)$ as the \emph{head} of
$\rho$, and to $R_1(\ybf_1) \wedge \ldots \wedge R_n(\ybf_n)$ as the \emph{body}
of~$\rho$. Every variable that occurs in the head of a rule $\rho$ is
required to also occur in the body of $\rho$. Empty rule heads are
denoted~$\bot$.  A \emph{disjunctive datalog (\DD) program} $\Pi$ is a
finite set of disjunctive datalog rules with a selected \emph{goal
  predicate} \mn{goal} that does not occur in rule bodies and only in
\emph{goal rules} of the form $\mn{goal}(\xbf) \leftarrow R_1(\xbf_1) \wedge \cdots \wedge
R_n(\xbf_n)$.  The \emph{arity of $\Pi$}
is the arity of the 
\mn{goal} relation.  Relation symbols that occur in the head of at
least one rule of $\Pi$ are \emph{intensional (IDB) predicates} of
$\Pi$, and all remaining relation symbols in $\Pi$ are \emph{extensional
  (EDB) predicates}.


Every \DD program $\Pi$ of arity $n$ naturally defines an $n$-ary
query $q_\Pi$ over the schema \Sbf that consists of the EDB predicates
of~$\Pi$: for every instance \Dmf over \Sbf, we have
%
$$
\begin{array}{@{}r@{\ }c@{\ }l@{\;}l}
q_\Pi(\Dmf) &=&
\{\abf\in \adom(\Dmf)^{n}\mid & \mn{goal}(\abf)\in
\Dmf' \\[1mm]
&&&\text{for all } \Dmf'\in \mn{Mod}(\Pi) \text{ with } \Dmf\subseteq \Dmf'\}.
\end{array}
$$ 
Here, $\mn{Mod}(\Pi)$ denotes the set of all instances over $\Sbf'$
that satisfy all rules in $\Pi$, with $\Sbf'$ the set of all IDB and
EDB predicates in $\Pi$. Note that the \DD programs considered in this
paper are negation-free.  Restricted to this fragment, there is no
difference between the different semantics of \DD studied e.g.\ in
\cite{EiterGottlob}.

We use $\mn{adom}(x)$ in rule bodies as a shorthand for ``$x$ is in
the active domain of the EDB predicates''. Specifically, whenever we
use $\mn{adom}$ in a rule of a \DD program $\Pi$, we assume that $\mn{adom}$
is an IDB predicate and that the program $\Pi$ includes all rules of the form 
$\mn{adom}(x) \leftarrow R(\xbf)$ where $R$ is an EDB predicate of
$\Pi$ and $\xbf$ is a tuple of distinct variables that includes $x$.


%

A \emph{monadic disjunctive datalog (\MDD) program} is a \DD program in which all
IDB predicates with the possible exception of \mn{goal} are monadic.
We use \MDD to denote the query language that consists of all queries
defined by an \MDD program.

\subsection{Ontologies Specified in Description Logics}\label{subs:obdaDL}
We show that (\ALC,UCQ) has the same expressive power as \MDD and
identify a fragment of \MDD that has the same expressive power as
(\ALC,AQ). In addition, we consider the extensions of \ALC with
inverse roles, role hierarchies, transitive roles, and the universal
role, which we also relate to \MDD and its fragments. To match the
syntax of \ALC and its extensions, we generally assume schemas to be
binary throughout this section.\footnote{In fact, this assumption is
  inessential for Theorems~\ref{thm:ALCtoMDD} and~\ref{thm:otherDLUCQ}
  (which speak about UCQs), but required for
  Theorems~\ref{thm:ALCMDD1}, \ref{thm:ALCMDD2}, and \ref{thm:ALCMDD3}
  (which speak about AQs) to hold.}

\smallskip\noindent\textbf{($\bm{\mathcal{ALC}}$,UCQ) and \MDD.}
The first main result of this section is Theorem~\ref{thm:ALCtoMDD}
below,
which relates 
(\ALC,UCQ) and \MDD. 

\begin{theorem}\label{thm:ALCtoMDD}
(\ALC,UCQ) and \MDD have the same expressive power.
\end{theorem}

\begin{proof} (sketch) We start with giving some intuitions about
  answering (\ALC,UCQ) queries which guide our translation of
  such queries into \MDD programs.
Recall that the definition of certain answers to an
ontology-mediated query on an instance $\Dmf$ involves a
quantification over all models of $\Omc$ which 
extend $\Dmf$.  It
turns out that in the case of (\ALC,UCQ) queries (and, as we will see
later, more generally for 
(UNFO,UCQ) queries), it suffices
to consider a particular type of extensions of
$\Dmf$ that we term \emph{pointwise extensions}.  Intuitively, such an
extension of $\Dmf$ 
corresponds to 
attaching domain-disjoint structures to the elements of $\Dmf$. 
Formally, for instances $\Dmf\subseteq \Dmf'$, we call 
$\Dmf'$ a pointwise extension of $\Dmf$ if $\Dmf'\setminus\Dmf$ is
the 
 union of 
 instances $\{\Dmf'_a\mid a\in
\mn{adom}(\Dmf)\}$ such that $\mn{adom}(\Dmf'_a)\cap \mn{adom}(\Dmf)
\subseteq \{a\}$ and $\mn{adom}(\Dmf'_a) \cap \mn{adom}(\Dmf'_b)=\emptyset$ for $a \neq b$.  
The fact that we need only consider models of $\Omc$ which are pointwise extensions of $\Dmf$
 is helpful because it
constrains the ways in which a CQ can be satisfied. 
%
%
Specifically, every homomorphism $h$ from $q$ to $\Dmf'$ gives rise to
a query $q'$ obtained from $q$ by identifying all variables that $h$
sends to the same element, and to a decomposition of $q'$ into a
collection of components $q'_0,\dots,q'_k$ where the `core component'
$q'_0$ comprises all atoms of $q'$ whose variables $h$ sends to
elements of \Dmf and for each $\Dmf'_a$ in the image of $h$, there is
a `non-core component' $q'_i$, $1 \leq i \leq k$, such that $q'_i$
comprises all atoms of $q'$ whose variables $h$ sends to elements
of $\Dmf'_a$. Note that the non-core components are pairwise
variable-disjoint and share at most one variable with the core
component.

 We now detail the translation from an ontology-mediated query $(\Sbf,\Omc,q) \in
  \text{(\ALC,UCQ)}$ into an equivalent \MDD program. 
  Let $\mn{sub}(\Omc)$ be the set of
  subconcepts (that is, syntactic subexpressions) of concepts that
  occur in \Omc, and let $\mn{cl}(\Omc,q)$ denote the union of $\mn{sub}(\Omc)$
  and the set of all CQs that have at most one free variable, use only
  symbols from $q$, and whose number of atoms is bounded by the
  number of atoms of $q$.
%
  A \emph{type} (for \Omc and $q$) is a subset of
  $\mn{cl}(\Omc,q)$. The CQs present in $\mn{cl}(\Omc,q)$
  include all potential `non-core components' from the intuitive
  explanation above. The free variable of a CQ in $\mn{cl}(\Omc,q)$
  (if any) represents the overlap between the core component and the
  non-core component.

  We introduce a fresh unary relation symbol $P_\tau$ for every type
  $\tau$, and we denote by $\Sbf'$ the schema that extends $\Sbf$ with
  these additional symbols.  In the \MDD program that we aim to
  construct, the relation symbols $P_\tau$ will be used as IDB
  relations, and the symbols from $\Sbf$
   will be the EBD relations. 
   
   We will say that  a
  relational structure $\Bmf$ over $\Sbf'\cup\mn{sig}(\Omc)$ is
  \emph{type-coherent} if $P_{\tau}(d)\in \Bmf$ just in the case that
  %
  $$
  \begin{array}{rcl}
    \tau&=&\{q' \in \mn{cl}(\Omc,q)\mid q' \text{ Boolean }, \Bmf
    \models q'\} \; \cup \\[1mm]
        &&\{C\in\mn{cl}(\Omc,q)\mid C \text{ unary}, \Bmf \models C[d]\}.
  \end{array}
  $$
  Set $k$ equal to the maximum of $2$ and the width of $q$,
  that is, the number of variables that occur in $q$. 
  By a \emph{diagram}, we mean a conjunction $\delta(x_1, \ldots, x_n)$ of
  atomic formulas over the schema $\Sbf'$, with $n\leq k$ variables.
  A diagram $\delta(\textbf{x})$ is \emph{realizable} if there exists
  a type-coherent $\Bmf\in \mn{Mod}(\Omc)$ that satisfies
  $\exists\textbf{x}\delta(\textbf{x})$.  A diagram
  $\delta(\textbf{x})$ \emph{implies} $q(\textbf{x}')$, with
  $\textbf{x}'$ a sequence of variables from $\textbf{x}$, if every
  type-coherent $\Bmf\in \mn{Mod}(\Omc)$ that satisfies
  $\delta(\textbf{x})$ under some variable assignment, satisfies
  $q(\textbf{x}')$ under the same assignment. 

  The desired \MDD program $\Pi$ consists of the
  following collections of rules:
  $$
  \begin{array}{r@{\;}c@{\;}l@{\!\!\!\!\!\!\!\!\!\!}l}
\displaystyle\bigvee_{\tau\subseteq \mn{cl}(\Omc,q)} \!\!\!\!\!\!
P_\tau(x) &\leftarrow& \mn{adom}(x)
 \\[-2.6mm]
\bot &\leftarrow&\delta(\textbf{x}) & \text{ for all non-realizable
 diagrams $\delta(\textbf{x})$} \\[1mm]
\mn{goal}(\textbf{x}') &\leftarrow& \delta(\textbf{x}) & \text{ for all 
  diagrams $\delta(\textbf{x})$ that imply $q(\textbf{x}')$}
  \end{array}
$$
Intuitively, these rules `guess' a pointwise extension $\Dmf'$ of
\Dmf.  Specifically, the types $P_\tau$ guessed in the first line
determine which subconcepts of \Omc are made true at each element of
$\Dmf'$. Since \MDD does not support existential quantifiers, the
$\Dmf'_a$ parts of $\Dmf'$ cannot be guessed explicitly. Instead, the
CQs included in the guessed types determine those non-core component
queries that matched in the $\Dmf'_a$ parts.  The second line
ensures coherence of the guesses and the last line guarantees that $q$
has the required match in $\Dmf'$.
It is proved in the full version of this paper that the \MDD query
$q_\Pi$ is indeed equivalent to $(\Sbf,\Omc,q)$.

For the converse direction, let $\Pi$ be an \MDD program.  For each
unary IDB relation $A$ of $\Pi$, we introduce two fresh unary
relations, denoted by $A$ and $\bar{A}$. The ontology $\Omc$ enforces
that $\bar{A}$ represents the complement of $A$, that is, it consists
of all inclusions of the form
$$\top \sqsubseteq (A\sqcup \bar{A})\sqcap\neg (A\sqcap \bar{A}).
$$
Let $q$ be the union of (i) all conjunctive queries that constitute
the body of a goal rule, as well as (ii) all conjunctive queries
obtained from a non-goal rule of the form
\[A_1(\xbf_1) \vee \cdots
\vee A_m(\xbf_m) \leftarrow R_1(\ybf_1)\land \cdots\land
R_n(\ybf_n)\] 
by taking the conjunctive query
$$\bar{A}_1(\xbf_1) \land \cdots
\land \bar{A}_m(\xbf_m) \land R_1(\ybf_1)\land \cdots\land
R_n(\ybf_n).$$
It can  be shown that the ontology-mediated query
$(\Sbf,\Omc,q)$, where $\Sbf$ is the schema that consists of 
the EDB relations of $\Pi$, is equivalent to the query defined by $\Pi$. 
\end{proof}

\smallskip\noindent\textbf{$\bm{\mathcal{ALC}}$ with Atomic Queries.}
We characterize ($\mathcal{ALC}$,AQ) by a 
fragment of \MDD.
%
%
This query language has the same expressive power as the OBDA language
($\ALC$,ConQ), where ConQ denotes the set of all \emph{\ALC-concept
  queries}, that is, queries $C(x)$ with $C$ a (possibly compound)
\ALC-concept. Specifically, each query $(\Sbf,\Omc,q) \in
(\mathcal{ALC},\text{ConQ})$ with $q=C(x)$ can be expressed as a query
$(\Sbf,\Omc',A(x)) \in (\mathcal{ALC},\text{AQ})$ where $A$ is a fresh
concept name (that is, it does not occur in $\Sbf \cup
\mn{sig}(\Omc)$) and $\Omc'=\Omc \cup \{ C \sqsubseteq A \}$.
As a consequence, (\ALC,AQ) also has the same expressive power as
($\ALC,\text{TCQ}$), where TCQ is the set of all CQs that take the
form of a directed tree with a single answer variable at the root.

Each disjunctive datalog rule can be associated with an undirected
graph whose nodes are the variables that occur in the rule and whose
edges reflect co-occurrence of two variables in an atom in the rule
body.  We say that a rule is \emph{connected} if its graph is
connected, and that a \DD program is connected if all its rules are
connected.  An \MDD program is \emph{simple} if each rule contains at
most one atom $R(\xbf)$ with $R$ an EDB relation; additionally, we
require that, in this atom, every variable occurs at most once.

%
%
\begin{theorem}\label{thm:ALCMDD1}
($\mathcal{ALC}$,AQ) has the same expressive power as 
unary connected simple \MDD.
\end{theorem}
\begin{proof} (sketch) The translation from (\ALC,AQ) to unary
  connected simple \MDD queries is a modified version of the
  translation given in the proof of Theorem~\ref{thm:ALCtoMDD}. Assume
  that $(\Sbf,\Omc,q)$ with $q=A(x)$ is given. 
  We now take types to be subsets of $\mn{sub}(\Omc)$ and then define
  diagrams exactly as before (with $k=2$).
The \MDD program $\Pi$
  consists of the following rules:
  $$
  \begin{array}{r@{\;}c@{\;}l@{\!\!\!\!\!\!\!\!\!\!}l}
\displaystyle\bigvee_{\tau\subseteq \mn{sub}(\Omc)} \!\!\!\!\!\!
P_\tau(x) &\leftarrow& \mn{adom}(x)
 \\[-2.5mm]
\bot &\leftarrow&\delta(\textbf{x}) & \text{ for all non-realizable
 diagrams $\delta(\textbf{x})$} \\[0.5mm]
  && & \text{ of
  the form $P_{\tau_1}(x) \wedge P_{\tau_2}(x)$, } \\[0.5mm]
    && & \text{ $P_\tau(x) \wedge A(x)$, or } \\[0.5mm]
 && & \text{ $P_{\tau_1}(x_1)\land S(x_1,x_2)
   \land P_{\tau_2}(x_2)$} \\[1mm]
\mn{goal}(x) &\leftarrow& P_\tau(x)  & \quad \text{ for all }P_\tau \text{ with } A \in P_\tau
  \end{array}
$$
%
Clearly, $\Pi$ is unary, connected, and simple. Equivalence of the queries $(\Sbf,\Omc,q)$
and $q_\Pi$ is proved  in the full version of this paper.

Conversely, let $\Pi$ be a unary connected simple \MDD program.  It is
easy to rewrite each rule of $\Pi$ into an equivalent \ALC-concept
inclusion, where $\mn{goal}$ is now regarded as a concept name. For
example, $\mn{goal}(x)\leftarrow R(x,y)$ is rewritten into $\exists
R.\top \sqsubseteq \mn{goal}$ and $P_{1}(x) \vee P_{2}(y) \leftarrow
R(x,y) \wedge A(x) \wedge B(y)$ is rewritten into $A \sqcap \exists
R.(B \sqcap \neg P_{2}) \sqsubseteq P_{1}$.  Let $\Omc$ be the
resulting ontology and let $q=\mn{goal}(x)$. Then the query $q_\Pi$ is
equivalent to the query $(\Sbf,\Omc,q)$, where $\Sbf$ consists of the
EDB relations in $\Pi$.
\end{proof}
Note that the connectedness condition is required since one cannot
express \MDD rules such as $\mn{goal}(x)\leftarrow \adom(x) \wedge
A(y)$ with $y\not=x$ in ($\mathcal{ALC}$,AQ).  Multiple variable
occurrences in EDB relations have to be excluded because programs such
as $\mn{goal}(x) \leftarrow A(x), \ \bot \leftarrow R(x,x)$ (return
all elements in $A$ if the instance contains no reflexive $R$-edge,
and return the active domain otherwise) also cannot be expressed in
($\mathcal{ALC}$,AQ).

\smallskip\noindent\textbf{Extensions of $\bm{\mathcal{ALC}}$.}  We
identify several standard extensions of ($\mathcal{ALC}$,UCQ) and
(\ALC,AQ) that have the same expressive power, and some that do not.
We introduce the relevant extensions only briefly and refer to
\cite{handbook} for more details.

  $\mathcal{ALCI}$ is the extension
of $\mathcal{ALC}$ in which one can state that a role name $R$ is the
\emph{inverse} of a role name $S$, that is, $\forall xy(R(x,y)
\leftrightarrow S(y,x))$; $\mathcal{ALCH}$ is the extension 
in which one can state that a role name $R$ is
\emph{included} in a role name $S$, that is, $\forall xy(R(x,y)
\rightarrow S(x,y))$; $\mathcal{S}$ is the extension 
 of $\mathcal{ALC}$ 
in which one can require some roles names to be
interpreted as \emph{transitive relations}; $\mathcal{ALCF}$ is the
extension 
in which one can state that some role names are interpreted as
\emph{partial functions}; and $\mathcal{ALCU}$ is the extension with
the \emph{universal role} $U$, interpreted as $\dom\times \dom$ in any
relational structure \Bmf with domain \dom.  Note that $U$ should be
regarded as a logical symbol and is not a member of any schema.
All these means of expressivity are included in the OWL2 DL
profile of the W3C-standardized ontology language OWL2 \cite{owl}. 

We use the usual naming scheme to denote combinations
of these extensions, for example $\mathcal{ALCHI}$ for the union of
$\mathcal{ALCH}$ and $\mathcal{ALCI}$ and $\mathcal{SHI}$ for the
union of $\mathcal{S}$ and $\mathcal{ALCHI}$.
The following result summarizes the expressive power of extensions of $\mathcal{ALC}$.
%
\begin{theorem}
\label{thm:otherDLUCQ}
~\\[-5mm]
\begin{enumerate}

\item ($\mathcal{ALCHIU}$,UCQ) has the same expressive power as \MDD and as
($\mathcal{ALC}$,UCQ).

\vspace*{-2mm}

\item ($\mathcal{S}$,UCQ) and ($\mathcal{ALCF}$,UCQ) are strictly more expressive than
($\mathcal{ALC}$,UCQ).

\vspace*{1mm}

\end{enumerate}
\end{theorem}
%
\begin{proof} (sketch) In Point~1, we start with ($\mathcal{ALCIU}$,UCQ),
  for which the result follows from Theorem~\ref{thm:UNFO} in Section \ref{sect:FOOntologies}
  since $\mathcal{ALCIU}$ is a fragment of UNFO.  Role inclusions $\forall
  xy(R(x,y) \rightarrow S(x,y))$ do not add expressive power since they
  can be simulated by adding to the ontology the inclusions $\exists
  R.C \sqsubseteq \exists S.C$ for all $C\in \mn{sub}(\Omc)$, and 
  replacing every atom $S(x,y)$ in the UCQ by $R(x,y)\vee S(x,y)$.

  For Point~2, we separate ($\mathcal{S}$,UCQ) from
  ($\mathcal{ALC}$,UCQ) by showing that the following
  ontology-mediated query $(\Sbf_{1},\Omc_{1},q_{1})$ cannot be
  expressed in ($\mathcal{ALC}$,UCQ): $\Sbf_{1}$ consists of two role
  names $R$ and $S$, $\Omc_{1}$ states that these role names are both
  transitive, and $q_{1}=\exists xy(R(x,y)\wedge S(x,y))$. For
  ($\mathcal{ALCF}$,UCQ), we show that $(\Sbf_{2},\Omc_{2},q_{2})$
  cannot be expressed in ($\mathcal{ALC}$,UCQ), where $\Sbf_2$
  consists of role name $R$ and concept name $A$, $\Omc_2$ states that
  $R$ is functional, and $q_{2}=A(x)$. Detailed proofs are provided in
  the full version of this paper. They rely on a characterization of
  ($\mathcal{ALC}$,UCQ) in terms of colored forbidden patterns
  \cite{DBLP:journals/siamcomp/MadelaineS07}, which is a by-product of
  the connection between ($\mathcal{ALC}$,UCQ) and MMSNP that will be
  established in Section~\ref{sect:obdammsnp}.
  %
\end{proof}
The next 
result is interesting when contrasted with Point~2 of
Theorem~\ref{thm:otherDLUCQ}: when (\ALC,UCQ) is replaced with
(\ALC,AQ), then the addition of transitive roles no longer
increases the expressive power. 
\begin{theorem}\label{thm:ALCMDD2}
($\mathcal{ALC}$,AQ) has the same expressive power as ($\mathcal{SHI}$,AQ).
\end{theorem}
\begin{proof} (sketch)
The proof of Theorem~\ref{thm:ALCMDD1} given above actually shows that
unary connected simple \MDD is at least as expressive as ($\mathcal{ALCI}$,AQ). Thus,
($\mathcal{ALC}$,AQ) has the same expressive power as ($\mathcal{ALCI}$,AQ).
Now it is folklore that in $\mathcal{ALCI}$ transitive
roles can be replaced by certain concept inclusions without changing the certain answers to
atomic queries. This can be done similarly to the elimination of role inclusions in the 
proof above, see \cite{DBLP:phd/de/Motik2006,DBLP:conf/dlog/Simancik12}. 
Thus ($\mathcal{ALCI}$,AQ) has the same expressive power as ($\mathcal{SHI}$,AQ), and the 
result follows.
\end{proof}
It follows from \cite{DBLP:conf/dlog/Simancik12} that this observation
can be extended to all complex role inclusions that are admitted in
the description logic $\mathcal{SROIQ}$. In contrast, the addition of the universal role on
the side of the OBDA query language extends the expressive power of
($\mathcal{ALC}$,AQ). Namely, it corresponds,
on the \MDD side, to dropping the requirement that rule bodies must be
connected.  For example, the \MDD query $\mn{goal}(x) \leftarrow
\adom(x) \wedge A(y)$ can then be expressed using the ontology
$\Omc=\{\exists U.A \sqsubseteq \mn{goal}\}$ and the AQ
$\mn{goal}(x)$.
%
\begin{theorem}\label{thm:ALCMDD3}
($\mathcal{ALCU}$,AQ) and
($\mathcal{SHIU}$,AQ) both have the same expressive power as
unary simple \MDD.
\end{theorem}
We close this section with a brief remark about \emph{Boolean atomic
  queries} (BAQs), that is, queries of the form $\exists x.A(x)$,
where $A$ is a unary relation symbol.  Such queries will be considered
in Section~\ref{sect:OBDAtoCSP}. 
It is possible to
establish modified versions of Theorems~\ref{thm:ALCMDD1} to
Theorem~\ref{thm:ALCMDD3} above in which AQs are replaced by BAQs and
unary goal predicates by $0$-ary goal-predicate, respectively. 


\subsection{Ontologies Specified in First-Order Logic}
\label{sect:FOOntologies}

Ontologies formulated in description logic are not able to speak about
relation symbols of arity greater than two.\footnote{There are
  actually a few DLs that can handle relations of unrestricted arity,
  such as those presented in \cite{DBLP:conf/pods/CalvaneseGL98}. We
  do not consider such DLs in this paper, but remark that large fragments
  of them can be translated into UNFO.} To
overcome this restriction, we consider the guarded fragment of
first-order logic and the unary-negation fragment of first-order logic
\cite{DBLP:conf/lics/BaranyGO10,UNFO}.  Both generalize the
description logic $\mathcal{ALC}$ in different ways.  We also consider
their natural common generalization, the guarded negation fragment of
first-order logic \cite{DBLP:journals/pvldb/BaranyCO12}.  Our results
from the previous subsection turn out to generalize to all these
fragments.
%
%
%
%
%
We start by considering the unary negation fragment.

The \emph{unary-negation fragment of first-order logic} (UNFO)
\cite{UNFO} is the fragment of first-order logic that consists of
those formulas that are generated from atomic formulas, including
equality, using conjunction, disjunction, existential quantification,
and \emph{unary negation}, that is, negation applied to a formula with
at most one free variable. Thus, for example, $\neg\exists xy R(x,y)$
belongs to UNFO, whereas $\exists xy \neg R(x,y)$ does not.  It is
easy to show that every \ALC-TBox is equivalent to a UNFO sentence.
%

\begin{theorem}
\label{thm:UNFO}
(UNFO,UCQ) has the same expressive power as \MDD.
\end{theorem}
\begin{proof}  (sketch)
%
%
  The translation from \MDD to (UNFO,UCQ) is given by
  Theorem~\ref{thm:ALCtoMDD}.  Here, we provide the translation from
  (UNFO,UCQ) to \MDD.  Let $Q=(\Sbf,\Omc,q) \in \text{(UNFO,UCQ)}$ be
  given.  We assume that $\Omc$ is a single UNFO sentence that is in
  the normal form generated by the following grammar:
  \[ \varphi(x) ::= \top ~\mid~ \neg\varphi(x) ~\mid~
  \exists\textbf{y}(\psi_1(x,\textbf{y})\land\cdots\land\psi_n(x,\textbf{y}))\]
  where each $\psi_i$ is either a relational atom or a formula with at
  most one free variable generated by the same grammar, and the free variables
  in $\psi_i$ are among $x,\ybf$.  Note that no
  equality is used and that all generated formulas have at most one
  free variable. Easy syntactic manipulations show that every
  UNFO-formula with at most one free variable is equivalent to a
  disjunction of formulas generated by the above grammar.  In the case
  of $\Omc$, we may furthermore assume that it is a \emph{single} such
  sentence, rather than a disjunction, because
  $\mn{cert}_{q,\Omc_1\lor\Omc_2}(\Dmf)$ is the intersection of
  $\mn{cert}_{q,\Omc_1}(\Dmf)$ and $\mn{cert}_{q,\Omc_2}(\Dmf)$, and
  \MDD is closed under taking intersections of queries.

  Let $\mn{sub}(\Omc)$ be the set of all subformulas
  of $\Omc$ with at most one free variable $z$ (we apply a one-to-one
  renaming of variables as needed to ensure that each formula in
  $\mn{sub}(\Omc)$ with a free variable has the same free
  variable~$z$). 
  Let $k$ be the maximum of the
  number of variables in $\Omc$ and the number of variables in $q$. We
  denote by $\mn{cl}_k(\Omc)$ the set of all formulas $\varphi(x)$ of the
  form
  \[ \exists \textbf{y}(\psi_1(x,
  \textbf{y})\land\cdots\land\psi_n(x,\textbf{y}))\] with
  $\textbf{y}=(y_1, \ldots, y_m)$, $m\leq k$, where each $\psi_i$ is
  either a relational atom that uses a symbol from $q$ or is of the
  form $\chi(x)$ or $\chi(y_i)$, for $\chi(z)\in \mn{sub}(\Omc)$. Note
  that, as in the proof of Theorem~\ref{thm:ALCtoMDD},
  $\mn{cl}_k(\Omc)$ contains all CQs that use only symbols from $q$
  and whose size is bounded by the size of $q$. A \emph{type} $\tau$
  is a subset of $\mn{cl}_k(\Omc)$; the set of all types is denoted
  $\mn{type}(\Omc)$.


   We introduce a fresh unary relation symbol $P_\tau$ for each type
   $\tau$, and we denote by $\Sbf'$ the schema that extends $\Sbf$
   with these additional relations.
   As before, we call a structure $\Bmf$ over
   $\Sbf'\cup\mn{sig}(\Omc)$ type-coherent if for all types
   $\tau$ and elements $d$ in the domain of $\Bmf$, we have
   $P_{\tau}(d)\in \Bmf$ just in the case that $\tau$ is the 
   (unique) type realized at $d$ in $\Bmf$. 
  Diagrams, realizability, and ``implying $q$'' are defined as in the proof of
  Theorem~\ref{thm:ALCtoMDD}.
   It follows
   from \cite{UNFO} that it is decidable whether
   a diagram implies a query, and whether a diagram is realizable.
   The \MDD program $\Pi$ is defined as in the proof of
   Theorem~\ref{thm:ALCtoMDD}, except that now in the first rule, 
   $\tau$ ranges over types in $\mn{type}(\Omc)$.
%
   In the full version of this paper, we prove that the resulting \MDD
   query $q_\Pi$ 
 is equivalent
   to $Q$.  
\end{proof}
Next, we consider the \emph{guarded fragment of first-order logic}
(GF). It comprises all formulas built up from atomic formulas using
the Boolean connectives and guarded quantification of the form
$\exists\textbf{x}(\alpha\land\varphi)$ and
$\forall\textbf{x}(\alpha\to\varphi)$, where, in both cases, $\alpha$ is
an atomic formula (a ``guard'') that contains all free variables of
$\varphi$.  To simplify the presentation of the results, we  consider
here the equality-free version of the guarded fragment.
We do allow one special
case of equality, namely the use of trivial equalities of the form
$x=x$ as guards, which is equivalent to allowing unguarded quantifiers
applied to formulas with at most one free variable.  This restricted
form of equality  is sufficient to translate every \ALC TBox into an
equivalent sentence of GF. 
%

It turns out that the OBDA language (GF, UCQ) is strictly more expressive than \MDD. 
\begin{proposition}
\label{prop:gfucqmdd}
  The Boolean query 
\begin{itemize}

\item[($\dagger$)] there are $a_1,\dots,a_n,b$, for some $n \geq 2$, such that
  $A(a_1)$, $B(a_n)$, and $P(a_i,b,a_{i+1})$ for all $1 \leq i < n$

\end{itemize}
is definable in (GF,UCQ)
  and not in \MDD.
\end{proposition}
\begin{proof}
 Let $\Sbf$ consist of unary predicates $A,B$ and a ternary predicate $P$,
 and let $Q$ be the $\Sbf$-query defined by $(\dagger)$. 
 It is easy to check that $Q$ can be expressed by the (GF,UCQ) query
 $q_{\Sbf, \Omc,\exists x U(x)}$ where 
 $$
 \begin{array}{rcl}
    \Omc &=& \forall xyz \; (P(x,z,y) \rightarrow ( A(x) \rightarrow
    R(z,x) )) \; \wedge \\[1mm]
    && \forall xyz \; (P(x,z,y) \rightarrow ( R(z,x) \rightarrow
    R(z,y) )) \; \wedge \\[1mm]
    && \forall xyz \; (R(x,y) \rightarrow ( B(y) \rightarrow
    U(y) ))
 \end{array}
 $$
 %
We show in the full version of this paper that $Q$ is not expressible in
 \MDD\ using the colored forbidden patterns characterization 
mentioned in the proof sketch of Theorem~\ref{thm:otherDLUCQ}.
 \end{proof}

As fragments of first-order logic, the unary-negation fragment and the
guarded fragment are incomparable in expressive power.  They have a
common generalization, which is known as the guarded-negation fragment
(GNFO) \cite{GNFO}. This fragment is defined in the same way as UNFO,
except that, besides unary negation, we allow \emph{guarded negation}
of the form $\alpha\land\neg\varphi$, where the guard $\alpha$ is an
atomic formula that contains all the variables of $\varphi$.  Again, for
simplicity, we
consider here
the equality-free version of the language, except that we allow the
use
of trivial equalities of the form $x=x$ as guards.
As we will see,
for the purpose of OBDA, GNFO is no more powerful than
GF. Specifically, (GF, UCQ) and (GNFO, UCQ) are expressively
equivalent to a natural generalization of \MDD, namely \emph{\FGDD.}
Recall that a datalog rule is \emph{guarded} if its body includes an
atom that contains all variables which occur in the rule
\cite{DBLP:journals/tocl/GottlobGV02}. A weaker notion of guardedness,
which we  call here \emph{frontier-guardedness}, inspired by
\cite{baget,DBLP:journals/pvldb/BaranyCO12}, requires that,
for each atom $\alpha$ in the head of the rule, there is an atom
$\beta$ in the rule body such that all variables that occur in
$\alpha$ occur also in $\beta$.  We define a \FGDD
query to be a query defined by a \DD program in which every rule is
frontier-guarded.  Observe that \FGDD subsumes \MDD.

\begin{theorem}
  \label{GFUCQfrontier}
  (GF,UCQ) and (GNFO,UCQ) have the same expressive power as \FGDD.
\end{theorem}

Theorem~\ref{GFUCQfrontier} is proved in the full version of this paper via translations
from (GNFO,UCQ) to \FGDD and back that are along the same lines as the
translations from (UNFO,UCQ) to \MDD and back. In addition, we use a
result from \cite{GNFO} to obtain a translation from (GNFO,UCQ) to
(GF,UCQ).

\section{OBDA and MMSNP}
\label{sect:obdammsnp}

We show that \MDD captures coMMSNP and thus, by the results obtained
in the previous section, the same is true for many OBDA languages
based on UCQs. We then use this connection to transfer results from
MMSNP to OBDA languages with UCQs, linking the data complexity of
these languages to the Feder-Vardi conjecture and establishing
decidability of query containment.  We also propose GMSNP, an
extension of MMSNP inspired by frontier guarded \DD, and show
that (GF,UCQ) and (GNFO,UCQ) capture coGMSNP, and that GMSNP has
the same expressive power as a previously proposed extension of MMSNP
called MMSNP$_2$.


An \emph{MMSNP formula} over schema $\Sbf$ has the form
$
\exists X_1 \cdots \exists X_n \forall x_1 \cdots \forall x_m \vp
$
with $X_1,\dots,X_n$ monadic second-order (SO) variables,
$x_1,\dots,x_m$ FO-variables, and $\vp$ a conjunction of
quantifier-free formulas of the form
$$
\psi = \alpha_1 \wedge \cdots
\wedge \alpha_n \rightarrow \beta_1 \vee \cdots \vee \beta_m \mbox{ with $n,m \geq 0$},
$$
where each $\alpha_i$ is of the form $X_i(\xbf)$,
$R(\xbf)$ (with $R \in \Sbf$), or $x=y$, and each $\beta_i$ is of the
form $X_i(\xbf)$. In order to use MMSNP as a query language, and in
contrast to the standard definition, we admit free FO-variables and
speak of \emph{sentences} to refer to MMSNP formulas without free
variables. To connect with the query languages studied thus far, we
are interested in queries obtained by the complements of MMSNP
formulas: each  MMSNP formula $\Phi$ over schema \Sbf with $n$ free
variables gives rise to a query
$$
q_{\Phi,\Sbf}(\Dmf) = \{ \abf \in \mn{adom}(\Dmf)^n \mid (\mn{adom}(\Dmf),\Dmf) \not\models \Phi[\abf]\}
$$
where we set $(\mn{adom}(\Dmf),\Dmf) \models \Phi$ to true when \Dmf
is the empty instance (that is, $\mn{adom}(\Dmf)=\emptyset$) and
$\Phi$ is a sentence.  We observe that the resulting query
language \emph{coMMSNP} has the same expressive power as \MDD. 

%
%
\begin{proposition}
  \label{prop:mmsnpmdd}
  coMMSNP and \MDD have the same expressive power.
\end{proposition}
\begin{proof}
  Let $\Phi=\exists X_1
  \cdots \exists X_n \forall x_1 \cdots \forall x_m \vp$ be an MMSNP formula with free variables $y_1,\dots,y_k$, 
  and let $q_{\Phi,\Sbf} \in \text{coMMSNP}$ be the corresponding query. We can assume w.l.o.g.\
  that all implications $\psi = \alpha_1 \wedge \cdots \wedge \alpha_n
  \rightarrow \beta_1 \vee \cdots \vee \beta_m$ in $\vp$ satisfy the
  following properties: (i)~$n > 0$ and, (ii)~each variable that occurs
  in a $\beta_i$ atom also occurs in an $\alpha_i$ atom.  In fact, we
  can achieve both (i) and (ii) by replacing violating implications
  $\psi$ with the set of implications $\psi'$ that can be obtained
  from $\psi$ by adding, for each variable $x$ that occurs only in the
  head of $\psi$, an atom $S(\xbf)$ where $S$ is a predicate that
  occurs in $\Phi$ and $\xbf$ is a tuple of variables that contains
  $x$ once and otherwise only fresh variables that do not occur in
  $\Phi$. 
  Define an \MDD program $\Pi_\Phi$ that consists of
  all implications in $\vp$ whose head is not $\bot$ plus a rule
  $$\mn{goal}(y_1,\dots,y_k)\leftarrow\vartheta\land\adom(y_1)\land\cdots\land\adom(y_k)$$ for each
  implication $\vartheta \rightarrow \bot$ in $\vp$. It can be proved
  that $q_{\Phi,\Sbf}=q_{\Pi_\Phi,\Sbf}$ for all schemas \Sbf. Finally, it
  is straightforward to remove the equalities from the rule bodies in
  $\Pi_\Phi$. 


  Conversely, let $\Pi$ be a $k$-ary \MDD program and assume w.l.o.g.\
  that each rule uses a disjoint set of variables. Reserve fresh
  variables $y_1,\dots,y_k$ as free variables for the desired MMSNP
  formula, and let $X_1,\dots,X_n$ be the IDB predicates in $\Pi$ and
  $x_1,\dots,x_m$ the FO-variables in $\Pi$ that do not occur in the
  goal predicate.  Set $\Phi_\Pi = \exists X_1 \cdots \exists X_n
  \forall x_1 \cdots \forall x_m \vp$ where $\vp$ is the conjunction
  of all non-goal rules in $\Pi$ plus the implication $\vartheta'
  \rightarrow \bot$ for each rule $\mn{goal}(\xbf)\leftarrow \vartheta$ in $\Pi$. Here, $\vartheta'$ is obtained from
  $\vartheta$ by replacing each variable $x \in \xbf$ whose left-most
  occurrence in the rule head is in the $i$-th position with $y_i$,
  and then conjunctively adding $y_i=y_j$ whenever the $i$-th and
  $j$-th position in the rule head have the same variable. It can be
  proved that $q_{\Pi,\Sbf}=q_{\Phi_\Pi,\Sbf}$ for all schemas \Sbf.
\end{proof}

Thus, the characterizations of OBDA languages in terms of \MDD
provided in Section~\ref{sect:alcUCQ} also establish the descriptive
complexity of these languages by identifying them with (the complement
of) MMSNP. Furthermore, Proposition~\ref{prop:mmsnpmdd} allow us to
transfer results from MMSNP to OBDA. We start by considering the
data complexity of the query evaluation problem: for a query $q$, the
\emph{evaluation problem} is to decide, given an instance \Dmf and a
tuple $\abf$ of elements from \Dmf, whether $\abf \in q(\Dmf)$. Our
first result is that the Feder-Vardi dichotomy conjecture for CSPs is
true if and only if there is a dichotomy between \PTime and \coNP for
query evaluation in (\ALC,UCQ), and the same is true for several other
OBDA languages.  For brevity, we say that a query language \emph{has a
  dichotomy between \PTime and} \coNP, referring only implicitly to
the evaluation problem.

%
%
%
The proof of the following theorem relies on
Proposition~\ref{prop:mmsnpmdd} and
Theorems~\ref{thm:ALCtoMDD},~\ref{thm:otherDLUCQ}, and~\ref{thm:UNFO}.
It also exploits the fact that the Feder-Vardi dichotomy conjecture
can equivalently be stated for MMSNP sentences
\cite{FederVardi,kun-derand}.  Some technical development is needed to
deal with the presence of free variables. Details are in the full version of this paper.
%
%
\begin{theorem}\label{dich-ucq}
  (\ALC,UCQ) has a dichotomy between \PTime and \coNP iff the
  Feder-Vardi conjecture holds. The same is true for
  ($\mathcal{ALCHIU}$,UCQ) and (UNFO,UCQ).
\end{theorem}
Recall that (\ALCF,UCQ) and (\Smc,UCQ) are two extensions of
(\ALC,UCQ) that were identified in Section~\ref{sect:alcUCQ} to be
more expressive than (\ALC,UCQ) itself.  It was already proved in
\cite{KR12-csp} (Theorem~27) that, compared to ontology-mediated
queries based on \ALC, the functional roles of \ALCF dramatically
increase the computational power. This is true even for atomic queries.
%
%
\begin{theorem}[\cite{KR12-csp}]\label{thm:ALCF1}  
  For every \NP-Turing machine $M$, there is a query $q$ in (\ALCF,AQ)
  such that the complement of the word problem of $M$ has the same
  complexity as evaluating $q$, up to polynomial-time reductions.
  Consequently, (\ALCF,AQ) does not have a dichotomy between \PTime
  and \coNP (unless \PTime = \NP).
\end{theorem}
We leave it as an open problem to analyze the computational power of (\Smc,UCQ).

\smallskip

There are other interesting results that can be transferred from MMSNP
to OBDA. Here, we consider query containment.  Specifically, the
following general containment problem was proposed in \cite{KR12-cont}
as a powerful tool for OBDA: given ontology-mediated queries
$(\Sbf,\Omc_i,q_i)$, $i \in \{1,2\}$, decide whether for all
\Sbf-instances \Dmf, we have $\mn{cert}_{q_1,\Omc_1}(\Dmf) \subseteq
\mn{cert}_{q_2,\Omc_2}(\Dmf)$.\footnote{\label{footnote:containment}In fact, this definition is
  slightly different from the one used in \cite{KR12-cont}. There,
  containment is defined only over instances \Dmf that are consistent
  w.r.t.\ $\Omc_1$ and $\Omc_2$, i.e., where there is at least one
  finite \Sbf-structure $(\mn{dom},\Dmf')$ such that $\Dmf \subseteq
  \Dmf'$ and $\Dmf' \in \mn{Mod}(\Omc_i)$.}  Applications include the
optimization of ontology-mediated queries and managing the effects on
query answering of replacing an ontology with a new, updated
version. In terms of OBDA languages such as (\ALC,UCQ), the above
problem corresponds to query containment in the standard sense: an
\Sbf-query $q_1$ is \emph{contained in} an \Sbf-query $q_2$, written
$q_1 \subseteq q_2$, if for every \Sbf-instance~\Dmf, we have
$q_1(\Dmf) \subseteq q_2(\Dmf)$. Note that there are also less general
(and computationally simpler) notions of query containment in OBDA
that do not fix the data schema \cite{DBLP:conf/pods/CalvaneseGL98}.

It was proved in \cite{FederVardi} that containment of MMSNP sentences is decidable. 
We thus obtain the following result for OBDA languages. 
\begin{theorem}
\label{thm:containment1}
Query containment is decidable for 
the OBDA languages (\ALC,UCQ), ($\mathcal{ALCHIU}$,UCQ), and (UNFO,UCQ).
\end{theorem}
Note that this result is considerably stronger than those in
\cite{KR12-cont}, which considered only containment of
ontology-mediated queries $(\Sbf,\Omc,q)$ with $q$ an atomic
query since already this basic case turned out to be technically
intricate. The treatment of CQs and UCQs was left open, including all cases
stated in Theorem~\ref{thm:containment1}.

\smallskip

We now consider OBDA languages based on the guarded fragment and GNFO.
By Proposition~\ref{prop:gfucqmdd}, (GF,UCQ) and (GNFO,UCQ) are
strictly more expressive than \MDD and we cannot use
Proposition~\ref{prop:mmsnpmdd} to relate these query languages to the
Feder-Vardi conjecture. Theorem~\ref{GFUCQfrontier} suggests that it
would be useful to have a generalization of MMSNP that is equivalent
to frontier-guarded \DD. Such a generalization is introduced next.

A formula of 
\emph{guarded monotone strict NP (abbreviated GMSNP)}
has the form $ \exists X_1 \cdots \exists X_n \forall x_1 \cdots
\forall x_m \vp $ with $X_1,\dots,X_n$ SO variables of any arity,
$x_1,\dots,x_n$ FO-variables, and $\vp$ a conjunction of formulas
$$
\psi = \alpha_1 \wedge \cdots \wedge \alpha_n \rightarrow \beta_1
\vee \cdots \vee \beta_m \mbox{ with $n,m \geq 0$},
$$
where each $\alpha_i$ is of the form $X_i(\xbf)$, $R(\xbf)$ (with $R
\in \Sbf$), or $x=y$, and each $\beta_i$ is of the form
$X_i(\xbf)$. Additionally, we require that for every head atom
$\beta_i$, there is a body atom $\alpha_j$ such that $\alpha_j$
contains all variables from $\beta_i$.  GMSNP gives rise to a query
language coGMSNP in analogy with the definition of coMMSNP.  It can be
shown by a straightforward syntactic transformation that every MMSNP
formula is equivalent to some GMSNP formula.  Together with 
Proposition~\ref{prop:gfucqmdd} and Theorem~\ref{GFUCQfrontier}, this yields the second statement
of the following lemma; the first statement can be proved similarly to Proposition \ref{prop:mmsnpmdd}. 

\begin{theorem}
  \label{lem:gmsnpfgdd}
  coGMSNP has the same expressive power as frontier-guarded \DD and
  is  strictly more expressive than coMMSNP.
\end{theorem}
%
%
Although defined in a different way, GMSNP is essentially the same
logic as MMSNP$_2$, which is studied in
\cite{Madeleine}. Specifically, MMSNP$_2$ is the extension of MMSNP in
which monadic SO-variables range over sets of domain elements
\emph{and facts}, and where atoms of the form $X(R(\xbf))$ are allowed
in place of atoms $X(x)$ with $X$ an SO-variable and $R$ from the
data schema \Sbf. Additionally, a guardedness condition is imposed,
requiring that whenever an atom $X(R(\xbf))$ occurs in a rule head,
then the atom $R(\xbf)$ must occur in the rule body.  Formally, the
SO-variables $X_{i}$ are interpreted in an instance $\mathfrak{D}$ as
sets $\pi(X_{i})\subseteq \adom(\Dmf)\cup \Dmf$ and
$\Dmf\models_{\pi}X(R(x_1, \ldots, x_n))$ 
if $R(\pi(x_{1}),\dots,\pi(x_{n}))\in \pi(X)$. We observe the following.
\begin{proposition}
\label{prop:gmmm2}
  GMSNP and MMSNP$_2$ have the same expressive power.
\end{proposition}
Details for the proofs of both Theorem~\ref{lem:gmsnpfgdd} and
Lemma~\ref{prop:gmmm2} are in the full version of this paper. In \cite{Madeleine},
it was left as an open question whether MMSNP$_2$ is more expressive
than MMSNP, which is resolved by the results above. 

\smallskip We leave it as an interesting open question whether
Theorem~\ref{dich-ucq} can be extended to (GF,UCQ) and (GNFO,UCQ),
that is, whether GMSNP (equivalently: MMSNP$_2$) has a dichotomy
between \PTime and \NP if the Feder-Vardi conjecture holds. While this
question is implicit already in \cite{Madeleine}, the results
established in this paper underline its significance from a different
perspective.

\section{OBDA and CSP}
\label{sect:OBDAtoCSP}%
We show that OBDA languages based on AQs capture CSPs (and
generalizations thereof), and we transfer results from CSPs to OBDA
languages. In comparison to the previous section, we obtain a richer set
of results, and often even worst-case optimal decision procedures.
Recall that each finite relational structure $\Bmf$ over a schema
$\Sbf$
gives rise to a \emph{constraint satisfaction
  problem} 
which is to decide, given a finite relational structure~$\Amf$ over
\Sbf, whether there is a homomorphism from $\Amf$ to $\Bmf$ (written
$\Amf \rightarrow \Bmf$). In this context, the relational structure
$\Bmf$ is also called the \emph{template} of the CSP.

\smallskip

CSPs give rise to a query language coCSP in the spirit of the query
language coMMSNP introduced in the previous section.  In its basic version,
this language is Boolean and turns out to have exactly the same
expressive power as ($\ALC$,BAQ), where BAQ is the class of
\emph{Boolean} atomic queries. To also cover non-Boolean AQs,
we consider two natural generalizations of  CSPs. First, a
\emph{generalized CSP} is defined by a finite \emph{set \Fmc of templates}, 
rather than only a single one
\cite{DBLP:journals/ejc/FoniokNT08}. The problem then consists in
deciding, given an input structure \Amf, whether there is a template
$\Bmf \in \Fmc$ such that $\Amf \rightarrow \Bmf$.
%
Second, in a \emph{(generalized) CSP with constant symbols}, both the
template(s) and the input structure are endowed with constant symbols
\cite{DBLP:journals/tcs/FederMS04,DBLP:journals/tods/AlexeCKT11}. To
be more precise, let $\Sbf$ be a schema and $\cbf= c_{1}, \ldots,
c_{m}$ a finite sequence of distinct constant symbols. A \emph{finite
  relational structure over} $\Sbf\cup \cbf$ has the form
$(\Amf,d_{1},\ldots,d_{m})$ with \Amf a finite relational structure
over $\Amf$ that, in addition, interprets the constant symbols $c_{i}$
by elements $d_{i}$ of the domain $\dom$ of $\Amf$, for $1\leq i \leq
m$.
Let $(\Amf,\abf)$ and $(\Bmf,\bbf)$ be finite relational structures
over $\Sbf\cup \cbf$. A mapping $h$ is a \emph{homomorphism} from
$(\Amf,\abf)$ to $(\Bmf,\bbf)$, written $(\Amf,\abf) \rightarrow
(\Bmf,\bbf)$, if it is a homomorphism from $\Amf$ to $\Bmf$ and
$h(a_{i})=b_{i}$ for $1\leq i \leq m$.
A (generalized) CSP with constant symbols is then defined like a
(generalized) CSP, based on this extended notion of homomorphism.

We now introduce the query languages obtained from the different versions
of CSPs, where generalized CSPs with constant symbols constitute the
most general case. 
Specifically, each finite set of templates $\Fmc$ over $\Sbf\cup\cbf$
with $\cbf=c_1,\ldots,c_m$ gives rise to an $m$-ary query
coCSP$(\Fmc)$ that maps every $\Sbf$-instance $\Dmf$ to
$$
\{ \dbf \in \adom(\Dmf)^{m}\!\mid\! \forall (\Bmf,\bbf)\!\in \mathcal{F}:\! 
(\Dmf,\dbf) \not\rightarrow (\Bmf,\bbf)\!\},
$$
where we view $(\Dmf,\dbf)$ as a finite relational structure whose
domain is $\mn{adom}(\Dmf)$.  The query language that consists of all
such queries is called \emph{generalized coCSP with constant symbols}.
The fragment of this query language that is obtained by admitting only
sets of templates $\mathcal{F}$ without constant symbols is called
\emph{generalized coCSP}, and the fragment induced by singleton sets
$\mathcal{F}$ without constant symbols is called \emph{coCSP}.
%
\begin{example}
  Selecting an illustrative fragment of Examples~\ref{ex:first}
  and~\ref{ex:second}, let 
  $$
  \begin{array}{r@{\ }c@{\ }l}
  \Omc &=& \{ \exists \mn{parent} . \mn{HereditaryDisposition} \sqsubseteq
   \mn{HereditaryDisposition} \} \\[1mm]
  \Sbf &=& \{
  \mn{HereditaryDisposition},\mn{parent} \}
  \end{array}
  $$
  Moreover, let $q_2(x)=\mn{HereditaryDisposition}(x)$ be the query from
  Example~\ref{ex:second}. To identify a query in coCSP with constant
  symbols that is equivalent to the ontology-mediated query
  $(\Sbf,\Omc,q_2)$, let \Bmc be the following template:
  %
  
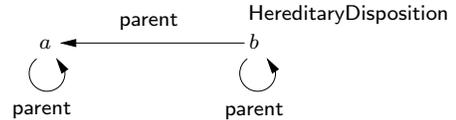
\begin{center}
\begin{picture}(0,0)%
\epsfig{file=template.pdftex}%
\end{picture}%
\setlength{\unitlength}{2901sp}%
\begingroup\makeatletter\ifx\SetFigFont\undefined%
\gdef\SetFigFont#1#2#3#4#5{%
  \reset@font\fontsize{#1}{#2pt}%
  \fontfamily{#3}\fontseries{#4}\fontshape{#5}%
  \selectfont}%
\fi\endgroup%
\begin{picture}(3433,1024)(4133,-3344)
\put(5048,-2528){\makebox(0,0)[lb]{\smash{{\SetFigFont{8}{9.6}{\familydefault}{\mddefault}{\updefault}{\color[rgb]{0,0,0}$\mn{parent}$}%
}}}}
\put(4366,-2716){\makebox(0,0)[lb]{\smash{{\SetFigFont{8}{9.6}{\familydefault}{\mddefault}{\updefault}{\color[rgb]{0,0,0}$a$}%
}}}}
\put(6159,-2716){\makebox(0,0)[lb]{\smash{{\SetFigFont{8}{9.6}{\familydefault}{\mddefault}{\updefault}{\color[rgb]{0,0,0}$b$}%
}}}}
\put(6151,-2476){\makebox(0,0)[lb]{\smash{{\SetFigFont{8}{9.6}{\familydefault}{\mddefault}{\updefault}{\color[rgb]{0,0,0}$\mn{HereditaryDisposition}$}%
}}}}
\put(4133,-3279){\makebox(0,0)[lb]{\smash{{\SetFigFont{8}{9.6}{\familydefault}{\mddefault}{\updefault}{\color[rgb]{0,0,0}$\mn{parent}$}%
}}}}
\put(5949,-3286){\makebox(0,0)[lb]{\smash{{\SetFigFont{8}{9.6}{\familydefault}{\mddefault}{\updefault}{\color[rgb]{0,0,0}$\mn{parent}$}%
}}}}
\end{picture}%

  \end{center}
  %
  It can be shown that for all instances \Dmf over \Sbf and for all $d
  \in \mn{adom}(\Dmf)$, we have $d\in \mn{cert}_{q_2,\Omc}(\Dmf)$ iff
  $(\Dmf,d) \not\rightarrow (\Bmf,a)$ and thus the query coCSP$(\Bmf)$ is
  as required.
\end{example}
The following theorem summarizes the connections between OBDA
languages with (Boolean) atomic queries, \MDD, and CSPs.  Note that we
consider binary schemas only.
%
%
%
\begin{theorem}\label{ALCUtoCSP} 
  The following are lists of query languages that have the same expressive power:
\begin{enumerate}

\vspace*{-1mm}

\item $(\mathcal{ALCU}$,AQ), ($\mathcal{SHIU}$,AQ), unary simple \MDD, and 
generalized coCSP with one constant symbol;

\vspace*{-1mm}

\item ($\mathcal{ALC}$,AQ), ($\mathcal{SHI}$,AQ), unary connected simple \MDD, and
generalized coCSPs with one constant symbol such that all templates are identical except 
for the interpretation of the constant symbol;

\vspace*{-1mm}
\item ($\mathcal{ALCU}$,BAQ), ($\mathcal{SHIU}$,BAQ), Boolean simple \MDD, and generalized 
coCSP;

\vspace*{-1mm}
\item ($\mathcal{ALC}$,BAQ), ($\mathcal{SHI}$,BAQ), Boolean connected simple \MDD, and coCSP.

\vspace*{-1mm}

\end{enumerate}
Moreover, given the ontology-mediated query or monadic datalog program,
the correponding CSP template is of at most exponential size and can be 
constructed in time polynomial in the size of the template.
\end{theorem}
\begin{proof} The equivalences between OBDA languages and fragments of
  \MDD have been proved in Section~\ref{sect:alcUCQ}.  We give a proof
  of the remaining claim of Point~1, namely that $(\mathcal{ALCU}$,AQ)
  and generalized coCSP with one constant symbol are equally
  expressive.  We extend the notation used in the proof of
  Theorem~\ref{thm:ALCtoMDD}.  For simplicity, throughout this proof
  we regard $\forall R.C$ as an abbreviation for $\neg \exists R.\neg
  C$.

  Let $Q=(\Sbf,\Omc,A(x))$ be an ontology-mediated query formulated in
  $(\mathcal{ALCU}$,AQ).  
 A \emph{type for \Omc} is a set $\tau\subseteq
  \mn{sub}(\Omc)$ and $\mn{tp}(\Omc)$ denotes the set of all types for
  \Omc. We say that $\tau \in \mn{tp}(\Omc)$ is \emph{realizable} if
  there is an $\Amf=(\mn{dom},\Dmf) \in \mn{Mod}(\Omc)$ and a $d \in
  \mn{dom}$ such that $C\in \tau$ iff $\Amf\models C^{\ast}[d]$ for
  all $C\in \mn{sub}(\Omc)$.  A set of types $T \subseteq
  \mn{tp}(\Omc)$ is \emph{realizable in a $Q$-countermodel} if there
  is an $\Amf\in \mn{Mod}(\Omc)$ that realizes exactly the types in
  $T$ and such that $A\not\in \tau$ for at least one $\tau\in T$.

  Let $\mathcal{C}$ be the set of all $T \subseteq \mn{tp}(\Omc)$ that
  are realizable in a $Q$-countermodel and maximal with this
  property. Note that the number of elements of $\mathcal{C}$ is
  bounded by the size of $\Omc$ since for any two distinct
  $T_{1},T_{2}\in \mathcal{C}$, there must be a concept $\exists
  U.D\in \mn{sub}(\Omc)$ such that $\exists U.D\in \tau$ for all
  $\tau\in T_{1}$ and $\exists U.D\not\in \tau$ for all $\tau\in
  T_{2}$ or vice versa; otherwise, we can take the disjoint union of
  any structures $\Amf_1,\Amf_2$ which show that $T_1,T_2$ are
  realizable in a $Q$-countermodel to obtain $Q$-countermodel that
  realizes $T_1 \cup T_2$. For $R\in \Sbf$, we call a pair
  $(\tau_{1},\tau_{2})$ of types \emph{$R$-coherent} if $\exists R.C \in \tau_{1}$
  for every $\exists R.C \in \mn{sub}(\Omc)$ such that  $C\in \tau_{2}$.

With each $T\in \mathcal{C}$, we associate the \emph{canonical \Sbf-structure $\Bmf_{T}$}
with domain $T$ and the following facts:
\begin{itemize}
\item $B(\tau)$ for all $\tau\in T$ and $B\in \Sbf$ such that $B\in \tau$;
\item $R(\tau_{1},\tau_{2})$  for all $\tau_{1},\tau_{2}\in T$ and $R\in \Sbf$ such that 
$(\tau_{1},\tau_{2})$ is $R$-coherent. 
\end{itemize}
Note that the construction of $\Bmf_{T}$ is well-known from the
literature on modal and description logic. For example, $\Bmf_T$ can
be viewed as a finite fragment of a canonical model of a modal logic
that is constructed from maximal consistent sets of formulas
\cite{modalbook}. Alternatively, $\Bmf_T$ can be viewed as the result of a type
elimination procedure \cite{pratt}.
%
%

We obtain the desired set \Fmc of CSP templates by setting
$$
\mathcal{F}= \{ (\Bmf_{T},\tau) \mid T\in \mathcal{C}, \tau\in T,A\not\in \tau\}.
$$
One can show that for every $\Sbf$-instance $\Dmf$ and $d\in
\adom(\Dmf)$, there exists $(\Bmf_{T},\tau)\in \mathcal{F}$ with
$(\Dmf,d) \rightarrow (\Bmf_{T},\tau)$ iff $d\not\in
q_{\Sbf,\Omc,A(x)}(\Dmf)$. Thus, the ontology-mediated query $Q$ is
equivalent to the query defined by $\mathcal{F}$.

\medskip

Conversely, assume that $\mathcal{F}$ is a finite set of $\Sbf$-structures with one constant.
Take some $(\Bmf,b)\in \mathcal{F}$, and for every $d$ in the domain $\dom(\Bmf)$ of $\Bmf$, 
create some fresh concept name $A_{d}$. 
Let $A$ be another fresh concept name, and set
\begin{eqnarray*}
\Omc_{\Bmf,b} & = & \{ A_{d} \sqsubseteq \neg A_{d'} \mid d\not=d'\}\cup\\
     &   & \{ A_{d} \sqcap \exists R.A_{d'} \sqsubseteq \bot \mid R(d,d')\not\in \Bmf,R\in \Sbf\}\cup\\
     &   & \{ A_{d} \sqcap B \sqsubseteq \bot \mid B(d)\not\in \Bmf,B \in \Sbf\}\cup\\
     &   & \{ \top \sqsubseteq \midsqcup_{d\in \dom(\Bmf)}A_{d}, \,\, \neg A_{b} \sqsubseteq A\}
\end{eqnarray*}
Consider the ontology-mediated query
$Q_{\Bmf,b}=(\Sbf,\Omc_{\Bmf,b},A(x)).$ One can show that for every
$\Sbf$-instance $\Dmf$ and $d\in \adom(\Dmf)$, $(\Dmf,d)\rightarrow
(\Bmf,b)$ iff $d\not\in q_{Q_{\Bmf,b}}(\Dmf)$. Thus, $Q_{\Bmf,b}$ is
the desired query if \Fmc is a singleton. For the general case, let
$\Omc$ be the disjunction over all $\Omc_{\Bmf,b}$ with $(\Bmf,b)\in
\Fmc$. Note that $\Omc$ can be expressed in $\mathcal{ALCU}$: first,
rewrite each $\Omc_{\Bmf,b}$ into a single inclusion of the form $\top
\sqsubseteq C_{\Bmf,b}$ and then set
$$
\Omc = \{ \top \sqsubseteq \midsqcup_{(\Bmf,b)\in \mathcal{F}}\forall U.C_{\Bmf,b}\}.
$$
Using the above observation about the
queries $Q_{\Bmf,b}$, it is not hard to show that 
the ($\mathcal{ALCU}$,AQ)-query $Q=(\Sbf,\Omc,A(x))$
is equivalent to
the query coCSP($\Fmc$). 

This completes the proof of Point~1. The proofs of Points~2 to 4 are
similar and given in the full version of this paper.
\end{proof}
%
%
%
%
Theorem~\ref{ALCUtoCSP} allows us to transfer results from the CSP
world to OBDA, which, in light of recent progress on CSPs,
turns out to be very fruitful. We start with data complexity.
\begin{theorem}
\label{thm:CSPFV}
 ($\mathcal{ALC}$,BAQ) has a dichotomy
  between \PTime and \coNP iff the Feder-Vardi conjecture holds.  The
  same is true for ($\mathcal{SHIU}$,AQ), and ($\mathcal{SHIU}$,BAQ).
\end{theorem}
Since $\mathcal{SHIU}$-ontologies can be replaced by
$\mathcal{ALCU}$-ontologies in ontology-mediated queries due to
Theorem~\ref{thm:ALCMDD3}, the ``if'' direction of (all cases
mentioned in) Theorem~\ref{thm:CSPFV} actually follows from
Theorem~\ref{dich-ucq}. The ``only if'' direction is a consequence of
Theorem~\ref{ALCUtoCSP}.  We now consider further interesting
applications of Theorem~\ref{ALCUtoCSP}, in particular to deciding
query containment, FO-rewritability, and datalog rewritability.

\subsection{Query Containment}

In Section~\ref{sect:obdammsnp}, we have established decidability
results for query containment in OBDA languages based on UCQs.  For
OBDA languages based on AQs and BAQs, we even obtain a tight
complexity bound.
It is easy to see that query containment in coCSP is characterized by
homomorphisms between templates. Consequently, it is straightforward
to show that query containment for generalized coCSP with constant
symbols is \NP-complete.  Thus, Theorem~\ref{ALCUtoCSP} yields the
following \NExpTime upper bound for query containment in OBDA
languages. The corresponding lower bound is proved in the full version of this paper by
a non-trivial reduction of a \NExpTime-complete tiling problem.
\begin{theorem}\label{ContNEXP}
Query containment in ($\mathcal{SHIU}$,AQ$\cup$BQ) is in \NExpTime.
It is \NExpTime-hard already for ($\mathcal{ALC}$,AQ) and for
$(\mathcal{ALC}$,BAQ).
\end{theorem}
It is a consequence of a result in \cite{KR12-cont} that query
containment is undecidable for $\mathcal{ALCF}$. We show in the
full version of this paper how the slight gap pointed out in
Footnote~\ref{footnote:containment} can be bridged.



%
\subsection{FO- and Datalog-Rewritability}

One prominent approach to answering ontology-mediated queries is to
make use of existing relational database systems or datalog engines,
eliminating the ontology by query rewriting
\cite{Romans,DBLP:conf/dlog/EiterOSTX12,Bernardo-IJCAI13}. Specifically,
an ontology-mediated query $(\Sbf,\Omc,q)$ is \emph{FO-rewritable} if
there exists an FO-query over $\Sbf$ that is equivalent to it and
\emph{datalog-rewritable} if there exists a datalog program over \Sbf
that defines it. We observe that every ontology-mediated query that is
FO-rewritable is also datalog-rewritable.
%
%
\begin{proposition}
  \label{prof:FOtoDLog} If $Q = (\Sbf,\Omc,q)$ is an ontology-mediated
  query with $\Omc$ formulated in equality-free FO and $q$ a UCQ, then
  $q_Q$ is preserved by homomorphisms. Consequently, it follows from
  \cite{DBLP:journals/jacm/Rossman08} that if $q_Q$ is FO-rewritable,
  then $q_Q$ is rewritable into a UCQ (thus into datalog).
\end{proposition}
Example~\ref{ex:second} illustrates that ontology-mediated
queries are not always rewritable into an FO-query, and the same holds
for datalog-rewritability. It is 
a central problem to decide,
given an ontology-mediated query, whether it is FO-rewritable and
whether it is datalog-rewritable. 
By leveraging the CSP connection, 
we show that both problems are decidable and
pinpoint their complexities.


On the CSP side, FO-rewritability corresponds to FO-definability, and
datalog-rewritability to datalog-definability. Specifically, an $\Sbf$-query
coCSP$(\Fmc)$ is \emph{FO-definable} if there is an FO-sentence
$\varphi$ over $\Sbf$ such that for all finite relational structures
$\Amf$ over \Sbf, we have $ \Amf \models \varphi$ iff $\Amf
\not\rightarrow \Bmf $ for all \Bmf in~\Fmc.
Similarly, coCSP$(\Fmc)$ is \emph{datalog-definable} if there exists a
datalog program $\Pi$ that defines it.  
%
FO-definability and datalog-definability have been studied extensively
for CSPs, culminating in the following results.
\begin{theorem}\label{thm:BartoKozik}
  Deciding, for a given finite relational structure $\Bmf$ without constant symbols, whether coCSP($\Bmf$) is FO-definable is \NP-complete
\cite{DBLP:journals/lmcs/LaroseLT07}. The same is true
for datalog-definability
\cite{FreeseEtAl}.%
\footnote{An \NP algorithm for datalog-definability is implicit in
  \cite{FreeseEtAl}, based on results from \cite{BartoKozik09}, see
  also \cite{BulatovWidth}.  We thank Benoit Larose and Liber Barto
  for pointing this out.
}
\vspace*{-3mm}
\end{theorem}
Combining the preceding theorem with Theorem~\ref{ALCUtoCSP}, we 
obtain \NExpTime upper bounds
for deciding FO-rewritability and datalog-rewritability of queries
from $(\mathcal{SHI}$,BAQ). 

To capture the more important AQs rather than only BAQs,
we show that Theorem~\ref{thm:BartoKozik} can be lifted, in a natural
way, to generalized CSPs with constant symbols. The central step is
provided by Proposition~\ref{prop:eliminating-constants} below.  For
each finite relational structure $\Bmf$ with constant symbols $c_1,
\ldots, c_n$, let us denote by $\Bmf^c$ the corresponding relational
structure without constant symbols over the schema that contains
additional unary relations $P_1, \ldots, P_n$, where each $P_i$
denotes the singleton set that consists of the element denoted by
$c_i$.
\begin{proposition}\label{prop:eliminating-constants}
  For every set of homomorphically incomparable structures $\Bmf_1,
  \ldots, \Bmf_n$ with constant symbols, 
\begin{enumerate}

\vspace*{-1mm}
\item coCSP($\Bmf_1, \ldots, \Bmf_n$) is FO-definable iff
  coCSP($\Bmf_i^c$) is FO-definable for  $1 \leq i\leq n$.

\vspace*{-2mm}

\item coCSP($\Bmf_1, \ldots, \Bmf_n$) is datalog-definable iff
  coCSP($\Bmf_i^c$) is datalog-definable for $1 \leq i\leq n$.

\vspace*{-1mm}
\end{enumerate}
\end{proposition}
A proof of Proposition~\ref{prop:eliminating-constants} is provided in
the full version of this paper. It relies on the characterization of FO-definable CSPs
as those CSPs that have \emph{finite obstruction sets}; this
characterization was given in \cite{Atserias2005} for structures
without constant symbols and follows from results in
\cite{DBLP:journals/jacm/Rossman08} for the case of structures with
constant symbols.

Note that every set of structures $\Bmf_1, \ldots, \Bmf_n$ has a
subset $\Bmf'_1, \ldots, \Bmf'_m$ which consists of homomorphically
incomparable structures such that coCSP($\Bmf_1, \ldots, \Bmf_n$) is
equivalent to coCSP($\Bmf'_1, \ldots, \Bmf'_m$). We use this 
observation to establish the announced lifting of Theorem~\ref{thm:BartoKozik}.
%
\begin{theorem}\label{thm:definability}
  FO-definability and datalog-definability of generalized CSP with
  constant symbols is \NP-complete.
\end{theorem}
\begin{proof}
  To decide whether a generalized CSP with constant symbols given as a
  set of templates $\mathcal{F}=\{\Bmf_1, \ldots, \Bmf_n\}$ is
  FO-definable, it suffices to first guess a subset
  $\mathcal{F}'\subseteq\mathcal{F}$ and then to verify that
  (i)~coCSP($\Bmf^c$) is FO-definable for each $\Bmf\in\mathcal{F}'$,
  and (ii)~for each $\Bmf\in\mathcal{F}$ there is a
  $\Bmf'\in\mathcal{F}'$ such that $\Bmf\to\Bmf'$. By
  Theorem~\ref{thm:BartoKozik}, this can be done in \NP. Correctness
  follows from Proposition~\ref{prop:eliminating-constants} and the
  fact that whenever there is a subset $\Fmc'$ satisfying (i) and
  (ii), then by the observation above there must be a subset $\Fmc''
  \subseteq \Fmc'$ of homomorphically incomparable structures such
  that coCSP($\Fmc''$) is equivalent to coCSP($\Fmc'$), which by~(ii)
  is equivalent to coCSP($\Fmc$).  Datalog-definability can be decided
  analogously.
\end{proof}
From Theorems~\ref{ALCUtoCSP} and \ref{thm:definability}, we obtain a
\NExpTime upper bound for deciding FO-rewritability and
datalog-rewritability of ontology-mediated queries based on DLs and (B)AQs.
The
corresponding lower bounds are proved in the full version of this paper using a
reduction from a \NExpTime-hard tiling problem (in fact, the same
problem as in the lower bound for query containment).
\begin{theorem}\label{FONExp}
It is in \NExpTime to decide FO-rewritability and datalog-rewritability 
of queries in ($\mathcal{SHIU}$,AQ$\cup$BAQ). Both problems are \NExpTime-hard
for ($\mathcal{ALC}$,AQ) and ($\mathcal{ALC}$, BAQ).
\end{theorem}   
Modulo a minor difference in the treatment of instances that are not
consistent (see Footnote~\ref{footnote:containment}), it follows
from a result in \cite{KR12-csp} that FO-rewritability is undecidable for
($\mathcal{ALCF}$,AQ). In the full version of this paper, we show how to bridge the
difference and how to modify the proof so that the result also applies
to datalog-rewritability.
\begin{theorem}
  FO-rewritability and datalog-rewritability are undecidable
  for ($\mathcal{ALCF}$,AQ) and ($\mathcal{ALCF}$,BAQ).
\end{theorem}






\section{Conclusion}

Another query language frequently used in OBDA with description logics
is conjunctive queries. The results in this paper imply that there is
a dichotomy between \PTime and \coNP for (\ALC,CQ) if and only if the
Feder-Vardi conjecture holds. We leave it open whether there is a
natural characterization of (\ALC,CQ) in terms of disjunctive
datalog. 

We mention two natural lines of future research. First, it would
be interesting to understand the data complexity and query containment
problem for (GF,UCQ) and (GNFO,UCQ). In particular, we would like to
know whether Theorems~\ref{dich-ucq} and~\ref{thm:containment1} extend
to (GF,UCQ) and (GNFO,UCQ).  As explained in
Section~\ref{sect:obdammsnp}, resolving this question for
Theorem~\ref{dich-ucq} is equivalent to clarifying the computational
status of GMSNP and MMSNP$_2$.

Another interesting topic for future work is to analyze
FO-rewritability and datalog-rewritability of ontology-mediated
queries based on UCQs (instead of AQs) as a decision problem. It
follows from our results that this is equivalent to deciding
FO-definability and datalog-definability of MMSNP formulas
(or even GMSNP formulas).


 \medskip
 \noindent {\bf Acknowledgements}.  We thank Benoit Larose
 and Liber Barto for discussions on datalog-definability of
 CSPs, and Florent Madeleine and Manuel Bodirsky for 
 discussions on MMSNP.

\smallskip

 Meghyn Bienvenu was supported by the ANR project PAGODA (ANR-12-JS02-007-01).
 Balder ten Cate was supported by NSF Grants IIS-0905276 and IIS-1217869.
 Carsten Lutz was supported by the DFG SFB/TR 8 ``Spatial Cognition''.





\appendix

\newsavebox{\stupid}
\savebox{\stupid}{\large\bf\ref{sect:alcUCQ}}
\section{Proofs for Section \usebox{\stupid}}


\subsection{Proofs for Section~\ref{subs:obdaDL}}

\medskip We remark that the direction ``from (\ALC,AQ) to \MDD'' of
Theorem~\ref{thm:ALCtoMDD} is actually a consequence of
Theorem~\ref{thm:UNFO}, which makes a strictly more general statement.
We still provide it here (and in the main paper) as a warmup for the
proof of Theorem~\ref{thm:UNFO}. As an extra bit of notation, we say
that an assignment $\pi$ of elements of an instance \Dmf to the
variables of a CQ $q$ is a \emph{match of $q$ in \Dmf} if \Dmf satisfies
$q$ under $\pi$.

\medskip
\noindent
{\bf Theorem~\ref{thm:ALCtoMDD}.}
\emph{
(\ALC,UCQ) and \MDD have the same expressive power.
}

\smallskip
\noindent
\begin{proof} (continued) We establish here the correctness of the
  translation from (\ALC,UCQ) to \MDD. Let $m$ be the arity of
  $(\Sbf,\Omc,q)$. We have to show the following.
  \\[2mm]
  {\bf Claim.} For all instances \Dmf over \Sbf and all $\abf \in
  \mn{adom}(\Dmf)^m$, we have $\abf \in \mn{cert}_{q,\Omc}(\Dmf)$ iff
  $\abf \in q_\Pi(\Dmf)$.
  \\[2mm]
  ``if''. Assume that $\abf \notin \mn{cert}_{q,\Omc}(\Dmf)$. Then
  there is a $(\mn{dom}',\Dmf') \in \mn{Mod}(\Omc)$ such that $\Dmf
  \subseteq \Dmf'$ and $\abf \notin q(\Dmf')$.  For each $b \in
  \mn{adom}(\Dmf)$, let $\mu(b)$ be the unique type realized at $b$ in
  $\Dmf'$, that is, 
\begin{align*}
\mu(b)= & \, \{q' \in \mn{cl}(\Omc,q)\mid q' \text{ is Boolean and } \Dmf' \models q'\} \, \cup \\
& \, \{C\in\mn{cl}(\Omc,q)\mid C \text{  is unary and } \Dmf'\models C[b]\}.
\end{align*}
  Let $\Dmf''$ be the instance that consists of the atoms in
  $\Dmf$ and the atom $P_{\mu(b)}(b)$ for each $b \in
  \mn{adom}(\Dmf)$. It can be verified that $\Dmf''$ is a model of
  $\Pi$. 
  In particular, it follows from the construction of $\Dmf''$ and the fact that
  $\abf \notin q(\Dmf')$ 
  that whenever a diagram $\delta(\xbf)$ has a match $\pi$
  in $\Dmf''$ and $\delta(\xbf)$ implies $q(\xbf')$, then $\pi(\xbf')
  \neq \abf$.  Since $\Dmf''$ is a model of $\Pi$ and 
  $\mn{goal}(\abf) \notin \Dmf''$, 
  we have $\abf\not\in
  q_\Pi(\Dmf)$.


\medskip

``only if''. Assume that $\abf\not\in q_\Pi(\Dmf)$, and let $\Dmf'\in
\mn{Mod}(\Pi)$ be such that $\Dmf\subseteq\Dmf'$ and $\Dmf'$ does not
contain $\mn{goal}(\abf)$. We assume w.l.o.g. that $\mn{adom}(\Dmf)=
\mn{adom}(\Dmf')$. 
Note that the first two rules of $\Pi$ ensure that
for each $a \in \mn{adom}(\Dmf)$, there is a unique type $\mu(a)$ such
that $P_{\mu(a)}(a) \in \Dmf'$.  
The second rule further ensures that for each $a \in \mn{adom}(\Dmf)$, there is
a model $(\mn{dom}_a,\Dmf_a)$ of \Omc in which $\mu(a)$ is realized at
$a$.
We may assume that these models have disjoint domains.  Let 
$(\mn{dom}'',\Dmf'')$ be the
relational structure obtained by first taking
the union of $(\mn{dom}_a,\Dmf_a)_{a \in \mn{adom}(\Dmf)}$, and then
adding all facts from $\Dmf$. To prove that $\abf \notin
\mn{cert}_{q,\Omc}(\Dmf)$, it suffices to show that
\begin{itemize}
\setlength{\itemindent}{1em}
\item[(i)]$(\mn{dom}'',\Dmf'')$ is a model
    of \Omc, and
\item[(ii)] $\abf\not\in q(\Dmf'')$.
\end{itemize}
For Point~(i), let $\mu(d)$ be the unique type realized by $d$ in
$(\mn{dom}_a,\Dmf_a)$, for all $d\in\mn{dom}_a$.  It is
not difficult to show by induction on the structural complexity of~$C$ that for
all concepts $C \in \mn{cl}(\Omc,q) \cap \mn{sub}(\Omc)$ and all $d
\in \mn{dom}''$, we have 
\begin{equation}
(\mn{dom}'',\Dmf'') \models C(d) \quad \text{ iff } C
\in \mu(d)\label{thm1-eq}
\end{equation}(refer to the proof of Theorem \ref{thm:ALCMDD1} for details).  
Since $\mn{cl}(\Omc,q)$ by definition includes $C$ and
$D$ whenever $C \sqsubseteq D$ is in \Omc, this implies Point~(i) as
desired.

It thus remains to establish Point~(ii). Assume to the contrary that
there is a disjunct $q'(\xbf')$ of $q$ such that $\abf \in
q'(\Dmf'')$, that is, there is a match $\pi$ of $q'(\xbf')$ in
$\Dmf''$ such that $\pi(\xbf')=\abf$. 
We define a diagram
$\delta(\xbf)$ based on the restriction of the original model $\Dmf'$
of $\Pi$ 
, as follows:
$\delta(\xbf)$ contains (a) all atoms $A(x)$ such that $\pi(x) \in
\mn{adom}(\Dmf')$ and $A(\pi(x)) \in \Dmf'$ (where $A$ can be either a concept name or
of the form $P_\tau$), (b) all atoms $R(x,y)$ such that $\pi(x),\pi(y)
\in \mn{adom}(\Dmf')$ and $R(\pi(x),\pi(y)) \in \Dmf'$, and (c)
all atoms $P_{\mu(d)}(z_d)$ (with $z_d$ a fresh variable) such that 
$P_{\mu(d)}(d) \in \Dmf'$ and there is some $\pi(w) \in \mn{dom}_d$.
Atoms of type (c) are used to handle the case in which 
a Boolean subquery $q''$ of $q'$ is mapped inside $\Dmc_d$, 
but the element $d$ does not itself belong to the image of $\pi$. 
We remark that the mapping $\pi$ can be straightforwardly extended to 
a match for $\delta(\xbf)$ in $\Dmf'$ by setting $\pi(z_d)=d$. 
Since $\delta(\xbf)$ is satisfied in $\Dmf'$ under $\pi$ and $\pi(\xbf')=\abf$,  
by the last rule of $\Pi$, we can obtain the desired contradiction by showing that
$\delta(\xbf)$ implies $q'(\xbf')$.

Thus, let $(\mn{dom},\Bmf) \in \mn{Mod}(\Omc)$ be a 
type-coherent structure, and
let $\tau$ be a match of $\delta(\xbf)$ in \Bmf. 
Consider the following CQs: 
\begin{itemize}
\item $q_0$ is the restriction of $q'$
to those variables that $\pi$ maps to elements of $\Dmf$; 
\item for each
 $a \in \mn{adom}(\Dmf)$ 
 such that some element of $\mn{dom}_a$ is 
in the range of $\pi$,  the CQ $q_a$ is obtained by  first 
taking the restriction of $q'$ to those
variables that $\pi$ maps to elements of $\mn{dom}_a$
and then
identifying all variables that $\pi$ maps to the same element
(preserving the names of free variables).
\end{itemize}
Clearly, each $q_a$ has at
most one free variable, which, if it exists, is mapped to $a$ by $\pi$.  

We start by showing that $q_0$ is satisfied in \Bmf under $\tau$.
For role atoms in $q_0$, this is immediate since all such atoms also belong to $\delta(\xbf)$. 
Thus, consider some concept atom $A(x) \in q_0$. 
Since $A(x) \in q'$ and $\pi$ is a match for 
$q'$ in $\Dmc''$, we have $A(\pi(x)) \in \Dmc''$. 
Then using the fact that $A\in \mn{cl}(\Omc,q) \cap \mn{sub}(\Omc)$ and Equation (1) above,
we obtain $A \in \mu(\pi(x))$. 
We know that $P_\mu(\pi(x))(\pi(x)) \in \Dmc'$, so by construction of $\delta(\xbf)$, 
we must have $P_\mu(\pi(x))(x) \in \delta(\xbf)$, hence 
$P_\mu(\pi(x))(\tau(x)) \in \Bmf$. 
Using the type-coherence of $\Bmf$ and the fact that $A \in \mu(\pi(x))$, 
we obtain $A(\tau(x)) \in \Bmf$, as desired. 

Now consider a query $q_a$. By construction, the length of $q_a$ cannot exceed the length of $q$,
and so $q_a \in \mn{cl}(\Omc,q)$. 
Since $q_a$ has a match in $\Dmf_a$ (such that, if $q_a$ has a free variable, it is
mapped to $a$) and $\Dmf_a$ realizes the type $\mu(a)$ at $a$, we must have $q_a
\in \mu(a)$. 
By construction of $\delta(\xbf)$, there is an atom $P_{\mu(a)}(x)\in \delta(\xbf)$. 
Since $\tau$ is a match for $\delta(\xbf)$ in $\Bmf$, we must have 
$P_{\mu(a)}(\tau(x)) \in \Bmf$.
Then, using the fact that  \Bmf is type-coherent, we can find a match $\tau_a$ of $q_a$ in \Bmf (such
that, if $q_a$ has a free variable, $\tau_a$ maps it to $\tau(x)$). It is not
hard to see that the matches $\tau$ and $\tau_a$ 
can be assembled into a
match $\tau'$ of $q'$ in \Bmf which coincides with $\tau$ on~$\xbf'$. 
\end{proof}

\medskip

\begin{trivlist}
\item \textbf{Theorem~\ref{thm:ALCMDD1}}
{\it ($\mathcal{ALC}$,AQ) has the same expressive power as 
unary connected simple \MDD.}\end{trivlist}

  \begin{proof} (continued)
    We establish here the correctness of the translation from 
    (\ALC,AQ) to \MDD. That is, 
    we show that, for every instance $\Dmf$ and elements
    $a\in \mn{adom}(\Dmf)$, we
    have $a\in \mn{cert}_{q,\Omc}(\Dmf)$ if and only if
    $a\in q_\Pi(\Dmf)$.
  
   \medskip

   ``if''. Assume that  $a\not\in \mn{cert}_{q,\Omc}(\Dmf)$. 
      Then there is $(\mn{dom},\Dmf')\in\mn{Mod}(\Omc)$ with 
      $\Dmf\subseteq \Dmf'$ such that $a\not\in q(\Dmf')$. 
     For
    each $b \in \mn{adom}(\Dmf)$, let $\mu(b)$ be the unique type
    realized at $b$ in $\Dmf'$. Let $\Dmf''$ be the instance that
    consists of the atoms in $\Dmf$ and an atom $P_{\mu(b)}(b)$ for each
    $b \in \mn{adom}(\Dmf)$. It can be checked that $\Dmf''$ is a model
    of $\Pi$. 
    Since $\mn{goal}(a) \notin \Dmf''$, we obtain
    $a\not\in q_\Pi(\Dmf)$. 

    \medskip

    ``only if''. Assume that $a\not\in q_\Pi(\Dmf)$ and let $\Dmf'$ be a model of
    $\Pi$ with $\Dmf\subseteq\Dmf'$ that does not contain
    $\mn{goal}(a)$. For each $b \in \mn{adom}(\Dmf)$,
    let $\mu(b)$ be a type such that $P_{\mu(b)}(b) \in \Dmf'$
    (in fact, the rules in $\Pi$ enforce that there is exactly one such $\mu(b)$). 
    Note that $A\not\in\mu(a)$. Also note that each type $\mu(b)$ must be realizable
    in some model of $\Omc$ (else, there would be a rule forbidding $P_\mu(b)$ atoms).
    Thus, for each $b \in \mn{adom}(\Dmf)$,
    we can find a model $(\mn{dom}_b,\Dmf_b)$ of \Omc 
    in which the type $\mu(b)$ is realized at
    $b$. We may assume that these
    models have disjoint domains.
    Let $(\mn{dom}'',\Dmf'')$ be obtained by first taking the
    union of $(\mn{dom}_b,\Dmf_b)_{b \in \mn{adom}(\Dmf)}$, and then
    adding all facts in $\Dmf$. By construction, $\Dmf \subseteq \Dmf''$ and $a\not\in q(\Dmf'')$.
    It remains to show that $(\mn{dom}'',\Dmf'')$ is a model
    of \Omc.

    Let $\mu(d)$ be the unique type realized by $d$ in
    $(\mn{dom}_a,\Dmf_a)$, for all $d\in\mn{dom}_a$. 
    We show the following by induction on the
    structural complexity of $C$:
    \begin{description}
    \item[($\ast$)]
      For every concept $C \in \mn{sub}(\Omc)$ and every $d \in
      \mn{dom}''$, we have $(\mn{dom}'',\Dmf'') \models C(d)$ iff $C \in
      \mu(d)$.
    \end{description}
    Note that it follows from ($\ast$) that $(\mn{dom}'',\Dmf'')$ is a model
    of \Omc.

    For the base case, first suppose that $A \in \mu(d)$, with $A$ a
    concept name and $d \in \mn{dom}_a$. 
    Then $A(d) \in \Dmf_a \subseteq \Dmf''$, 
    so $(\mn{dom}'',\Dmf'') \models A(d)$. Next suppose that       
    $(\mn{dom}'',\Dmf'') \models A(d)$. Then $A(d) \in \Dmf''$,
    so either $A(d) \in \Dmf_a$, or $d=a$ and $A(d) \in \Dmf$. In the former case,
    we immediately obtain $A \in \mu(d)$. In the latter case, note that if $A \not \in \mu(d)$,
    then $\Pi$ would contain the rule $\bot \leftarrow P_{\mu(d)}(x) \wedge A(x)$, 
    and this would yield a contradiction since $\{A(d), P_{\mu(d)}(d)\}  \subseteq \Dmf'$.
        
  The inductive step for the Boolean operators is trivial, so we
   consider only the case of the $\exists R$ constructor (the
 argument for the $\forall R$ constructor is similar). 
 Thus, let $C = \exists R. D$ and $d\in\mn{dom}_a$, and suppose that 
  $C \in \mu(d)$. Then $(\mn{dom}_a,\Dmf_a) \models
     \exists R. D(d)$, so there exists $e \in \mn{dom}_a$ such that 
     $R(d,e) \in \Dmf_a$ and $(\mn{dom}_a,\Dmf_a) \models
     D(e)$. It follows that $D \in \mu(e)$, and hence by the induction hypothesis,
     we must have $(\mn{dom}'',\Dmf'') \models D(e)$.
   Since $\Dmf_a \subseteq\Dmf''$, we have $R(d,e) \in \Dmf''$,
   which yields $(\mn{dom}'',\Dmf'')\models C(d)$.

    Conversely, suppose $(\mn{dom}'',\Dmf'')$ satisfies
    $\exists R. D(d)$, that is, there is an element $e$ such that
    $(\mn{dom}'',\Dmf'')$ satisfies $R(d,e)$ and $D(e)$.
    If $e\in\mn{dom}_a$, the claim ($\ast$) follows immediately
    from the induction hypothesis. Otherwise, we must
   have that $e\in\mn{adom}(\Dmf)$ and,
    by induction hypothesis, $D\in\mu(e)$. It follows
    that $\exists R.D \in \mu(d)$, because 
    otherwise $P_{\mu(d)}(x) \land R(x,y) \wedge P_{\mu(e)}(y)$ would be a
    non-realizable diagram,
   and $\Pi$ would derive an inconsistency.
  \end{proof}

\medskip

\noindent
{\bf Theorem~\ref{thm:otherDLUCQ}.}
{\it
\begin{enumerate}

\item ($\mathcal{ALCHIU}$,UCQ) has the same expressive power as \MDD and as
($\mathcal{ALC}$,UCQ).

\item ($\mathcal{S}$,UCQ) and ($\mathcal{ALCF}$,UCQ) are strictly more expressive than
($\mathcal{ALC}$,UCQ).

\end{enumerate}}

\medskip

To complete the proof of Theorem \ref{thm:otherDLUCQ}, we need to show that the queries from 
($\mathcal{S}$,UCQ) and ($\mathcal{ALCF}$,UCQ) indicated in the proof sketch cannot be expressed in 
($\mathcal{ALC}$,UCQ), or equivalently, \MDD.  
We start by providing a means of identifying queries which cannot be expressed in \MDD, using the notion of 
colored instances, defined as follows:

\begin{definition}\label{color-def}
Let $\Sbf$ be a schema and $\Cmc$ be a set of unary predicates
(colors) $\{C_{1},\ldots,C_{n}\}$ disjoint from $\Sbf$. 
A \emph{$\Cmc$-colored $\Sbf$-structure} is an $\Sbf\cup
\Cmc$-structure $(\dom,\Dmf)$ such that
\begin{itemize}
\item For every $d \in \dom$, $C_{i}(d) \in \Dmf$ for some $i$;
\item If $C_{i}(d) \in \Dmf$, then $C_{j}(d) \not \in \Dmf$ for every $j \neq i$.
\end{itemize}
$\Dmf$ is called a \emph{$\Cmc$-coloring} of an $\Sbf$-structure
$\Dmf'$ if $\Dmf'$ is the $\Sbf$-reduct of $\Dmf'$. 

Now for each $k>0$, fix $\Cmc_{k}$ with $|\Cmc_{k}|=k$ and $\Cmc_{k} \cap \Sbf = \emptyset$.
Then a \emph{$k$-coloring} of $\Dmf$ is simply a 
$\Cmc_{k}$-coloring of $\Dmf$. 
\end{definition}

We will also utilize the notion of forbidden pattern problems from 
\cite{DBLP:journals/siamcomp/MadelaineS07,DBLP:journals/ejc/KunN08,DBLP:journals/siamdm/BodirskyCF12}, whose definition we recall here.

\begin{definition}\label{forb-def}
Given a set $\Fmc$ of $\Cmc$-colored $\Sbf$-structures (called \emph{forbidden patterns}), 
we define $\forb(\Fmc)$ as the set of all $\Sbf$-structures $\Dmf$
such that there exists a $\Cmc$-coloring $\Dmf'$ of $\Dmf$
for which $\Fmf \not \rightarrow  \Dmf'$ 
for every $\Fmf \in \Fmc$. The \emph{forbidden patterns problem defined by $\Fmc$} 
is to decide whether a given $\Sbf$-structure belongs to $\forb(\Fmc)$. 
\end{definition}

Analogously to coMMSNP, we can define a query language coFPP
consisting of all those Boolean queries $q_{\Fmc,\Sbf}$ defined by
$$
q_{\Fmc,\Sbf}(\Dmf) = 1 \quad \text{iff} \quad (\mn{adom}(\Dmf),\Dmf) \not\in \forb(\Fmc)
$$
with $\Fmc$ a set of $\Cmc$-colored $\Sbf$-structures.
It follows directly from results in \cite{DBLP:journals/siamcomp/MadelaineS07} 
that coMMSNP and coFPP have the same expressive power. 
Combining this result with Proposition \ref{prop:mmsnpmdd} (from Section~\ref{sect:obdammsnp}),
we obtain the following:


\begin{proposition}\label{fpp-mddlog}
coFPP and Boolean \MDD have the same expressive power. 
\end{proposition}

We use Proposition \ref{fpp-mddlog} in the proof of the following lemma, 
whose purpose is to establish a sufficient condition for non-expressibility in \MDD. 

\begin{lemma}\label{lem:crit}
A Boolean query $Q$ over schema $\Sbf$ does not belong to \MDD if for every $m,n >0$, there exist $\Sbf$-instances $\Dmf_0$ and $\Dmf_1$ with $Q(\Dmf_0)=0$ and $Q(\Dmf_1)=1$ such that for every $m$-coloring $\Bmf_0$ of $(\adom(\Dmf_{0}),\Dmf_0)$, there exists an $m$-coloring $\Bmf_1$ of $(\adom(\Dmf_1),\Dmf_{1})$ such that 
from every substructure of $\Bmf_1$ having at most $n$ elements there is a homomorphism to $\Bmf_0$.
\end{lemma}
\begin{proof}
Assume for a contradiction that the conditions of the lemma hold for every $n,m>0$ but that $Q$ is equivalent to some query in \MDD. Then, by Proposition~\ref{fpp-mddlog}, there is a set $\Fmc$ of $\Cmc$-colored 
$\Sbf$-structures such that for all $\Sbf$-instances $\Dmf$, we have $Q(\Dmf)=1$ if and only if $(\adom(\Dmf),\Dmf) \not \in \forb(\Fmc)$. Let $m_{0} = |\Cmc|$, and let $n_{0}$ be the maximal number of elements in the domain of
some $\Fmf\in \mathcal{F}$. 
We can assume w.l.o.g. that $\Cmc=\Cmc_{m_{0}}$.

Take $\Sbf$-instances $\Dmf_0$ and $\Dmf_1$ satisfying the conditions of the lemma
for $m_{0},n_{0}$. 
As $Q(\Dmf_0)=0$, 
there exists a $\Cmc$-coloring $\Bmf_{0}$ of $(\adom(\Dmf_{0}), \Dmf_{0})$
such that $\Fmf \not \rightarrow \Bmf_{0}$ for every $\Fmf \in \Fmc$. 
It follows that there exists a $\Cmc$-coloring $\Bmf_{1}$ of $(\adom(\Dmf_{1}), \Dmf_{1})$
such that from every substructure 
of $\Bmf_1$
with at most $n_{0}$ elements,
there exists a homomorphism 
to $\Bmf_0$.
Since $Q(\Dmf_1)=1$, we know that there must exist 
some $\Fmf \in \Fmc$ such that $\Fmf \rightarrow \Bmf_{1}$. 
As $\Fmf$ contains at most $n_{0}$ elements, we can compose this homomorphism with the previous homomorphism to obtain a homomorphism of $\Fmf$ into $\Bmf_0$, 
contradicting the fact that $(\adom(\Dmf_{0}), \Dmf_{0})\in \forb(\Fmc)$.
\end{proof}


Using the preceding lemma, we can now prove that the queries mentioned in the proof sketch cannot be expressed in \MDD.

\begin{lemma}
There exist queries in ($\mathcal{S}$,UCQ) which do not belong to \MDD.
\end{lemma}
\begin{proof}
Consider $Q=(\Sbf,\Omc,q)$ where $\Sbf=\{R,S\}$, $\Omc$ asserts transitivity of $R$ and $S$, and $q=\exists xy(R(x,y)\wedge S(x,y))$. 

We apply Lemma~\ref{lem:crit}.
Assume that $m,n>0$ are given. Let $k= n-1$ and $k'= m^{k+2} + 1$.
Define $\Dmf_1$ and $\Dmf_0$ as follows:
\begin{itemize}

\item $\Dmf_1$ has elements $e, f$ and $a_{1},\ldots, a_{k}$ and $b_{1},\ldots,b_{k}$
  and the atoms $R(e,a_{1}), R(a_{k},f)$ and $R(a_{i},a_{i+1})$ for $1\leq i < k$, and
  $S(e,a_{1}), S(a_{k},f)$ and $S(b_{i},b_{i+1})$ for $1\leq i < k$.
  
\smallskip

\item $\Dmf_0$ has elements $e^{1}, \ldots, e^{k'}$ and $f^{1},\ldots,f^{k'}$
as well as $a_{1}^{j},\ldots,a_{k}^{j}$ for $1\leq j \leq k'$ and $b_{1}^{i,j},\dots,b_{k}^{i,j}$ for
$1\leq j<i \leq k'$. The atoms of  $\Dmf_0$ consist of:
\begin{itemize}
\item  $R(e^{i},a_{1}^{i}), R(a_{k}^{i},f^{i})$, and $R(a_{j}^{i},a_{j+1}^{i})$ for $1\leq i \leq k'$
and $1\leq j < k$;
\item $S(e_{i},b_{1}^{i,j})$ and $S(b_{k}^{i,j},f_{j})$ for $1\leq j<i \leq k'$,  
and 
\newline
$S(b_{l}^{i,j},b_{l+1}^{i,j})$ for $1\leq l < k$ and $1\leq j < i \leq k'$.
\end{itemize}
\end{itemize}
It is readily checked that $Q(\Dmf_0)=0$ and $Q(\Dmf_1)=1$, as required.
Let $\Bmf_0$ be an $m$-coloring of $(\adom(\Dmf_0),\Dmf_{0})$.
Since 
$k'= m^{k+2} + 1$, we can find $i,i'$ with $i>i'$ such that
the colorings of $e^{i}, a_{1}^{i},\ldots,a_{k}^{i},f^{i}$ and 
$e^{i'},a_{1}^{i'},\ldots,a_{k}^{i'},f^{i'}$ coincide. 
Define an $m$-coloring of $(\adom(\Dmf_1),\Dmf_{1})$ by taking the
coloring of $e^{i}, a_{1}^{i},\ldots,a_{k}^{i},f^{i}$ for $e,a_{1},\ldots,a_{k},f$ 
and the coloring of $b_{1}^{i,i'},\ldots,b_{k}^{i,i'}$ for $b_{1},\ldots,b_{k}$.
Denote by $\Bmf_{1}$ the resulting colored structure.

Consider a subset $C$ of $\adom(\Bmf_1)$
having at most $n$ elements, and let $\Bmf_1'$
be 
the restriction of $\Bmf_1$ to the elements in $C$. 
We define a function $h$ from $C$ to $\adom(\Bmf_0)$ as follows:
\begin{itemize}
\item If $e\not\in C$, then let $h$ be the restriction of the following mapping to $C$:
$h(a_{l})=a_{l}^{i'}$, $h(b_{l})= b_{l}^{i,i'}$ and $h(f)=f^{i'}$;
\item If $f\not\in C$, then let $h$ be the restriction of the following mapping to $C$:
$h(a_{l})=a_{l}^{i}$, $h(b_{l})= b_{l}^{i,i'}$ and $h(e)=e^{i}$; 
\item Otherwise there exists $a_{i_{0}}
\not\in C$. Then
let $h$ be the restriction of the following mapping to $C$:
$h(e) = e^{i}$, $h(a_{l})=a_{l}^{i}$ for all $l<i_{0}$,
$h(a_{l})=a_{l}^{i'}$ for all $l>i_{0}$,
$h(b_{l})= b_{l}^{i,i'}$ for all $1 \leq l \leq k$, and $h(f)=f^{i'}$.
\end{itemize}
It is easily verified that $h$ is a homomorphism from 
$\Bmf_1'$ to $\Bmf_0$.
\end{proof}

\begin{lemma}
There exist queries in ($\mathcal{ALCF}$,UCQ) which do not belong to \MDD. 
\end{lemma}
\begin{proof}
Consider $Q=(\Sbf,\Omc,\exists x. A(x))$ where $\Sbf=\{S,A\}$ and $\Omc$ states that $S$ is functional. Set $\Dmf_1=\{S(a,b), S(a,c)\}$ and $\Dmf_0=\{S(a,b)\}$.  Note that $q_{Q}(\Dmf_1)=1$ (since no model of $\Omc$ contains $\Dmf_1$) and $q_{Q}(\Dmf_0)=0$. 
Let $\Bmf_0$ be any $m$-coloring of 
$(\adom(\Dmf_0),\Dmf_{0})$. We define an $m$-coloring $\Bmc_{1}$ of $\Dmf_1$ by assigning $a,b$ the same colors as in $\Bmf_{0}$ and giving $c$ the same color as $b$. 
Then the mapping sending $a$ to itself and $b,c$ to $b$ defines a homomorphism from $\Bmf_1$ to $\Bmf_0$
(and hence also defines a homomorphism from any substructure of $\Bmf_1$ to $\Bmf_0$).
It follows by Lemma \ref{lem:crit} that $Q$ is not definable in \MDD. 
\end{proof}

\medskip

\noindent
%
%
%
%
%
{\bf Theorem~\ref{thm:ALCMDD3}}
\emph{($\mathcal{ALCU}$,AQ) and ($\mathcal{SHIU}$,AQ) both have the
same expressive power as unary simple \MDD.}

\medskip

\noindent
\begin{proof}
We first show
\begin{itemize}
\item ($\mathcal{ALCU}$,AQ) is at least as expressive as unary simple \MDD;
\item unary simple \MDD is at least as expressive as ($\mathcal{ALCIU}$,AQ).
\end{itemize}
For Point~1, let $\Pi$ be a unary simple \MDD program. The rewriting of
each rule of $\Pi$ into an equivalent $\mathcal{ALCU}$-concept
inclusion is similar to the proof of Theorem~\ref{thm:ALCMDD1} except that now
one also has to concider non-connected bodies.
They can be translated using the universal role. For example,  
$$
P_{1}(x) \vee P_{2}(y) \leftarrow A(x) \wedge B(y)
$$
is rewritten into $A \sqcap \exists
U.(B \sqcap \neg P_{2}) \sqsubseteq P_{1}$.

Now consider Point~2. The translation from ($\mathcal{ALCIU}$,AQ) to unary simple \MDD queries is 
a modified version of the translation given in the proof of Theorem~\ref{thm:ALCMDD1} for the translation from
($\mathcal{ALC}$,AQ) to \emph{connected} unary simple \MDD queries.

Assume that $(\Sbf,\Omc,q)$ with $q=A(x)$ is given. 
As in Theorem~\ref{thm:ALCMDD1}, 
we take types to be subsets of $\mn{sub}(\Omc)$. 
%
The \MDD program $\Pi$ consists of the following rules:
  $$
  \begin{array}{r@{\;}c@{\;}l@{\!\!\!\!\!\!\!\!\!\!}l}
\displaystyle\bigvee_{\tau\subseteq \mn{sub}(\Omc)} \!\!\!\!\!\!
P_\tau(x) &\leftarrow& \mn{adom}(x)
 \\[-3.5mm]
\bot &\leftarrow&\delta(\textbf{x}) & \text{ for all non-realizable
 diagrams $\delta(\textbf{x})$} \\[0.5mm]
 && & \text{ of
  the form $P_{\tau_{1}}(x_{1}) \wedge P_{\tau_{2}}(x_{2})$,}\\[0.5mm]
 && & \text{$P_{\tau}(x) \wedge A(x)$, or }\\[0.5mm]  
 && & \text{$P_{\tau_1}(x_1)\land S(x,y)\land P_{\tau_2}(x_2)$} \\[0.5mm]
\mn{goal}(x) &\leftarrow& P_\tau(x)  & \quad \text{ for all }P_\tau \text{ with } A \in P_\tau
  \end{array}
$$
%
Note that the only difference with the rules in the proof of Theorem~\ref{thm:ALCMDD1}
is the presence of rules of the form
$$
\bot \leftarrow P_{\tau_{1}}(x_{1}) \wedge P_{\tau_{2}}(x_{2})
$$
which are not connected. $\Pi$ is still unary and simple. Equivalence of $(\Sbf,\Omc,q)$
and $q_\Pi$ can now be proved similarly to Theorem~\ref{thm:ALCMDD1}.

\medskip

It remains to be shown that ($\mathcal{ALCIU}$,AQ) and ($\mathcal{SHIU}$,AQ)
are equally expressive. But this is again folkore \cite{DBLP:phd/de/Motik2006,DBLP:conf/dlog/Simancik12}:
it is known that for every $\mathcal{SHIU}$-ontology $\Omc$, there exists 
an $\mathcal{ALCIU}$-ontology $\Omc'$ (possibly using additional concept names)
such that (i) $\Omc'\models \Omc$ and (ii) for every $\Amf\in \mn{Mod}(\Omc)$, there
exists a model $\Amf'\in \mn{Mod}(\Omc')$ with the same domain and interpreting the concept names 
of $\Omc$ in the same way as $\Amf$ and interpreting the role names as relations containing their 
interpretation in $\Amf$.
It follows that ($\mathcal{ALCIU}$,AQ) and ($\mathcal{SHIU}$,AQ) are equally expressive.
\end{proof}

We briefly discuss \emph{Boolean atomic queries} (BAQs),
i.e., queries of the form $\exists x.A(x)$, where $A$ is a unary relation symbol.
BAQs behave similarly to AQs and one can show modified versions of  
Theorems~\ref{thm:ALCMDD1} to Theorem~\ref{thm:ALCMDD3} above in which AQs are replaced by 
BAQs and unary goal predicates by $0$-ary goal-predicate, respectively.

\begin{theorem}
Theorems~\ref{thm:ALCMDD1} to Theorem~\ref{thm:ALCMDD3} hold if AQs are replaced by 
BAQs and unary goal predicates by $0$-ary goal-predicate, respectively.
\end{theorem}
\begin{proof}
We show the required modifications to the proof of Theorem~\ref{thm:ALCMDD1}. The remaining
results are proved by similar modifications and left to the reader. For the translation from
($\mathcal{ALC}$,BAQ) to Boolean connected simple \MDD, the only difference
to the program constructed in the proof of Theorem~\ref{thm:ALCMDD1} is that
rules of the form $\mn{goal}(x) \leftarrow P_\tau(x)$ are replaced by rules of the form 
$\mn{goal} \leftarrow P_\tau(x)$.
Conversely, for the translation from Boolean connected simple \MDD to
($\mathcal{ALC}$,BAQ), we regard $\mn{goal}$ as a concept name and take the 
BAQ $\exists x.\mn{goal}(x)$.
The rewriting of goal rules must also be accordingly modified.
For example, $\mn{goal}\leftarrow R(x,y)$ is rewritten into $\exists R.\top \sqsubseteq \mn{goal}$.
\end{proof}

\subsection{Proofs for Section~\ref{sect:FOOntologies}}

\begin{trivlist}
\item \textbf{Theorem~\ref{thm:UNFO}}
(UNFO,UCQ) has the same expressive power as \MDD.
\end{trivlist}

  \begin{proof} (continued)
    We establish here the correctness of the translation from
    (UNFO,UCQ) to \MDD. That is,
    we show that, for every instance $\Dmf$ and elements
    $\abf 
     \in \mn{adom}(\Dmf)$, we
    have $\abf\in \mn{cert}_{q,\Omc}(\Dmf)$ if and only if
    $\abf\in q_\Pi(\Dmf)$. The ``if'' direction proceeds exactly as in
    the proof of Theorem~\ref{thm:ALCtoMDD}, so here we focus on the ``only if''
    direction.
 
%

    \medskip

    ``only if''. Assume that $\abf\not\in q_\Pi(\Dmf)$ and let $\Dmf'$ be a model of
    $\Pi$ with $\Dmf\subseteq\Dmf'$ that does not contain
    $\mn{goal}(\abf)$. For each $a \in \mn{adom}(\Dmf)$, let
    $\mu(a)$ be the unique type such that $P_{\mu(a)}(a) \in \Dmf'$, and
    let $(\mn{dom}_a,\Dmf_a)$ be a model of \Omc in which $\mu(a)$ is realized at
    $a$. Note that such a model must exist because otherwise the diagram
    $P_{\mu(a)}(x)$ would be non-realizable and $\Pi$ would include a
    rule $\bot\leftarrow P_{\mu(a)}(x)$. We may assume that these
    models have disjoint domains.
    Let $(\mn{dom}'',\Dmf'')$ be obtained by first taking the
    union of $(\mn{dom}_a,\Dmf_a)_{a \in \mn{adom}(\Dmf)}$, and then
    adding to it all facts of $\Dmf$. We show that
\begin{enumerate}
\setlength{\itemindent}{1em}
\item[(i)] $(\mn{dom}'',\Dmf'')$ is a model
    of \Omc, and
\item[(ii)] $\abf\not\in q(\Dmf'')$.
\end{enumerate}

    We start with the first claim.
    Let $\mu(d)$ be the unique type realized by $d$ in
    $(\mn{dom}_a,\Dmf_a)$, for all $d\in\mn{dom}_a$.
     We show the following by induction on the structure
    of $\vp$:
    \begin{description}
    \item[($\ast$)]
      For all 
      $\vp \in \mn{cl}_k(\Omc)$ and $d \in
      \mn{dom}''$, we have that
      $\varphi \in \mu(d)$ iff $ (\mn{dom}'',\Dmf'') \models \vp[d] $. 
      \end{description}
    %
    Note that $\vp$ may be either a sentence or a formula with exactly one free variable,
    and in the former case, we interpret $\vp[d]$ as $\vp$.
    Since all types $\mu(d)$ must include the sentence $\Omc$, 
    $(\ast)$ implies (i).

    The base case  ($\vp=\top$) and the inductive step for formulas of the
    form $\neg\psi(x)$ are omitted since they are straightforward.
    Thus, let $\vp$  be a formula from $\mn{cl}_k(\Omc)$ of the form
    $\exists \ybf \bigwedge_i \psi_{i}(x,\ybf)$, and let $d\in\mn{dom}_a$.
    We may assume that $\vp$ is connected, meaning that
    the graph whose nodes are the subformulas $\psi_i$ and containing
    an edge between $\psi_i$ and $\psi_j$ if they share a 
    variable,
    is connected. This is because, if $\vp$ is not
    connected, then the claim follows immediately from the analogous
   claims for each of the connected components of $\vp$.
    We present the proof for the case where $\vp$ has answer variable~$x$
    (the argument for sentences is similar).


     First suppose that $\vp \in \mu(d)$, which means $(\mn{dom}_a,\Dmf_a) \models
     \vp[d]$. It follows that there is an assignment $\pi$ of elements of
    $\mn{dom}_a$ to the variables $x,\mathbf{y}$ such that
    $\pi(x)=d$ and for every $i$,
    $(\mn{dom}_a,\Dmf_a) \models \psi_{i}(\pi(x,\mathbf{y}))$.
     If $\psi_{i}$ is an atomic formula, then using the fact that
     $\Dmf_a\subseteq\Dmf''$, we obtain
     $(\mn{dom}'',\Dmf'') \models \psi_{i}(\pi(x,\mathbf{y}))$.
     If $\psi_i$ is not atomic, then it must have at most
     one free variable $u$. We thus have that 
     $(\mn{dom}_a,\Dmf_a) \models \psi_{i}[\pi(u)]$,
     so $\psi_i \in \mu(\pi(u))$. 
     Applying the induction
     hypothesis, we obtain $(\mn{dom}'',\Dmf'') \models \psi_{i}[\pi(u)]$. 
     It follows that $\pi$ is a satisfying assignment for $\vp$ in $(\mn{dom}'',\Dmf'')$,
     hence $(\mn{dom}'',\Dmf'')\models\vp[d]$.

    Conversely, suppose $(\mn{dom}'',\Dmf'')
    \models \vp[d]$, that is, $(\mn{dom}'',\Dmf'')$ satisfies $\bigwedge_i
    \psi_{i}(x,\mathbf{y})$ for some assignment $\pi$ of elements of
    $\mn{dom}''$ to the variables $x,\mathbf{y}$ such that
    $\pi(x)=d$.
%
%
    First assume that the image of
    $\pi$ is entirely contained in $\mn{dom}_a$.
    Using the induction hypothesis to treat the non-atomic $\psi_i$ as
    before, we then get that
    $(\mn{dom}_a,\Dmf_a)
    \models \vp[d]$, hence $\varphi \in \mu(d)$ as required.
   
    Next suppose that the image of
    $\pi$ is not wholly contained in $\mn{dom}_a$, and
    let $I$ be the set consisting of the elements of
    $\mn{adom}(\Dmf)$ that are in the range of $\pi$.
    By the connectedness assumption and the fact that
    $d \in \mn{dom}_a$, the set $I$ contains $a$.
    In what follows, we will define a
    number of formulas by syntactic operations on $\varphi$. It will
    follow from the definition of $\mn{cl}_k(\Omc)$ that each of these
    formulas again belongs to $\mn{cl}_k(\Omc)$, and hence, is
    subject to the induction hypothesis.
Let $\varphi'$ be obtained from $\varphi$ by
    identifying all variables $z,z'$ such that $\pi(z)=\pi(z') \in I$.
    We assume that the free variable $x$ retains its name, and
    use $\psi_i'$ to denote the conjunct of $\varphi'$ which corresponds to $\psi_i$.
    For each $b \in I$, let $z_b \in \ybf \cup \{x\}$ be the unique
    variable in $\vp'$ with $\pi(z_b)=b$.
    Let $\vp'_b$ be the restriction of $\vp'$ to those  $\psi'_i$
    which contain only variables $z$
    with $\pi(z)\in \mn{dom}_b$, with free variable $z_b$.
    We have
   $(\mn{dom}'',\Dmf'')
    \models \vp'_b[b]$ via the restriction of $\pi$ to the
    variables in $\vp'_b$, thus, 
    by the earlier argument (since all witnessing elements are contained
    in $\mn{dom}_b$), we have $\vp'_b \in \mu(b)$. Let $\vp'_0$ be $\vp'$, but
    with free variable $z_a$ instead of $x$.
    Note that $(\mn{dom}'',\Dmf'')
    \models \vp'_0[a]$.
   
    Consider the diagram $\delta$ obtained by taking the restriction of
    $\Dmf'$ to $I$, and then replacing each $b \in I$ with $z_b$.
    Since $\delta$ is made true by $\Dmf'$, and $\Dmf'$ is a model of $\Pi$,
    we have that $\delta$ is a realizable diagram. 
    Moreover, using the fact that $P_{\mu(b)}(z_b) \in \delta$ and $\vp'_b \in \mu(b)$ for every $b \in I$,
    one can show that the diagram $\delta$ implies the query $\vp'_0$.
    This together with the realizability of $\delta$ yields $\vp'_0 \in
    \mu(a)$, hence $(\mn{dom}_a,\Dmf_a)
    \models \vp'_0[a]$.
    Let $\pi'$ be a satisfying assignment of $\vp'_0$ in
    $\Dmf_a$ such that $\pi'(z_a)=a$. 
    We use $\pi'$ to construct a satisfying assignment $\pi''$ of
    $\varphi'$ mapping $x$ to $d$,
    such that the range of $\pi''$ lies entirely inside
    $\mn{dom}_a$. The assignment $\pi''$ is defined as follows: 
    for all $u$ with $\pi(u)$ in $\mn{dom}_a$, set
    $\pi''(u)=\pi(u)$; for all other $u$, set $\pi''(u)=\pi'(u)$. 
    To see that $\pi''$ is indeed a satisfying assignment of
    $\vp'$,  note that each conjunct of $\vp'$ contains, besides $z_a$, 
    either only variables $u$ with $\pi(u)\in\mn{dom}_a$, 
    or only variables $u$ with $\pi(u)\not\in\mn{dom}_a$. 
    The former conjuncts are satisfied because $\pi$ is a match, and the 
   latter conjuncts are satisfied because $\pi'$ is a match.
    Moreover,
    $\pi''(x)=d$. 
    Therefore, $(\mn{dom}_a,\Dmf_a) \models \vp[d]$ and
    hence  $\vp\in\mu(d)$ as required.

    \medskip
    Finally, we can show (ii)  
    in a similar way. We suppose, for the sake of contradiction, that
    $\abf\in q(\Dmf'')$ under some assignment $\pi$ to the
    existentially quantified variables in $q$. Let
    $\bbf$ be the elements of $\mn{adom}(\Dmf)$ belonging to the
     range of $\pi$ (here again we focus on the case in which
     $q$ is connected and contains at least one free variable). 
     Then, in the same way as above, we can decompose
    $q$ into unary subqueries $q_b$ that are satisfied in the different
    subinstances $\Dmf_b$ with $b\in\bbf$, and conclude
    that $q_b\in\mu(b)$ for each $b\in\bbf$. We can then show that the
    diagram obtained
    by taking all facts in $\Dmf'$ over
     elements in $\bbf$ and replacing each $b \in \bbf$ by $z_b$
     implies the query $q$. This yields the desired contradiction since $\Dmf'$ is a model of $\Pi$.
  \end{proof}

\medskip

\noindent\textbf{Proposition
 \ref{prop:gfucqmdd}}. 
 \emph{   The Boolean query 
\begin{itemize}
\item[($\dagger$)] there are $a_1,\dots,a_n,b$, for some $n \geq 2$, such that
  $A(a_1)$, $B(a_n)$, and $P(a_i,b,a_{i+1})$ for all $1 \leq i < n$
\end{itemize}
is definable in (GF,UCQ)
  and not in \MDD.}
   
\medskip  

\noindent\begin{proof}
Let $\Sbf$ consist of unary predicates $A,B$ and a ternary predicate $P$,
and let $Q$ be the $\Sbf$-query defined by $(\dagger)$. 
A (GF,UCQ) query expressing $Q$ was given in the body of the paper. 
It thus remains to show that $Q$ cannot be expressed in
 \MDD. We make use of the characterization of \MDD queries 
 in terms of $k$-colorings provided by Lemma \ref{lem:crit}. 
 
 Assume that $m,n$ are given. Let $k= m^{n}+2n$.
 Define $\Sbf$-instances $\Dmf_{1}$ and $\Dmf_{0}$ as follows:
 \begin{itemize}
 
 \item $\Dmf_{1}$ has elements $d_1,\dots,d_{k},e$ and the
   atoms $A(d_1)$, $B(d_{k})$, and $P(d_i,e,d_{i+1})$ for $1
   \leq i< k$.
 
 \smallskip
 
 \item $\Dmf_{0}$ has elements $d_1,\dots,d_{k}$, and $e_{1},\dots,e_{k}$ and
   the following atoms: $A(d_1)$, $B(d_{k})$, and
   $P(d_i,e_j,d_{i+1})$ whenever $1 \leq i < k$, $1 \leq j < k$, and
   $j \not= i$.
 
 \end{itemize}
 It is readily checked that $Q(\Dmf_{1})=1$ and $Q(\Dmf_{0})=0$, as required.
 Let $\Bmf_{0}$ be an $m$-coloring of $\Dmf_{0}$. 
 Define an $m$-coloring $\Bmf_{1}$ of $\Dmf_{1}$ by giving all elements of 
  $\{d_{1},\ldots,d_{k}\}$ exactly
 the same color as in $\Bmf_{0}$. Choose $i$ with $n<i<k-n$ in such a way that
 for every sequence $d_{l},\ldots,d_{l+n}$ with $l>1$ and $l+n<k$ 
 there exists a sequence $d_{l'},\ldots,d_{l'+n}$ with $l'>1$ and $l'+n<k$
 such that the coloring of $d_{l},\ldots,d_{l+n}$ coincides with the coloring
 of $d_{l'},\ldots,d_{l'+n}$ and $i\not\in \{l',l'+n\}$. Such an $i$ exists since
 $k\geq m^{n}+2n$. Now give $e$ the color of $e_{i}$. One can now easily
 construct, for every structure corresponding to an $n$-element subset of $\Bmf_{1}$, a homomorphism to~$\Bmf_{0}$.
 \end{proof}

\begin{trivlist}
\item\textbf{Theorem~\ref{GFUCQfrontier}}
  (GF,UCQ) and (GNFO,UCQ) have the same expressive power as \FGDD.
\end{trivlist}

\begin{proof}
  We start by describing the translation from \FGDD to (GNFO,UCQ). Let
  $\Pi$ be a \FGDD query. 
   It is easily verified that if we write out the implication symbol in a
   frontier-guarded \DD rule using conjunction and
   negation, the resulting formula belongs to GNFO.
  Thus, we can take $\Omc$ to be the set of
  all non-goal rules of $\Pi$, viewed as a GNFO sentence, and let $q$ 
  be the UCQ that consists of all bodies of rules whose conclusion contains
  the IDB relation $\mn{goal}$. It is easy to check that the ontology-mediated
  query ($\Sbf,\Omc,q)$, where $\Sbf$ is the schema consisting
  of all EDB relations, is equivalent to the \FGDD query $q_\Pi$.

   Next, we explain how to translate (GNFO, UCQ) to \FGDD. Since every
   sentence of GF is equivalent to a sentence of GNFO \cite{GNFO}, this
   also yields a translation of (GF,UCQ) to \FGDD.
%
%
   Recall that we used a specific normal form for UNFO sentences. For
   GNFO, we can use an analogous normal form. Specifically, we can
   assume that $\Omc$ is generated by the following grammar:
   \[ \varphi(\textbf{x}) ::= \top ~\mid~
   \alpha(\textbf{x})\land\neg\varphi(\textbf{x}) ~\mid~
   \exists\textbf{y}(\psi_1(\xbf,\textbf{y})\land\cdots\land\psi_n(\xbf,\textbf{y}))\]
   where each $\psi_i$ is either a relational atom or a formula
   generated by the same grammar whose free variables are among
   $x,\ybf$.  The ``guard'' $\alpha$ is an atomic formula, possibly an
   equality, containing \emph{all variables} in $\textbf{x}$.

   Let $\mn{sub}(\Omc)$ be the set of all subformulas of $\Omc$.
  Let $k$ be the maximum of the
  number of variables in $\Omc$ and the number of variables in $q$.
  For $\ell\geq 0$, we
  denote by $\mn{cl}^\ell_k(\Omc)$ the set of all formulas
  $\chi(\textbf{x})$ with $\textbf{x}=(x_1,
  \ldots, x_\ell)$ of the
  form
  \[ \exists \textbf{y}(\psi_1(\textbf{x},
  \textbf{y})\land\cdots\land\psi_n(\xbf,\textbf{y}))\] with
  $\textbf{y}=(y_1, \ldots, y_m)$, $m+\ell\leq k$, and such that each
  $\psi_i$ is either an atomic formula that uses a symbol from $q$ or is of the form
  $\chi(\textbf{z})$ for some $\chi(\textbf{z}') \in \mn{sub}(\Omc)$.  

  A \emph{guarded $\ell$-type} $\tau$ 
  is a subset of $\mn{cl}^\ell_k(\Omc)$ that contains at least one
  atomic relation (possibly equality) containing all variables $x_1,
  \ldots, x_\ell$, and also contains the sentence $\Omc$ itself. We
  denote the set of all guarded $\ell$-types by
  $\mn{type}_\ell(\Omc)$.  Note that, by definition, there are no
  guarded $\ell$-types for $\ell$ greater than the maximal arity of a
  relation from $\Sbf$.

   We now proceed the same way as we did in the case of UNFO (but
   using guarded $\ell$-types instead of unary types).
  We introduce a fresh $\ell$-ary relation symbol $P_\tau$ for
  each guarded $\ell$-type $\tau$, 
   and we denote
  by $\Sbf'$ the schema that extends $\Sbf$ with these additional
  relations. Diagrams, realizability, and implying a query are defined in
   the same way as before. The \DD program is also constructed
   in essentially the same manner, 
  except that the first rule of the program is replaced by the
  following:
\[\!\!\!\!\!\!\mathop{\bigvee_{\tau \text{ a guarded
      $\ell$-type}}}_{\text{with $R(\textbf{x})\in \tau$}} \!\!\!\!\!\!\!\!\!\!\!\!\! P_\tau(\xbf) \leftarrow R(\xbf)
~~ \text{ for each relation $R$ of arity $\ell\geq 0$.}\]

We establish the correctness of the translation. That is,
we show that, for every instance $\Dmf$ and elements $\abf =
a_1,\dots,a_n \in \mn{adom}(\Dmf)$, we have $\abf\in
\mn{cert}_{q,\Omc}(\Dmf)$ if and only if $\abf\in q_\Pi(\Dmf)$.
  
   \medskip

   ``if''. Assume that  $\abf\not\in \mn{cert}_{q,\Omc}(\Dmf)$. 
      Then there is $(\mn{dom},\Dmf')\in\mn{Mod}(\Omc)$ with 
      $\Dmf\subseteq \Dmf'$ such that $\abf\not\in q(\Dmf')$. 
     For every fact $R(\bbf)$ of $\Dmf$, let $\mu(\bbf)$ be the unique
     guarded $\ell$-type (with $\ell=|\bbf|$)
    realized at $a$ in $\Dmf'$. Let $\Dmf''$ be the instance that
    consists of the atoms in $\Dmf$ and the atom $P_{\mu(\abf)}(\bbf)$ for each
    fact $R(\bbf)$ in $\Dmf$. It can be checked that $\Dmf''$ is a model
    of $\Pi$. Since $\mn{goal}(\abf) \notin \Dmf''$,
    $\abf\not\in q_\Pi(\Dmf)$. 

    \medskip

    ``only if''. Assume that $\abf\not\in q_\Pi(\Dmf)$ and let $\Dmf'$ be a model of
    $\Pi$ with $\Dmf\subseteq\Dmf'$ that does not contain
    $\mn{goal}(\abf)$. We say that a tuple $\bbf$ is ``live'' in
    $\Dmf$ if $\Dmf$ contains $R(\bbf)$ for some relation symbol  $R$.
    For each live tuple $\bbf$ of $\Dmf$, let
    $\mu(\bbf)$ be the unique guarded $\ell$-type  
    (with $\ell=|\bbf|$) such that $P_{\mu(\bbf)}(\bbf) \in \Dmf'$, and
    let $(\mn{dom}_{\bbf},\Dmf_{\bbf})$ be a model of \Omc in which $\mu(\bbf)$ is realized at
    $\bbf$ (such a model must exist because otherwise the diagram
    $P_{\mu(\bbf)}(\xbf)$ would be non-realizable and $\Pi$ would include a
    rule $\bot\leftarrow P_{\mu(\bbf)}(\xbf)$). We may assume that for 
   distinct live tuples $\bbf$ and $\cbf$, $\mn{dom}_{\bbf}$ and
   $\mn{dom}_{\cbf}$ overlap only (possibly) on $\{\bbf\}\cap\{\cbf\}$.
    Let $(\mn{dom}'',\Dmf'')$ be obtained by first taking the
    union of $(\mn{dom}_\bbf,\Dmf_\bbf)$ for all live tuples $\bbf$ of $\Dmf$, and then
    adding to it all facts of $\Dmf$. We show that
\begin{enumerate}
\setlength{\itemindent}{1em}
\item[(i)] $(\mn{dom}'',\Dmf'')$ is a model
    of \Omc and
\item[(ii)] $\abf\not\in q(\Dmf'')$.
\end{enumerate}

For all live tuples $\dbf$ of $\Dmf_\bbf$, let $\mu(\dbf)$ be the
unique guarded $\ell$-type realized by $\dbf$ in
$(\mn{dom}_\bbf,\Dmf_\bbf)$, for all $d\in\mn{dom}_a$. Note that a
tuple $\dbf$ may be live in $\Dmf_\bbf$ for several different choices
of $\bbf$, but
%
%
then the guarded $\ell$-type realized by $\dbf$ in each such
$(\mn{dom}_\bbf,\Dmf_\bbf)$ is the same: otherwise, there must be some
atom $R(\ybf)$ that belongs to $\mu(\bbf)$, but not to $\mu(\bbf')$,
and then the diagram $P_{\mu(\bbf')}(\xbf) \wedge R(\ybf)$ is
non-realizable and thus ruled out by $\Pi$.

    Claim (i) is proved by establishing the following, by induction on the
    length of $\vp$:
    \begin{description}
    \item[($\ast$)]
      For all formulas $\vp(\xbf) \in \mn{cl}_k^\ell(\Omc)$ and for
      each live $\ell$-tuple $\dbf$ of $\Dmf''$, we have $(\mn{dom}'',\Dmf'') \models \vp[\dbf]$ iff $\vp \in
      \mu(\dbf)$.
    \end{description}
    We omit the proofs of ($\ast$) and of (ii), as they proceed similarly to the
    proofs of Theorem~\ref{thm:ALCtoMDD} and \ref{thm:UNFO}.
\end{proof}

\savebox{\stupid}{\large\bf\ref{sect:obdammsnp}}
\section{Proofs for Section~\usebox{\stupid}}

In Section~\ref{subsec:b1}, we start by establishing a central
technical result about MMSNP extended with constant symbols which
allows us to lift key results from MMSNP sentences to 
coMMSNP queries (with free variables). 
Then in Section~\ref{subsec:b2}, we provide the proofs for the results
stated in Section~\ref{sect:obdammsnp} of the main paper. 

\subsection{MMSNP with Constant Symbols}
\label{subsec:b1}

For readability, throughout this subsection, we will adopt a more 
convenient notation for schemas and structures involving constant 
symbols. If $\Sbf$ is a schema and 
$\cbf$ a (possibly empty) set of constant symbols,  then we will use 
$\Sbf_\cbf$ as a shorthand for $\Sbf \cup \cbf$. 
A $\Sbf _\cbf$-structure 
$\Bmf$ will be given by a pair 
$(\dom(\Bmf), \cdot^{\Bmf})$, where 
$\dom(\Bmf)$ 
is a finite, non-empty set and  $\cdot^{\Bmf}$ is a function assigning to each $n$-ary predicate in $\Sbf$ an $n$-ary relation $P^{\Bmf}$ over $\dom(\Bmf)$ 
and to each constant symbol $c \in \cbf$ an element $c^{\Bmf} \in \dom(\Bmf)$. 
We use $\adom(\Bmf)$ to denote the active domain of $\Bmf$, and we call 
$\Bmf$ an \emph{active domain structure} if $\dom(\Bmf)=\adom(\Bmf)$. 


\medskip

Our objective is to establish the following theorem, which lifts the containment and dichotomy results for MMSNP sentences \cite{FederVardi} to coMMSNP queries:

\begin{theorem}\label{lift-commsnp}
coMMSNP has a dichotomy between \PTime and \coNP iff the Feder-Vardi conjecture holds. Containment of coMMSNP queries is decidable. 
\end{theorem}

We prove Theorem \ref{lift-commsnp} in several steps. We consider the language
\emph{MMSNP with constant symbols} (abbreviated MMSNP$_{c}$), consisting 
of all sentences which can be obtained from MMSNP formulas by replacing 
each free variable by a constant symbol. 
The evaluation problem for MMSNP$_{c}$ 
consists in deciding whether an MMSNP$_{c}$ sentence with schema $\Sbf$ and 
constant symbols $\cbf$ holds in a given $\Sbf_\cbf$-structure 
$\Bmf$. 
The containment problem for MMSNP$_{c}$ is to decide for two MMSNP$_{c}$ sentences $\Psi_{1}, \Psi_{2}$ with relations $\Sbf$ and constants symbols $\cbf$, whether $\Bmf \models \Psi_{1}$  implies $\Bmf \models \Psi_{2} $ for all  $\Sbf_\cbf$-structures $\Bmf$.  
We use $\Psi_{1} \subseteq \Psi_{2}$ to denote containment. 

MMSNP$_{c}$ will serve as a bridge between coMMSNP queries (with free variables)
and MMSNP sentences. More precisely, we will first show that evaluation of 
coMMSNP queries is polynomially equivalent to evaluation of MMSNP$_{c}$ sentences,
and show a polynomial reduction from coMMSNP query containment to 
containment of MMSNP$_{c}$ sentences. Afterwards, we will move from MMSNP$_{c}$
sentences to MMSNP sentences, again showing polynomial equivalence of the evaluation
problems and a polynomial reduction for containment. 

To link coMMSNP queries and MMSNP$_{c}$, it will actually prove 
more convenient to suppose that MMSNP$_{c}$ sentences are interpreted over 
active domain structures, whereas to relate MMSNP$_{c}$
with plain MMSNP, we will wish to work over arbitrary structures. 
Thus, as a preliminary step, we relate the two variants of the 
MMSNP$_{c}$ evaluation and containment problems. 


\begin{lemma}\label{eval-activedomain}
The evaluation problem for MMSNP$_{c}$ restricted to active domain structures is polynomially 
equivalent to the evaluation problem for MMSNP$_{c}$ (over general structures).
\end{lemma}
\begin{proof}
Let $\Phi=\exists X_1 \cdots \exists X_{\ell} \forall x_1 \cdots \forall x_{m} \vp$ 
be an MMSNP$_{c}$ sentence over schema $\Sbf$ and constants $\cbf$,
which is interpreted over active domain structures. 
Pick a fresh second-order variable $Y$ 
and a fresh constant $c$ not appearing in $\cbf$. 
Let $\vp'$ be the formula obtained from $\vp$ 
by replacing every conjunct $\psi_1 \rightarrow \psi_2$ of $\varphi$ by 
$\psi_1 \rightarrow (\psi_2 \vee Y(c))$. 
Let $\chi$ be the conjunction of all formulas of the form $R(x_{1}, \ldots, x_{k}) \rightarrow \neg Y(x_{i})$,
where $R$ is a $k$-ary relation in $\Sbf$, and $x_{i}$ is one of the variables among $x_{1}, \ldots, x_{k}$.
Define a new MMSNP$_{c}$ sentence 
$$\Phi' = \exists X_1 \cdots \exists X_{\ell} \exists Y \forall x_1 \cdots \forall x_{m} (\vp' \wedge \chi)$$
We claim that the evaluation problem for $\Phi$ over active domain structures is 
polynomially equivalent to the evaluation problem for $\Phi'$ over general structures.
The first reduction is trivial since for every $\Sbf_\cbf$-structure $\Amf$ 
such that $\dom(\Amf)=\adom(\Amf)$,
we have $\Amf \models \Phi$ if and only if $\Amf \models \Phi'$. To see why, notice that 
$\chi$ ensures that $Y$ is false everywhere on the active domain, so the additional disjuncts 
have no effect. For the second reduction, we remark that $\Bmf \models \Phi'$
for a general $\Sbf_\cbf$-structure $\Bmf$ if and only if $\dom(\Bmf)\neq \adom(\Bmf)$ (since
we can trivially satisfy $\Phi'$ by sending $c$ to an element outside the active domain and including that
element in $Y$)
or $\dom(\Bmf)=\adom(\Bmf)$ and $\Bmf \models \Phi$. 

It remains to be shown that every evaluation problem for MMSNP$_{c}$ over general structures
is polynomially equivalent to an evaluation problem for MMSNP$_{c}$ over active domain structures.
Let $\Phi$ be an MMSNP$_{c}$ sentence with schema $\Sbf$ and constant symbols $\cbf$,
and select a fresh monadic second order variable $Y$, a fresh input relation $\mn{Elem}$, and   
and a fresh constant symbol  $c$. We define 
$\Phi'$ as the sentence over $\Sbf \cup \{\mn{Elem}\}\cup\cbf \cup \{c\}$  
obtained from $\Phi$ by: 
\begin{itemize}
\item replacing every conjunct $\psi_1 \rightarrow \psi_2$ 
by $\psi_1 \wedge \bigwedge_{t \in T}\mn{Elem}(t) \rightarrow \psi_2 \vee Y(c)$, where 
$T$ is the set of terms appearing in  $\psi_1 \rightarrow \psi_2$,
\item adding a new conjunct $\mn{Elem}(x) \rightarrow \neg Y(x)$, and
\item adding $Y$ to the initial sequence of existentially quantified monadic second-order variables.
\end{itemize}
We claim that the evaluation problem for $\Phi$ over general structures is
polynomially equivalent to the evaluation problem for $\Phi'$ over active domain structures. 
For the first reduction, we have that for every $\Sbf_\cbf$-structure $\Bmf$,  
$\Bmf \models \Phi$ if and only if $\Bmf' \models \Phi'$, 
where $\Bmf'$ extends $\Bmf$ by setting $\mn{Elem}^{\Bmf'}= \dom(\Bmf)$ and letting $c^{\Bmf'}$
be any element in $\dom(\Bmf)$. 
For the other reduction, we have that for every $\Sbf \cup \{\mn{Elem}\}\cup\cbf \cup \{c\}$-structure
$\Bmf$ with $\dom(\Bmf)=\adom(\Bmf)$, 
$\Bmf \models \Phi'$ if and only if either $\mn{Elem}^\Bmf \neq \dom(\Bmf)$ or $\Bmf'  \models \Phi$,
where $\Bmf'$ is obtained by taking the $\Sbf \cup\cbf$-reduct of $\Bmf$.  
\end{proof}

\begin{lemma}\label{ad-cont}
Containment of MMSNP$_{c}$ over active domain structures is
polynomially reducible to containment of MMSNP$_{c}$ (over arbitrary structures). 
\end{lemma}
\begin{proof}
Consider MMSNP$_{c}$ sentences $\Phi_{1}, \Phi_{2}$ with schema $\Sbf$ and 
constants $\cbf$. We apply the construction from the first part of the proof of Lemma \ref{eval-activedomain}
to obtain MMNSP$_{c}$ sentences $\Phi_{1}'$ and $\Phi_{2}'$ with the property that 
$\Bmf \models \Phi_{i}'$ for a general $\Sbf_\cbf$-structure $\Bmf$ if and only if 
$\dom(\Bmf)\neq \adom(\Bmf)$ 
or $\dom(\Bmf)=\adom(\Bmf)$ and $\Bmf \models \Phi_{i}$ (for $i \in \{1,2\}$). 
It is readily verified that $\Phi_{1} \subseteq \Phi_{2}$ for the class of active domain 
structures  if and only if 
$\Phi_{1}' \subseteq \Phi_{2}'$.
\end{proof}

By the preceding lemmas, we can choose to work with active domain structures. 
It is then straightforward to relate the evaluation and containment problems for coMMSNP
queries with the corresponding problems for
MMSNP$_{c}$ sentences.

\begin{lemma}\label{meghynlem2}
The evaluation problem for coMMSNP is polynomially equivalent to the evaluation problem for MMSNP$_{c}$. 
Containment of coMMSNP queries is polynomially reducible to containment of MMSNP$_{c}$ sentences.
\end{lemma}

The next step, and the core technical contribution of this subsection, is 
to relate the evaluation and containment of MMSNP$_{c}$ sentences 
to the analogous problems for MMSNP sentences. To simplify the 
technical constructions, it will prove convenient to work with forbidden pattern problems \cite{DBLP:journals/siamcomp/MadelaineS07,DBLP:journals/ejc/KunN08,DBLP:journals/siamdm/BodirskyCF12}.

We extend forbidden patterns problems to handle constant symbols, by
simply substituting $\Sbf\cup\cbf$-structures for 
$\Sbf$-structures in Definitions \ref{color-def} and \ref{forb-def}.
We denote by FPP$_{c}$ the class of forbidden patterns problems thus defined, 
and use FPP to refer to the restriction to structures without constant symbols.
Note that both FPP$_{c}$ and FPP define 
problems over structures, not instances 
(although this distinction is irrelevant in the absence of constant symbols).

It was shown in \cite{DBLP:journals/siamcomp/MadelaineS07} that MMSNP sentences and 
FPP have the same expressive power. This result can be 
straightforwardly extended to handle constant symbols:

\begin{lemma}
MMSNP$_{c}$ and FPP$_{c}$ have the same expressive power (over structures with constant symbols). 
\end{lemma}

By the previous lemma and the fact that FPP is a subset of FPP$_{c}$, 
to show polynomial equivalence of MMSNP$_{c}$ and MMSNP 
it suffices to show that every problem in FPP$_{c}$ 
is polynomially equivalent to some problem in FPP.  
To formulate the reductions, we will require some additional
notation and terminology, which we introduce next.

Let $\Sbf$ be a schema, $\cbf=\{c_{1}, \ldots,c_{n}\}$ be a set of constant symbols, and 
$P=\{P_{1},\ldots,P_{n}\}$ be a set of unary predicates which do not appear in $\Sbf$. 
We will 
abbreviate $\Sbf \cup P$ to $\Sbf_{P}$.

We define operations which allow us to transform
$\Sbf_{P}$-structures into $\Sbf_{\cbf}$-structures, and vice-versa. 
With every $\Sbf_{P}$-structure $\Bmf$ with $P_{i}^{\Bmf}\not=\emptyset$ 
for all $1\leq i\leq n$, we associate the $\Sbf_{\cbf}$-structure 
$\Bmf^{c}$, called the \emph{collapse} of $\Bmf$, by factorizing through the $P_{i}^{\Bmf}$.
Specifically, let $\sim$ be the smallest equivalence relation such that whenever 
$d,d'\in P_{i}^{\Bmf}$ for some $i$, then $d \sim d'$. 
Then $\dom(\Bmf^{c})$ is $\{ [d] \mid d\in \Delta^{\Bmf}\}$, 
where $[d]$ denotes the equivalence class of $d$ w.r.t.~$\sim$. 
For convenience, when $[d]=\{d\}$, we will use $d$ in place of $[d]$. 
Set $c_{i}^{\Bmf^{c}}= [d]$, for some $d\in P_{i}^{\Bmf}$, and define $R^{\Bmf^{c}}$
as follows: $([d],[e]) \in R^{\Bmf^{c}}$ if and only if there exist $d' \in [d]$ and 
$e' \in [e]$ such that $(d',e') \in R^\Bmf$. 
Note that the mapping
$g: d\mapsto [d]$ defines an $\Sbf$-homomorphism from $\Bmf$ to $\Bmf^{c}$, 
which we call the \emph{canonical homomorphism}.

For a $\Sbf_{\cbf}$-structure $\Amf$, we define the $\Sbf_{P}$-structure $\hat{\Amf}$
which interprets the predicates in $\Sbf$ in the same way as $\Amf$ and interprets the predicates in $P$ as follows: $P_{i}^{\hat{\Amf}}=\{c_{i}^{\Amf}\}$. 
With every $\Sbf_{\cbf}$-structure $\Bmf$, one can associate a finite set of finite
$\Sbf_{P}$-structures, $\Bmf^{ac}$, called its \emph{anti-collapse}, 
such that the following
two properties hold:
\begin{enumerate}
\item for all $\Sbf_{P}$-structures $\Amf$:\\  $\Bmf \rightarrow \Amf^{c}$ (and $\Amf^{c}$ is defined)
if and only if there exists $\Bmf'\in \Bmf^{ac}$ such that $\Bmf' \rightarrow \Amf$.
\item for all $\Sbf_{\cbf}$-structures $\Amf$:\\ $\Bmf \rightarrow \Amf$ iff there exists
$\Bmf' \in \Bmf^{ac}$ such that $\Bmf' \rightarrow \hat{\Amf}$.
\end{enumerate}

To employ the anti-collapse $\Bmf^{ac}$ for the reduction of FPP$_{c}$ to FPP, 
we require some properties from the construction of $\Bmf^{ac}$ (cf. pages 43-45 of \cite{DBLP:journals/tods/AlexeCKT11}). The domain $\Delta^{\Bmf'}$ of each 
$\Bmf'\in \Bmf^{ac}$ consists of $\Delta^{\Bmf}\setminus \{c_{1}^{\Bmf},\ldots,c_{n}^{\Bmf}\}$ (the \emph{unnamed} individuals in $\Bmf$)
together with the union $\bigcup_{1\leq i \leq n}D_{i}$ of fresh non-empty 
(but possibly not mutually disjoint)
sets $D_{1},\ldots,D_{n}$ with $P_{i}^{\Bmf'}= D_{i}$. Moreover, in
Point~1 and Point~2 we have the following more detailed statement:
\begin{description}
\item[(1a)] if $h: \Bmf \rightarrow \Amf^{c}$ (and $\Amf^{c}$ is defined),
and $g: \Amf \rightarrow \Amf^{c}$ is the canonical homomorphism, 
then $h': \Bmf' \rightarrow \Amf$
can be chosen in such a way that $h'(d)\in g^{-1}(h(d))$ for all unnamed individuals $d$ in $\Bmf$
and $h'(d)\in g^{-1}(c_{i}^{\Amf^{c}})$ for all $d\in D_{i}$.
\item[(1b)] if $h: \Bmf' \rightarrow \Amf$, then $h': \Bmf \rightarrow \Amf^{c}$ can be defined
such that $h'(c_{i}^{\Bmf})= c_{i}^{\Amf^{c}}$ and $h'(d)=g(h(d))$ if $d$ is not named.
\item[(2b)] if $h:\Bmf' \rightarrow \hat{\Amf}$, then $h': \Bmf\rightarrow \Amf$
can be constructed in such a way that $h'(d)=h(d)$ for all unnamed $d$. 
\end{description}

In what follows, we will be interested in 
colorings of $\Sbf_{P}$-structures which 
respects the intuitive meaning of the predicates $P_{i}$. 
A $\Cmc$-coloring $\Bmf[\Cmc]$ of a $\Sbf_{P}$-structure $\Bmf$
is said to be a \emph{uniform $\Cmc$-coloring} of $\Bmf$ if for every $1\leq i\leq n$, 
$d,d'\in P_{i}^{\Bmf}$ implies that $d$ and $d'$ have the same color in $\Bmf[\Cmc]$. 
Given a set $\Gmc$ of $\Cmc$-colored $\Sbf_{P}$-structures,
we define $\mn{Forb}^{un}(\mathcal{G})$ as the set of $\Sbf_{P}$-structures $\Amf$
such that there exists a uniform $\Cmc$-coloring $\Amf[\Cmc]$ of $\Amf$
such that there exists no $\Gmf \in \mathcal{G}$ with $\Gmf \rightarrow \Amf[\Cmc]$. 



%
%

We are now ready to present the reduction from FPP$_{c}$ to FPP. 
Suppose that we are given a FPP$_{c}$ problem defined by the set 
$\mathcal{F}$ of $\Cmc$-colored $\Sbf_{\cbf}$-structures 
(where $\Cmc=\{T_{1},\ldots,T_{k}\}$). We construct a set 
$\Gmc$ which contains all uniform $\Cmc$-colored 
$\Sbf_{P}$-structures $\Gmf$ such that
\begin{itemize}
\item There exists $\Fmf\in \mathcal{F}$ and a member $\Fmf'$ of the anti-collapse of the $\Sbf_{\cbf}$-reduct of 
$\Fmf$ such that $\Gmf$ is the $\Cmc$-coloring of $\Fmf'$ defined as follows:\\[1mm]
($\dagger$) $d\in T_{j}^{\Gmf}$ iff $d$ is unnamed in $\Fmf$ and $d\in T_{j}^{\Fmf}$
or there exists $1\leq i \leq n$ such that $d\in D_{i}$ and $c_{i}^{\Fmf}\in T_{j}^{\Fmf}$.\\[1mm]
(Note that we require that in the resulting structure $T_{j}^{\Gmf}\cap T_{j'}^{\Gmf}=\emptyset$
for $j\not=j'$, otherwise $\Gmf$ is not in $\Gmc$).
\end{itemize}
It is easy to see that this construction guarantees that 
every $\Gmf \in \Gmc$ is such that $P_{i}^{\Gmf}\neq \emptyset$ for every $1 \leq i \leq n$.

We let $\Gmc_{u} = \Gmc \cup \Umc$, where $\Umc$ is the set of all $\Sbf_{P} \cup \Cmc$-structures
of the form $\{P_{i}(d), P_{i}(e), T_{j}(d), T_{\ell}(e)\}$ with $1 \leq i \leq n$ and
$1 \leq j < \ell \leq k$.  

Notice that $\mn{Forb}^{un}(\mathcal{G})= 
\forb(\mathcal{G}_{u})$.
 


\begin{lemma} \label{eval-fppc}
FPP$_{\cbf}$ is polynomially equivalent to FPP. Specifically:
\begin{itemize}
\item For all $\Sbf_{P}$-structures $\Amf$, 
$\Amf\in \forb(\mathcal{G}_{u})$ iff $\Amf^{c}$ is undefined or $\Amf^{c} \in \mn{Forb}(\Fmc)$;
\item For all $\Sbf_{\cbf}$-structures $\Amf$,
$\Amf \in \mn{Forb}(\Fmc)$ iff $\hat{\Amf}\in \forb(\mathcal{G}_{u})$.
\end{itemize}
\end{lemma} 
\begin{proof}
First let $\Amf$ be a $\Sbf_{P}$-structure such that 
$\Amf\in \forb(\mathcal{G}_{u})$.  Since $\forb(\mathcal{G}_{u})=\mn{Forb}^{un}(\mathcal{G})$, we 
have $\Amf \in\mn{Forb}^{un}(\mathcal{G})$, and so there exists a uniform $\Cmc$-colored
expansion $\Amf[\Cmc]$ of $\Amf$ such that 
there exists no $\Gmf\in \mathcal{G}$ with 
$\Gmf\rightarrow \Amf[\Cmc]$.
Assume the collapse $\Amf^{c}$ is defined (i.e., $P_{i}^{\Amf}\not=\emptyset$ for $1\leq i \leq n$).
We want to show $\Amf^{c} \in \mn{Forb}(\Fmc)$.
By uniformity of $\Amf[\Cmc]$, we obtain a $\Cmc$-colored $\Sbf_{\cbf}$-structure $\Amf^{c}[\Cmc]$ 
extending $\Amf^{c}$ by setting $d\in T_{j}^{\Amf^{c}[\Cmc]}$ iff $d$ is unnamed and 
$d\in T_{j}^{\Amf[\Cmc]}$ or $d=c_{i}^{\Amf^{c}}$ and
$P_{i}^{\Amf[\Cmc]}\subseteq T_{j}^{\Amf[\Cmc]}$. Assume for a 
contradiction that $h:\Fmf \rightarrow \Amf^{c}[\Cmc]$ for $\Fmf \in \Fmc$.
Then $h$ is a homomorphism from the $\Sbf_{\cbf}$-reduct $\Fmf^{r}$ of $\Fmf$ to the 
$\Sbf_{\cbf}$-reduct $\Amf^{c}$ of
$\Amf^{c}[\Cmc]$. By (1a), we find $\Fmf' \in (\Fmf^{r})^{ac}$ and $h': \Fmf'\rightarrow \Amf$
such that $h'(d)\in g^{-1}(h(d))$ for all unnamed individuals $d$ in $\Fmf^{r}$
and $h'(d)\in g^{-1}(c_{i}^{\Amf^{c}})$ for all $d\in D_{i}$.
Let $\Fmf'[\Cmc]$ be the $\Cmc$-coloring of $\Fmf'$ defined with ($\dagger$). 
To see that $\Fmf'[\Cmc]$ is well-defined, note that $d \in D_i \cap D_j$ implies that 
$P_i^{\Fmf'}\cap P_j^{\Fmf'}\neq \emptyset$, which yields 
$P_i^{\Amf}\cap P_j^{\Amf}\neq \emptyset$, hence $c_i^{\Amf^c} = c_j^{\Amf^c}$.
It follows that $c_i^{\Amf^c}$ and $c_j^{\Amf^c}$ have the same colour in $\Amc^c[\Cmc]$,
and thus also in $\Fmf$, which ensures that each element in $\Fmf'$ is assigned a unique 
colour by ($\dagger$). 
Now to obtain the desired contradiction, we show that $h'$
is a $\Sbf_{P}\cup \Cmc$-homomorphism from $\Fmf'[\Cmc]$ to $\Amf[\Cmc]$. Let $d\in \dom(\Fmf')$ 
and $d\in T_{j}^{\Fmf'[\Cmc]}$. If $d$ is unnamed in $\Fmf$, then
$d\in T_{j}^{\Fmf'[\Cmc]}$ implies that $d\in T_{j}^{\Fmf}$. 
Hence $h(d)\in T_{j}^{\Amf^{c}[\Cmc]}$ and
$h'(d) \in g^{-1}(h(d))\subseteq T_{j}^{\Amf[\Cmc]}$. 
If $d\in D_{i}$, then $d\in T_{j}^{\Fmf'[\Cmc]}$ implies $c_{i}^{\Fmf}\in T_{j}^{\Fmf}$, 
hence $c_{i}^{\Amf^{c}}\in T_{j}^{\Amf^{c}[\Cmc]}$ 
and $P_{i}^{\Amf[\Cmc]}\subseteq T_{j}^{\Amf[\Cmc]}$. 
From $h'(d)\in g^{-1}(c_{i}^{\Amf^{c}})$, we know that there exists a sequence $A_{\ell_{1}}, \ldots, A_{\ell_{p}}$ of predicates from 
$\{P_{1}, \ldots, P_{n}\}$ such that $h'(d) \in A_{\ell_{1}}^{\Amf[\Cmc]}$, $A_{\ell_{p}}= P_{i}$, and $A_{\ell_{k}}^{\Amf[\Cmc]}\cap A_{\ell_{k+1}}^{\Amf[\Cmc]}\neq \emptyset$ for every $1 \leq k \leq \ell_{p}$. By uniformity of $\Amf[\Cmc]$ and $P_{i}^{\Amf[\Cmc]}\subseteq T_{j}^{\Amf[\Cmc]}$, we obtain $A_{\ell_{1}}^{\Amf[\Cmc]} \subseteq T_{j}^{\Amf[\Cmc]}$, hence $h'(d) \in  T_{j}^{\Amf[\Cmc]}$. 

Conversely, if $\Amf^{c}$ is undefined, 
then $\Amf\in \mn{Forb}(\Gmc_{u})$ since
$P_{i}^{\Gmf}\not=\emptyset$ for all $\Gmf\in \Gmc$ and $1\leq i \leq n$,
and so any uniform $\Cmc$-coloring of $\Amf$ will avoid $\Gmf_{u}$.
Assume now that $\Amf^{c}\in \mn{Forb}(\Fmc)$. 
There exists a $\Cmc$-colored expansion 
$\Amf^{c}[\Cmc]$ of $\Amf^{c}$ such that there exists no 
$\Fmf\in \mathcal{F}$ with $\Fmf\rightarrow \Amf^{c}[\Cmc]$.
We define a (uniform) $\Cmc$-colored expansion 
$\Amf[\Cmc]$ of $\Amf$ in the obvious way; 
let $g: \Amf \rightarrow \Amf^{c}$
be the canonical mapping and set 
$T_{j}^{\Amf[\Cmc]}= g^{-1}(T_{j}^{\Amf^{c}[\Cmc]})$, for $1\leq j \leq k$.
Assume for a contradiction that $\Gmf\rightarrow \Amf[\Cmc]$ for $\Gmf\in \Gmc$. 
Then $\Gmf$ is obtained from some $\Fmf\in \mathcal{F}$ and some 
member $\Fmf'$ of the anti-collapse of the $\Sbf_{\cbf}$-reduct of $\Fmf$ as described in ($\dagger$).
Assume $h: \Gmf \rightarrow \Amf[\Cmc]$. Then $h:\Fmf' \rightarrow \Amf$ and so,
by (1b) there exists $h': \Fmf^{r} \rightarrow \Amf^{c}$ that can be defined
such that $h'(c_{i}^{\Fmf^r})= c_{i}^{\Amf^{c}}$ and $h'(d)=g(h(d))$ if $d$ is not named,
where $\Fmf^{r}$ is the $\Sbf_{\cbf}$-reduct of $\Fmf$.
We derive a contradiction by showing that $h'$ a homomorphism from $\Fmf$ to $\Amf^{c}[\Cmc]$.
First suppose that $d \in T_{j}^{\Fmf}$, and $d$ is unnamed in $\Fmf$.
Then $d \in T_{j}^{\Gmf}$, hence $h(d) \in T_{j}^{\Amf[\Cmc]}$. 
It follows from the definition of $T_{j}^{\Amf[\Cmc]}$ that $h'(d) =g(h(d))\in T_{j}^{\Amf^{c}[\Cmc]}$.
Next consider the case where $c_{i}^{\mathfrak{F}} \in T_{j}^{\Fmf}$.
Then there must exist $e$ such that $e \in T_{j}^{\Gmf}$ and $e \in P_{i}^{\Gmf}$.
It follows that $h(e) \in T_{j}^{\Amf[\Cmc]}$ and 
$h(e) \in P_{i}^{\Amf[\Cmc]}$. The definition of $T_{j}^{\Amf[\Cmc]}$ together with 
$g(h(e)) = c_{i}^{\Amf^{c}[\Cmc]}$ yields 
$h'(c_{i}^{\Fmf})= c_{i}^{\Amf^{c}[\Cmc]} \in T_{j}^{\Amf^{c}[\Cmc]}$.

The second statement follows easily from the first, 
since for every $\Sbf_{\cbf}$-structure $\Amf$, we have $\Amf = (\hat{\Amf})^{c}$. 
\end{proof}

\begin{lemma}\label{fppc-cont}
Containment of FPP$_{c}$ is polynomially reducible to
containment of FPP.  
\end{lemma} 
\begin{proof}
Consider $\mn{Forb}(\mathcal{F}_{1})$ and $\mn{Forb}(\mathcal{F}_{2})$, both over $\Sbf_{c}$. 
Let $\mathcal{G}_{u,1}$ and $\mathcal{G}_{u,2}$ be the corresponding FPPs over schema $\Sbf_{P}$,
which satisfy statements in  Lemma \ref{eval-fppc}. We claim that $\mn{Forb}(\mathcal{F}_{1}) \subseteq \mn{Forb}(\mathcal{F}_{2})$ iff
$\mn{Forb}(\mathcal{G}_{u,1}) \subseteq \mn{Forb}(\mathcal{G}_{u,2})$. 

For the first direction, suppose that $\mn{Forb}(\mathcal{F}_{1}) \subseteq \mn{Forb}(\mathcal{F}_{2})$. Let $\Amf$ be a $\Sigma_{P}$-structure such that 
$\Amf \in \mn{Forb}(\mathcal{G}_{u,1})$. 
If $\Amf^{c}$ is undefined, then we immediately obtain 
$\Amf \in \mn{Forb}(\mathcal{G}_{u,2})$. Otherwise, we have 
$\Amf^{c} \in \mn{Forb}(\mathcal{F}_{1})$, and hence  
$\Amf^{c} \in \mn{Forb}(\mathcal{F}_{2})$ 
and $\Amf \in \mn{Forb}(\mathcal{G}_{u,2})$. 

For the second direction, suppose that 
$\mn{Forb}(\mathcal{G}_{u,1}) \subseteq \mn{Forb}(\mathcal{G}_{u,2})$, and 
let $\Bmf$ be a $\Sbf_{c}$-structure such that 
$\Bmf \in \mn{Forb}(\mathcal{F}_{1})$. Then applying the previous lemma, 
we have $\hat{\Bmf} \in \mn{Forb}(\mathcal{G}_{u,1})$, hence 
$\hat{\Bmf} \in \mn{Forb}(\mathcal{G}_{u,2})$. Again applying the lemma, 
we obtain $\Bmf \in \mn{Forb}(\mathcal{F}_{2})$. \end{proof}

By combining in a straightforward manner Lemmas \ref{eval-activedomain}
to \ref{fppc-cont}, we obtain Theorem \ref{lift-commsnp}. 


%

\subsection{Proofs for Section~\ref{sect:obdammsnp}}
\label{subsec:b2}

\medskip
\noindent
{\bf Theorem~\ref{dich-ucq}. }
\emph{  (\ALC,UCQ) has a dichotomy between \PTime and \coNP iff the
  Feder-Vardi conjecture holds. The same is true for
  ($\mathcal{ALCHIU}$,UCQ) and (UNFO,UCQ).}
\medskip

\noindent
\begin{proof}
Easily obtained by combining  Proposition  \ref{prop:mmsnpmdd} and Theorems \ref{thm:ALCtoMDD}, \ref{thm:otherDLUCQ}, \ref{thm:UNFO}, and \ref{lift-commsnp}. 
\end{proof}

\medskip
\noindent
{\bf Theorem~\ref{thm:containment1}. }
\emph{Query containment is decidable for 
the OBDA languages (\ALC,UCQ), ($\mathcal{ALCHIU}$,UCQ), and (UNFO,UCQ).}

\medskip

\noindent
\begin{proof}
Here again we straightforwardly combine Proposition \ref{prop:mmsnpmdd} and Theorems \ref{thm:ALCtoMDD}, \ref{thm:otherDLUCQ}, \ref{thm:UNFO}, and \ref{lift-commsnp}
\end{proof}

\medskip
\noindent
{\bf Theorem~\ref{lem:gmsnpfgdd}.}
  \emph{coGMSNP has the same expressive power as frontier-guarded \DD and
  is  strictly more expressive than coMMSNP.}
\medskip

\noindent
\begin{proof}
  The proof of the first part follows the lines of the proof of
  Proposition~\ref{prop:mmsnpmdd} and is omitted. It thus remains to show
  that coGMSNP is strictly more expressive than coMMSNP.  Note first
  that it is at least as expressive: we can convert any MMSNP formula
  into an equivalent one satisfying conditions~(i) and~(ii) from the
  proof of Proposition~\ref{prop:mmsnpmdd}, and clearly every such MMSNP
  formula is also a GMSNP formula. To see that coGMSNP is indeed
  strictly more expressive than coMMSNP, note that by
  Proposition~\ref{prop:gfucqmdd}, there is a (GF,UCQ) query $q$ that
  is not expressible in \MDD. 
  By Proposition~\ref{prop:mmsnpmdd},
   $q$ is not expressible in coMMSNP; by
  Theorem~\ref{GFUCQfrontier} and the first part of
  Theorem~\ref{lem:gmsnpfgdd}, $q$ is expressible in coGMSNP.
\end{proof}

\medskip
\noindent
{\bf Proposition
\ref{prop:gmmm2}}
 \emph{ GMSNP and MMSNP$_2$ have the same expressive power.}\\
 
\noindent\begin{proof}
For simplicity, we prove the result for sentences (no free variables) and without
equality in the body of implications.

We start by proving that every MMSNP$_{2}$ sentence is equivalent to a GMSNP sentence.
Assume $\Phi=\exists X_1 \cdots \exists X_n \forall x_1 \cdots \forall x_m \vp$
is a MMSNP$_{2}$ sentence. Introduce for each $X_{i}$ a monadic SO-variable $X_{i}^{1}$ and, for every $R\in \Sbf$ of
arity $n$, an $n$-ary SO-variable $X_{i}^{R}$. Now replace in $\vp$ every $X_{i}(x)$ by $X_{i}^{1}(x)$ and every
$X_{i}(R(\xbf))$ by $X_{i}^{R}(\xbf)$.
The resulting formula is a GMSNP sentence that is equivalent to $\Phi$.

\medskip

Conversely, assume we are given a GMSNP sentence $\Phi=\exists X_1 \cdots \exists X_n \forall x_1 \cdots \forall x_m \vp$. 
It is straightforward to show that $\Phi$ is equivalent to a GMSNP sentence in which 
\begin{itemize}
\item each $X_{i}(\xbf)$ in the head of an implication is guarded by an input relation: for every $X_{i}(\xbf)$ in 
the head of an implication
$\psi$ there exists an $R\in \Sbf$ such that $R(\ybf)$ is in the body of $\psi$ and $\xbf \subseteq \ybf$.
(If this is not the case, one can introduce additional conjuncts $R(\ybf)$ in the body of implications).
%
%
\item $\vp$ is closed under identifying individual variables: if $\psi'$ is the result of identifying variables
in an implication $\psi$ of $\vp$, then $\psi$ is a conjunct of $\vp$ (module renaming of individual variables).
\item the individual variables used in distinct implications of $\vp$ are disjoint. 
\end{itemize}
It follows that we may also assume that distinct occurrences of SO-variables $X_{i}$ in $\vp$
determine distinct atoms $X_{i}(\xbf_{i})$. From now we assume that $\Phi$ satisfies these conditions. 

For the translation, we take for every atom $A=X_i(\xbf)$ in the head of an implication $\psi$ in $\vp$, 
a fresh second-order domain and fact variable $X_{A}$. Moreover, we fix a guard $R_{A}(\ybf_A)$ with $R_{A}\in \Sbf$ 
for $A$ from the body of the (unique) implication in which $A$ occurs.
Consider now an implication $\psi$ in $\vp$ of the form
\begin{eqnarray*}
&   & R_{1}(\xbf_{1}) \wedge \cdots \wedge R_{k}(\xbf_{k}) \wedge X_{k+1}(\xbf_{k+1}) \wedge \cdots \wedge 
X_{n}(\xbf_{n})\\
& \rightarrow & X_{n+1}(\xbf_{n+1})\vee \cdots \vee X_{m}(\xbf_{m})
\end{eqnarray*}
First replace all atoms $A_{j}=X_{j}(\xbf_{j})$, $n+1\leq j \leq m$, by $X_{A_{j}}(R_{A_j}(\ybf_{A_j}))$, where 
$R_{A_j}(\ybf_{A_j})$
is the guard 
 for $A_{j}$ selected above.
Next consider every possible choice 
$$
A_{k+1}=X_{k+1}(\zbf_{k+1}),\ldots,A_{n}=X_{n}(\zbf_{n})
$$ 
of atoms in the heads of implications in $\vp$ such that the componentwise mappings $\rho_{l}:\xbf_{l} \rightarrow \zbf_{l}$, $k+1\leq l \leq n$, are bijections between the sets of variables in $\xbf_{l}$ and $\zbf_{l}$ and replace every 
$X_{l}(\xbf_{l})$, $k+1\leq l \leq n$, by 
$$
X_{A_{l}}(R_{A_{l}}(\ybf'_l))
$$
where $\ybf'_l$ is obtained from the guard $R_{A_{l}}(\ybf_{A_l})$ associated with $A_{l}$ 
above by replacing each $\rho_{l}(x)$ by $x$ and each individual variable that is not in the
range of $\rho_{l}$ by some fresh individual variable.
Let $\psi'$ be the conjunction over all implications derived from $\psi$ in this manner, let
$\vp'$ be the conjunction of all of the  $\psi'$, and let 
$\Phi'$ be the resulting MMSNP$_{2}$ sentence when existential quantification 
over non-monadic variables is replaced 
by existential quantification over all $X_{A}$ such that $A$ an atom in a head of an implication of $\vp$.
Note that $\Phi'$ contains all individual variables in $\Phi$, but may also contain additional individual
variables not in $\Phi$.

\medskip

We show that $\Phi$ and $\Phi'$ are equivalent. Assume first that $(\adom(\Dmf),\Dmf)\models \Phi'$. Take an assignment $\pi$ for the second-order domain and fact
variables of $\Phi'$ such that $(\adom(\Dmf),\Dmf)\models_{\pi} \forall x_1 \cdots \forall x_m \vp'$.
For every non-monadic second-order variable $X$ of $\Phi$, define $\pi(X)$ as the union of all 
$$
\{ \rho(\xbf) \mid R_{A}(\rho(\ybf_A)) \in \pi(X_{A}), \mbox{ $\rho$ injective variable assignment}\},
$$
such that $A=X(\xbf)$ appears in the head of some implication in $\vp$ and $R_{A}(\ybf_A)$ is the guard selected
for $A$. We show that $(\adom(\Dmf),\Dmf)\models_{\pi} \Phi$. 
Assume for a contradiction that this is not the case. Take an implication $\psi$ in $\vp$ of the form
\begin{eqnarray*}
&   & R_{1}(\xbf_{1}) \wedge \cdots \wedge R_{k}(\xbf_{k}) \wedge X_{k+1}(\xbf_{k+1}) \wedge \cdots \wedge 
X_{n}(\xbf_{n})\\
& \rightarrow & X_{n+1}(\xbf_{n+1})\vee \cdots \vee X_{m}(\xbf_{m})
\end{eqnarray*}
and let $\rho$ be an individual variable assignment such that $(\adom(\Dmf),\Dmf)\not\models_{\pi,\rho} \psi$.
We may assume that $\rho$ is injective. The following holds:
\begin{enumerate}
\item for every $1\leq i \leq k$, we have $R_i(\pi(\xbf_i)) \in \Dmf$.
\item for every $k+1\leq i \leq n$, there exists $A_{i}=X_i(\zbf_{i})$ in the head of some implication of $\varphi$
with $R_{A_{i}}(\zbf_{i}')$ the guard selected for $A_{i}$, and an injective variable assignment $\rho_{i}$ such that
$R_{A_{i}}(\rho_{i}(\zbf_{i}'))\in \pi(X_{A_{i}})$ and $\rho_{i}(\zbf_i)=\rho(\xbf_{i})\in \pi(X_{i})$.
\item for no $n+1\leq i \leq m$ does there exist $A_{i}=X_i(\zbf_{i})$ in the head of some implication of $\varphi$
with $R_{A_{i}}(\zbf_{i}')$ the guard selected for $A_{i}$, and an injective variable assignment $\rho'$ such that
$R_{A_{i}}(\rho'(\zbf_{i}'))\in \pi(X_{A_{i}})$ and $\rho'(\zbf_i)=\rho(\xbf_{i})\in \pi(X_{i})$.   
\end{enumerate}
Consider the following sequences of atoms
\begin{align*}
&A_{k+1}=X_{k+1}(\zbf_{k+1}),\ldots,A_{n}=X_{n}(\zbf_{n})\\
& A_{n+1}=X_{n+1}(\xbf_{n+1}), \ldots, A_m=X_{m}(\xbf_{m})
\end{align*}
It follows from construction of $\Phi'$ that the formula $\varphi'$ contains the implication 
\begin{eqnarray*}
\zeta=&   & R_{1}(\xbf_{1}) \wedge \cdots \wedge R_{k}(\xbf_{k}) \wedge 
\\
&   & X_{A_{k+1}}(R_{A_{k+1}}(\ybf'_{k+1}))\wedge \cdots \wedge 
X_{A_{n}}(R_{A_{n}}(\ybf'_{n}))\\
& \rightarrow & X_{A_{n+1}}(R_{A_n+1}(\ybf_{A_n+1}))\vee \cdots \vee 
X_{A_{m}}(R_{A_m}(\ybf_{A_m}))
\end{eqnarray*}
where the $\ybf'_i$ are defined in the same way as earlier. 
Let $\mu$ be an individual variable assignment satisfying:
\begin{itemize}
\item $\mu(x)=\rho(x)$ for $x$ in the image of $\rho$
\item $\mu(u)=\rho_i(z)$ if $u$ is the fresh variable introduced to 
replace $z \in \zbf'_i$ 
\end{itemize}
Note that such an assignment must exist 
since every variable in $\Phi'$ is in the image of exactly 
one assignment among $\rho$ and the $\rho_i$.
It follows from the properties of $\mu$ and 
points 1 and 2 above that the body of the implication $\zeta$
is satisfied under assignments $\pi, \mu$. 
From point 3, we can derive that none of the head atoms is 
satisfied under $\pi, \mu$. 
It follows that the implication $\zeta$ is refuted, so 
$(\adom(\Dmf),\Dmf)\not\models \Phi'$, and we have the desired contradiction.

\medskip
For the other direction, assume that $(\adom(\Dmf),\Dmf)\models \Phi$. Take an assignment $\pi$ for the SO-variables
of $\Phi$ such that $(\adom(\Dmf),\Dmf)\models_{\pi} \forall x_1 \cdots \forall x_{m}\,\, \vp$.
Now define, for $A=X(\xbf)$ in the head of an implication of $\vp$ with selected guard $R_{A}(\ybf_A)$:
\begin{align*}
\pi(X_{A})&=\{ R_{A}(\rho(\ybf_A)) \in \Dmf \mid  \rho(\xbf) \in \pi(X),\\
&\hspace*{3.1cm} \mbox{$\rho$ variable assignment} \}
\end{align*}
It can be verified that $(\adom(\Dmf),\Dmf)\models \Phi'$.
\end{proof}

\savebox{\stupid}{\large\bf\ref{sect:OBDAtoCSP}}
\section{Proofs for Section \usebox{\stupid}}
\noindent
{\bf Theorem~\ref{ALCUtoCSP}}  
\emph{In each case, the following query languages are equally expressive:}
\begin{itemize}
\item $(\mathcal{ALCU}$,AQ), ($\mathcal{SHIU}$,AQ), unary simple \MDD, and 
generalized coCSP with one constant symbol;
\item ($\mathcal{ALC}$,AQ), ($\mathcal{SHI}$,AQ), unary connected simple \MDD, and
generalized coCSPs with one constant symbol such that all templates are identical except 
for the interpretation of the constant symbol;
\item ($\mathcal{ALCU}$,BAQ), ($\mathcal{SHIU}$,BAQ), Boolean simple \MDD, and generalized coCSP;
\item ($\mathcal{ALC}$,BAQ), ($\mathcal{SHI}$,BAQ), Boolean connected simple \MDD, and coCSP.

\end{itemize}
\emph{Moreover, given the ontology-mediated query or monadic datalog program,
the correponding CSP template is of at most exponential size and can be 
constructed in time polynomial in the size of the template.}

\medskip

\noindent\begin{proof} 
Recall that the equivalences between the OBDA languages and fragments of 
monadic disjunctive datalog have been proved already. Moreover, Point~1 has been
proved in the paper. It thus remains to be proved that the following query languages
are equally expressive:

(a) ($\mathcal{ALC}$,AQ) and
generalized coCSPs with one constant symbol such that all templates are identical except 
for the interpretation of the constant symbol;

(b) ($\mathcal{ALC}$,BAQ) and coCSP;

(c) ($\mathcal{ALCU}$,BAQ) and generalized coCSP.
  
We use the notation from the proof of Point~1. In particular, $\Bmf_{T}$ denotes the canonical 
$\Sbf$-structure with domain $T$.
For (a), assume $\Sbf$, $\Omc$, and $A(x)$ are given, where $\Omc$ is an $\mathcal{ALC}$-ontology.
Let $T$ be the set of all types $\tau$ that are realizable for 
$\Omc$ and define 
$$
\mathcal{F}= \{ (\Bmf_{T},\tau) \mid \tau\in T,A\not\in \tau\}.
$$
One can show that for every $\Sbf$-instance $\Dmf$ and $d\in \adom(\Dmf)$:
$(\Dmf,d) \rightarrow (\Bmf_{T},\tau)$ for some $(\Bmf_{T},\tau) \in \mathcal{F}$
iff $d\not\in q_{\Sbf,\Omc,A(x)}(\Dmf)$. Thus, the query defined by
$(\Sbf,\Omc,A(x))$ is equivalent to the query defined by $\mathcal{F}$.

Conversely, assume that $\mathcal{F}$ is a finite set of $\Sbf\cup\{c\}$-structures which coincide
except for the interpretation of the constant symbol~$c$, and let $\Bmf$ be the $\Sbf$-reduct of 
these structures. 
Take for every $d$ in the domain $\dom(\Bmf)$ of $\Bmf$
a fresh concept name $A_{d}$, let $A$ be another fresh concept name, and set
\begin{eqnarray*}
\Omc & = & \{ A_{d} \sqsubseteq \neg A_{d'} \mid d\not=d'\}\cup\\
     &   & \{ A_{d} \sqcap \exists R.A_{d'} \sqsubseteq \bot \mid R(d,d')\not\in \Bmf,R\in \Sbf\}\cup\\
     &   & \{ A_{d} \sqcap B \sqsubseteq \bot \mid B(d)\not\in \Bmf,B\in \Sbf\}\cup\\
     &   & \{ \top \sqsubseteq \bigsqcup_{d\in \dom(\Bmf)}A_{d}\}\cup\\
     &   & \{ \bigsqcap_{(\Bmf,b)\in \mathcal{F}}\neg A_{b} \sqsubseteq A\}
\end{eqnarray*}
One can show that for every $\Sbf$-instance $\Dmf$ and $d\in \adom(\Dmf)$, 
$(\Dmf,d)\rightarrow (\Bmf,b)$ for some $(\Bmf,b)\in \mathcal{F}$  iff $d\not\in q_{\Sbf,\Omc,A(x)}(\Dmf)$.
Thus $(\Sbf,\Omc,A(x))$ expresses the same query as $\mathcal{F}$.

\medskip

For (b) assume that a query $(\Sbf,\Omc, \exists x. A(x)) \in
\text{(\ALC,BAQ)}$ is given.  We assume w.l.o.g.\ that $\Omc
\not\models \top \sqsubseteq \exists U.A$ because otherwise we have
$q_Q(\Dmf)=1$ for all $\Sbf$-instances $\Dmf$, and so $q_Q$ is trivial.
Let $T$ be the set of all types $\tau\subseteq \mn{sub}(\Omc)$ that are realized in a model $\Amf$ of 
$\Omc$ with $\Amf\not\models \exists x.A(x)$.
Since $\Omc
\not\models \top \sqsubseteq \exists U.A$, the set $T$ is
non-empty.
One can show that for every $\Sbf$-instance $\Dmf$:
$\Dmf \rightarrow \Bmf_{T}$
iff $Q_{\Sbf,\Omc,\exists x.A(x)}(\Dmf)=0$. Thus, the query defined by
$(\Sbf,\Omc,\exists x.A(x))$ is equivalent to the query defined by $\Bmf_{T}$.

Conversely, for a CSP template $\Bmc$ over schema \Sbf, we construct
an ontology-mediated query $(\Sbf,\Omc,q)$ as follows. 
Take for every $d$ in the domain $\dom(\Bmc)$ of $\Bmc$
a fresh concept name $A_{d}$, let $A$ be another fresh concept name, and set
$q=\exists x.A(x)$ and 
%
\begin{eqnarray*}
\Omc & = & \{ A_{d} \sqcap A_{d'}\sqsubseteq A \mid d\not=d'\}\cup\\
     &   & \{ A_{d} \sqcap \exists R.A_{d'} \sqsubseteq A \mid R(d,d')\not\in \Bmf,R\in \Sbf\}\cup\\
     &   & \{ A_{d} \sqcap B \sqsubseteq A \mid B(d)\not\in \Bmf,B\in \Sbf\}\cup\\
     &   & \{ \top \sqsubseteq \midsqcup_{d\in \dom(\Bmf)}A_{d}\}
\end{eqnarray*}
The query $(\Sbf,\Omc,\exists x.A(x))$ is equivalent to the query 
defined by the template $\Bmf$.

\medskip

The proof of Point~(c) is similar and left to the reader.

\end{proof}

\medskip

\noindent
{\bf Theorem~\ref{ContNEXP}}
\emph{Query containment in ($\mathcal{SHIU}$,AQ$\cup$BQ) is in \NExpTime.
It is \NExpTime-hard already for ($\mathcal{ALC},AQ)$ and for
$(\mathcal{ALC}$,BAQ).}

\medskip

\noindent
\begin{proof}
We provide the proof of the lower bound. 
The proof is by reduction of a \NExpTime-hard $2^{n}\times 2^{n}$-tiling problem.
An instance of this tiling problem is given by a natural number $n>0$ and a 
triple $(\Tmf,H,V)$ with $\Tmf$ a non-empty, finite set of \emph{tile
types} including an \emph{initial tile} $T_\mn{init}$ to be placed
on the lower left corner, $H \subseteq \Tmf \times \Tmf$ a
\emph{horizontal matching relation}, and $V \subseteq \Tmf \times
\Tmf$ a \emph{vertical matching relation}. A \emph{solution} for the 
$2^{n}\times 2^{n}$-tiling problem for
$(\Tmf,H,V)$ is a map $f:\{0,\dots,2^{n}-1\} \times \{0,\dots,2^{n}-1\}
\rightarrow \Tmf$ such that $f(0,0)=T_\mn{init}$, $(f(i,j),f(i+1,j)) \in H$ for all 
$i < 2^{n}-1$, and
$(f(i,j),f(i,j+1)) \in V$ for all $j < 2^{n}-1$. It is \NExpTime-complete to decide
whether an instance of the $2^{n}\times 2^{n}$-tiling problem has a solution.

For the reduction, let $n>0$ and $(\Tmf,H,V)$ be an instance of the $2^{n}\times 2^{n}$-tiling
problem with $\Tmf=\{T_1,\dots,T_p\}$. 
We construct a schema $\Sbf$, two $\ALC$-ontologies $\Omc_{1}$ and $\Omc_{2}$, and a query $E(x)$ 
with $E$ a unary relation symbol 
such that $(\Tmf,H,V)$ has a solution if
and only if $q_{\Sbf,\Omc_{1},E(x)} \subseteq q_{\Sbf,\Omc_{2},E(x)}$ if and only if
$q_{\Sbf,\Omc_{1},\exists x.E(x)} \subseteq q_{\Sbf,\Omc_{2},\exists x.E(x)}$.

We first define an ontology $\Gmc$ (for grid) which encodes the $2^{n}\times 2^{n}$-grid.
To define $\Gmc$, we use role names $x$ and $y$ to represent 
the $2^{n}\times 2^{n}$-grid and two binary counters $X$ and $Y$ for counting from $0$ to $2^{n}-1$.
The counters use concept names $X_0,\dots,X_{n-1},\overline{X}_{0},\ldots,\overline{X}_{n-1}$
and $Y_0,\dots,Y_{n-1},\overline{Y}_{0},\ldots,\overline{Y}_{n-1}$ as their bits, respectively. 

$\Gmc$ contains the inclusions
$$
\overline{X}_{i} \sqsubseteq \neg X_{i}, \quad \overline{Y}_{i} \sqsubseteq \neg Y_{i},
$$
for $i=0,\ldots,n-1$. Counters are relevant only if the concept
$$
{\sf Def}= (\bigsqcap_{0=1..n-1}(X_{i} \sqcup \overline{X}_{i})) \sqcap 
(\bigsqcap_{0=1..n-1}(Y_{i} \sqcup \overline{Y}_{i}))
$$
is true. $\Gmc$ contains the following well-known inclusions stating that the value
of the counter $X$ is incremented when going to
$x$-successors (and ${\sf Def}$ is true) and the value of the counter $Y$ is incremented when 
going to $y$-successors (and ${\sf Def}$ is true): for $k= 0,\ldots,n-1$,
$$
{\sf Def} \sqcap \bigsqcap_{j=0..k-1} X_j 
\sqsubseteq P_{k}
$$
where
$$
P_{k}= (X_k \rightarrow \forall x.({\sf Def} \rightarrow \overline{X}_k )) \sqcap 
        (\overline{X}_k \rightarrow \forall x.({\sf Def} \rightarrow X_k )  )  
$$
and 
$$
{\sf Def} \sqcap \bigsqcup_{j=0..k-1} \overline{X}_j 
\sqsubseteq Q_{k}
$$
where
$$
Q_{k}=
(X_k \rightarrow \forall x. ({\sf Def} \rightarrow X_k )) \sqcap 
        (\overline{X}_k \rightarrow \forall x. ({\sf Def} \rightarrow \overline{X}_k )) 
$$
and similarly for $Y$ and $y$. $\Gmc$ also states that
the value of the counter $X$ does not change when going to $y$-successors 
and the value of the counter $Y$ does not change when going to  $x$-successors:
for $i=0,\ldots,n-1$,
$$
{\sf Def} \sqcap X_{i} \sqsubseteq \forall y.({\sf Def}\rightarrow X_{i}), \quad 
{\sf Def} \sqcap \overline{X}_{i} \sqsubseteq \forall y.({\sf Def} \rightarrow \overline{X}_{i})
$$
and similarly for $Y$ and $x$.
In addition, $\Gmc$ states that when the counter $X$ is $2^{n}-1$, there is no $x$-successor
(with ${\sf Def}$) and if the counter $Y$ is $2^{n}-1$, there is no $y$-successor (with ${\sf Def}$):
$$
{\sf Def} \sqcap X_{0}\sqcap \cdots \sqcap X_{n-1} \sqsubseteq \forall x.({\sf Def} \rightarrow \bot)
$$
and
$$
{\sf Def} \sqcap Y_{0}\sqcap \cdots \sqcap Y_{n-1} \sqsubseteq \forall y.({\sf Def} \rightarrow \bot)
$$
This finishes the definition of $\Gmc$. Define the schema
\begin{eqnarray*}
\Sbf_{\Gmc} & = & \{x,y,X_{0},\ldots,X_{n-1},\overline{X}_{0},\ldots,\overline{X}_{n-1}\}\cup\\
     &   & \{Y_{0},\ldots,Y_{n-1},\overline{Y}_{0},\ldots,\overline{Y}_{n-1}\}.
\end{eqnarray*}
We set $\Omc_{2}=\Gmc \cup \{E \sqsubseteq E\}$ (the latter inclusion merely serves to
ensure $E$ is part of the schema of $\Omc_{2}$).

\medskip

We now extend $\Gmc$ to another ontology $\Gmc^{t}$. In addition to the inclusions in
$\Gmc$, $\Gmc^{t}$ states that $T_{\mn{init}}$ holds at $(0,0)$:
$$
\neg X_{0}\sqcap \cdots \sqcap \neg X_{n-1} \sqcap \neg Y_{0} \sqcap \cdots \sqcap \neg Y_{n-1} 
\sqsubseteq T_{\mn{init}}
$$
and that the tiling is complete on ${\sf Def}$:
$$
{\sf Def} \sqsubseteq \bigsqcup_{i=1..p}T_{i},
$$
Next, $\Gmc^{t}$ states that if a tiling condition is violated, then a concept name $E$ is true.
For all $i\not=j$:
$$
T_{i} \sqcap T_{j} \sqsubseteq E,
$$
for all $(i,j)\not\in H$: 
$$
T_{i} \sqcap \exists x.T_{j} \sqsubseteq E,
$$
and for all $(i,j)\not\in V$:
$$
T_{i} \sqcap \exists y.T_{j} \sqsubseteq E.
$$
Finally, $E$ is propagated along $x$ and $y$:
$$
\exists x.E \sqsubseteq E, \quad \exists y.E \sqsubseteq E
$$ 
We set $\Omc_{1}=\Gmc^{t}$ and show:

\medskip
\noindent
\emph{Claim.} The following conditions are equivalent:
\begin{enumerate}
\item the $2^{n}\times 2^{n}$-tiling problem for $(\Tmf,H,V)$ has no solution;
\item $q_{\Sbf_{\Gmc},\Omc_{1},E(x)}$ is not contained in $q_{\Sbf_{\Gmc},\Omc_{2},E(x)}$;
\item $q_{\Sbf_{\Gmc},\Omc_{1},\exists x.E(x)}$ is not contained in $q_{\Sbf_{\Gmc},\Omc_{2},\exists x.E(x)}$.
\end{enumerate}
Assume first that $(\Tmf,H,V)$ admits no $2^{n}\times 2^{n}$-tiling. 
Define a $\Sbf_{\Gmc}$-instance $\Dmf_{\Gmc}$ as follows.
We regard the pairs
$(i,j)$ with $i\leq 2^{n}-1$ and $j\leq 2^{n}-1$ as constants and let 
\begin{itemize}
\item $x((i,j),(i+1,j))\in \Dmf_{\Gmc}$
for $i<2^{n}-1$ and 
\item $y((i,j),(i,j+1))\in \Dmf_{\Gmc}$ for $j<2^{n}-1$. 
\end{itemize}
We also set
\begin{itemize}
\item $X_{k}(i,j)\in \Dmf_{\Gmc}$ if the $k$th bit of $i$ is $1$,
\item $\overline{X}_{k}(i,j)\in \Dmf_{\Gmc}$ if the $k$th bit of $i$ is $0$,
\item $Y_{k}(i,j)\in \Dmf_{\Gmc}$ if the $k$th bit of $j$ is $1$, and 
\item $\overline{Y}_{k}(i,j)\in \Dmf_{\Gmc}$ if the $k$th bit of $j$ is $0$. 
\end{itemize}
Then
\begin{itemize}
\item $q_{\Sbf_{\Gmc},\Omc_{2},E(x)}(\Dmf_{\Gmc})=\emptyset$ and
\item $q_{\Sbf_{\Gmc},\Omc_{2},\exists x.E(x)}(\Dmf_{\Gmc})=0$
\end{itemize}
since $\Dmf_{\Gmc}$ counts correctly, and hence is satisfiable w.r.t. $\Omc_{2}$.
However, since $(\Tmf,H,V)$ admits no $2^{n}\times 2^{n}$-tiling, it follows that 
\begin{itemize}
\item $(0,0)\in q_{\Sbf_{\Gmc},\Omc_{1},E(x)}(\Dmf_{\Gmc})$;
\item $q_{\Sbf_{\Gmc},\Omc_{1},\exists x.E(x)}(\Dmf_{\Gmc})=1$.
\end{itemize}
We have proved Points~2 and 3.

\medskip
 
Conversely, assume that $(\Tmf,H,V)$ admits a $2^{n}\times 2^{n}$-tiling 
given by $f:\{0,\dots,2^{n}-1\} \times \{0,\dots,2^{n}-1\}\rightarrow \Tmf$.
We show that $q_{\Sbf_{\Gmc},\Omc_{1},\exists x.E(x)}(\Dmf)=0$ for all $\Sbf_{\Gmc}$-instances
$\Dmf$ which are satisfiable w.r.t.~$\Omc_{2}$. Then Points~2 and 3 are refuted, as required.

Assume $\Dmf$ is satisfiable w.r.t.~$\Omc_{2}$. We define a model $(\dom,\Dmf')$ of $\Omc_{1}$ with 
$\Dmf'\supseteq \Dmf$ as follows: the domain of $\Dmf'$ coincides with $\adom(\Dmf)$.
Symbols from $\Sbf_{\Gmc}$ are defined in $\Dmf'$ in exactly the same way
as in $\Dmf$. To define the facts involving tile types $T_{k}$ associate with every $d\in \adom(\Dmf)$ 
such that ${\sf Def}$ applies to $d$, the uniquely determined pair $v(d)=(i,j)$ given to the values of the counters 
$X$ and $Y$ by ${\sf Def}$. Then set $T_{k}(d)\in \Dmf'$ iff $f(v(d)) = T_{k}$.
Note that $\Dmf'$ contains no facts involving $E$.
It is readily checked that the resulting structure is a model of $\Omc_{1}$.
\end{proof}

\medskip
\noindent
{\bf Proposition~\ref{prof:FOtoDLog}.} 
\emph{If $Q = (\Sbf,\Omc,q)$ is an ontology-mediated
  query with $\Omc$ formulated in equality-free FO and $q$ a UCQ, then
  $q_Q$ is preserved by homomorphisms. Consequently, it follows from
  \cite{DBLP:journals/jacm/Rossman08} that if $q_Q$ is FO-rewritable,
  then $q_Q$ is rewritable into a UCQ (thus into datalog).}

\medskip
\noindent
\begin{proof}
  Let $h:\Dmf_1\to\Dmf_2$ be a homomorphism, and $\abf$ a tuple from
  $\adom(\Dmf_1)$ such that $\abf\in q_Q(\Dmf_1)$. Furthermore,
  suppose for the sake of contradiction that $h(\abf)\not\in
  q_Q(\Dmf_2)$. Then there is a finite relational structure
  $(\mn{dom}_2,\Dmf'_2)\models\Omc$ such that $\Dmf_2\subseteq
  \Dmf'_2$ and $h(\abf)\not\in q(\Dmf'_2)$. Let $(\mn{dom}_1,\Dmf'_1)$
  be the inverse image of $(\mn{dom}_2,\Dmf'_2)$ under $h$. More
  precisely, $\mn{dom}_1 = \adom(\Dmf_1)\cup
  (\mn{dom}_2\setminus\adom(\Dmf_2))$, and $\Dmf'_1$ contains all
  facts whose $\widehat{h}$-image is a fact of $\Dmf'_2$ where
  $\widehat{h}$ is the map that extends $h$ by sending every element
  of $\adom(\Dmf'_2)\setminus\adom(\Dmf_2)$ to itself.  Clearly,
  $\Dmf_1\subseteq\Dmf'_1$. Furthermore, $\abf\not\in q(\Dmf'_1)$
  because $\widehat{h}:\Dmf'_1\to\Dmf'_2$ is a homomorphism and $q$ is
  preserved by homomorphisms. To obtain a contradiction against
  $\abf\in q_Q(\Dmf_1)$, it therefore only remains to show that
  $(\mn{dom}_1,\Dmf'_1)\models\Omc$.  It is known that equality-free
  first-order sentences are preserved by passing from a structure to
  its quotient under an equivalence relation that is a congruence. By
  construction, the kernel of the map $\widehat{h}$ is a congruence
  relation on the structure $(\mn{dom}_1,\Dmf'_1)$ and its quotient is
  isomorphic to $(\mn{dom}_2,\Dmf'_2)$.
\end{proof}

The following lemma reduces the problem of deciding FO-rewritability from
generalized CSP with constants to generalized CSP without constants.
\begin{lemma}\label{lem:constants}
Let $\mathcal{F}$ be a finite set of $\Sbf\cup \cbf$-structures. The following
conditions are equivalent:
\begin{enumerate}
\item coCSP($\mathcal{F}$) is FO-definable;
\item coCSP($\mathcal{F}^{c})$ is FO-definable;
\end{enumerate}
\end{lemma}
\begin{proof}
If coCSP($\mathcal{F}^{c}$) is defined by a first-order sentence $\varphi$,
then replacing every subformula of the form $P_i(x)$ in $\varphi$ by
$x=c_i$ yields a first-order sentence defining coCSP($\mathcal{F}$).

For the converse, we make use 
a characterization of FO-definability of generalized coCSPs with constants
using finite obstruction sets.
Let $\mathcal{F}$ be a finite set of $\Sbf\cup \cbf$-structures.
A set $\mathcal{D}$ of $\Sbf\cup\cbf$-structures
is an \emph{obstruction set for CSP($\mathcal{F}$)} if for all $\Sbf \cup \cbf$-structures $\Dmf$
the following conditions are equivalent:
\begin{itemize}
\item there exists $\Bmf\in \mathcal{F}$ such that $\Dmf \rightarrow \Bmf$;
\item there does not exist $\Amf\in \mathcal{D}$ such that $\Amf\rightarrow \Dmf$.
\end{itemize}
It is known that, for any finite set of structures $\mathcal{F}$,
coCSP($\mathcal{F}$) is FO-definable if and only if $\mathcal{F}$ has
a finite obstruction set. This was shown in \cite{Atserias2005} for
structures without constant symbols, and follows easily from results
in \cite{DBLP:journals/jacm/Rossman08} even for the case of structures
with constants. Finally, it was shown in 
Proposition A.2 (1) in \cite{DBLP:journals/tods/AlexeCKT11} that
if coCSP($\mathcal{F}$)  has a finite obstruction set, then so does
coCSP($\mathcal{F}^c$).
%
%
%
\end{proof}
The following lemma reduces the problem of deciding FO-definability from
generalized CSP without constants to CSP without constants.
\begin{lemma}\label{lem:redgen}
Let $\mathcal{F}$ be a finite set of $\Sbf\cup \cbf$-structures. 
\begin{itemize}
\item If coCSP($\Bmf$) is FO-definable
for all $\Bmf\in \mathcal{F}$, then coCSP($\mathcal{F}$) is FO-definable.
\item
Conversely, if all $\Bmf\in \mathcal{F}$ are mutually homomorphically incomparable,
and coCSP($\mathcal{F}$) is FO-definable, then each coCSP($\Bmf$), $\Bmf\in \mathcal{F}$,
is FO-definable.
\end{itemize}
\end{lemma}
\begin{proof}
For Point~1 choose for every $\mathfrak{B}\in \mathcal{F}$ a FO-sentence $\varphi_{\mathfrak{B}}$ 
such that $(\dom,\Dmf)\models \varphi_{\Bmf}$ iff $\Dmf\not\rightarrow\Bmf$ for all $\Sbf$-instances $\Dmf$.
Let $\varphi$ be the conjunction over all $\varphi_{\Bmf}$ with $\Bmf\in \mathcal{F}$. Then 
$(\dom,\Dmf)\models \varphi$ iff $\Dmf\not\rightarrow\Bmf$ for any $\Bmf\in \mathcal{F}$ holds for
all $\Sbf$-instances $\Dmf$, as required.

To prove the other direction we require the notion of a \emph{critical obstruction}: a $\Sbf$-structure
$\Amf$ is called a critical obstruction for CSP($\mathcal{G}$) iff $\Amf\not\rightarrow \Bmf$ for any
$\Bmf\in \mathcal{G}$ but for any proper substructure $\Amf'$ of $\Amf$ there exists a $\Bmf\in \mathcal{F}$
such that $\Amf'\rightarrow \Bmf$. It is readily checked that coCSP($\mathcal{G}$) has a finite
obstruction set iff there only exist finitely many critical obstructions for CSP($\mathcal{G}$).

For Point~2 assume that all $\Bmf\in \mathcal{F}$ are mutually homomorphically incomparable
and that coCSP($\mathcal{F}$) is FO-definable. Assume for a proof by contradiction 
that coCSP($\Bmf_{0}$) is not FO-definable for some $\Bmf_{0}\in \mathcal{F}$. Then the set $\mathcal{C}$
of critical obstructions for CSP($\Bmf_{0}$) is infinite. Let $\Bmf_{0}'$ be a substructure of $\Bmf_{0}'$ such
that no proper substructure of $\Bmf_{0}$ can be homomorphically mapped to any 
$\Bmf\in \mathcal{F}\setminus\{\Bmf_{0}\}$.
It is readily checked that the set $\Cmc'$ of disjoint unions $\Amf \cup \Bmf_{0}'$, $\Amf\in \Cmc$,
are critical obstructions for CSP($\mathcal{F}$). Thus coCSP($\mathcal{F}$) is not FO-definable
and we have derived a contradiction.
\end{proof}

Next, we move on the datalog-definability.

\begin{lemma}
Let $\mathcal{F}$ be a finite set of $\Sbf\cup\cbf$-structures. 
\begin{enumerate}
\item If coCSP($\Bmf^{c}$) is datalog-definable
for all $\Bmf\in \mathcal{F}$, then coCSP($\mathcal{F}$) is datalog-definable.
\item
Conversely, if all $\Bmf\in \mathcal{F}$ are mutually homomorphically incomparable,
and coCSP($\mathcal{F}$) is datalog-definable, then each coCSP($\Bmf^{c}$), $\Bmf\in \mathcal{F}$,
is datalog-definable.
\end{enumerate}
\end{lemma}
\begin{proof}
  (1) If each coCSP($\Bmf^c$) is
  datalog-definable, then, since datalog is closed under conjunction,
  we also have that coCSP($\mathcal{F}^c$) is datalog-definable.
  Let $\Pi$ be a datalog program that defines coCSP($\mathcal{F}^c$).
  A datalog program $\Pi'$ defining coCSP($\mathcal{F}$) may be
  obtained from $\Pi$ by replacing every $P_{i}(x)$ with $x=c_{i}$.
  
  For (2), we make use of a characterization of
  datalog-definability in terms of \emph{obstruction sets of bounded
    treewidth}. Recall from the proof of Lemma~\ref{lem:constants} the
  notion of an obstruction set for a set of structures. Suppose that
  coCSP($\mathcal{F}$) is definable by a datalog program whose
  rules contain at most $k$ variables. Then $\mathcal{F}$ has an
  obstruction set of treewidth $k$, namely, the set of all
  canonical structures of non-recursive datalog programs obtained 
  by unfolding the given datalog program finitely many times (a
  standard argument). 

We claim that, in fact, each
  $\Bmf\in\mathcal{F}$ has an obstruction set of treewidth $k$.
  We prove this claim by contraposition: if some $\Bmf\in\mathcal{F}$ does not have an
  obstruction set of treewidth at most $k$, there is a structure $\Amf$
  such that $\Amf\not\to\Bmf$, while, at the same time,
   $\Bmf'\to\Amf$ implies $\Bmf'\to\Bmf$ for all structures
   $\Bmf'$ of treewidth at most $k$. Now, take $\Amf'$ to be 
  the disjoint union of $\Amf$ and $\Bmf$. Then we have that
  $\Amf\not\to\mathcal{F}$ (here, we are using also the fact that
  $\mathcal{F}$ consists of homomorphically incomparable
  structures). At the same time, $\Bmf'\to\Amf$ implies $\Bmf'\to\Bmf$ for all structures
   $\Bmf'$ of treewidth at most $k$. Therefore, coCSP($\mathcal{F}$)
   has no obstruction set of bounded treewidth, a contradiction.

   So far, we have shown that, for each $\Bmf\in\mathcal{F}$,
   coCSP($\Bmf$) has an obstruction set of bounded tree width.  By
   Proposition A.2 (1) in \cite{DBLP:journals/tods/AlexeCKT11}, we
   have that, for all structures $\Amf$ with constant symbols, if
   coCSP($\Amf$) has an obstruction set of bounded treewidth, then 
   coCSP($\Amf^c$) has an obstruction set of bounded treewidth too
  (although it is not explicitly stated, it can easily be verified
  that the relevant construction
  used
  there preserves bounded treewidth).
Thus, we obtain that, for each $\Bmf\in\mathcal{F}$,
  coCSP($\Bmf^c$) has an obstruction set of bounded width.
    It was shown
  in  \cite{FederVardi} that, for any structure $\Amf$ without
  constant
  symbols, coCSP($\Amf$) is datalog-definable if and only if 
  $\Amf$ has an obstruction  set of bounded tree-width. 
  Therefore we have that, for each $\Bmf\in\mathcal{F}$,
  coCSP($\Bmf^c)$ is datalog-definable. 
\end{proof}

The above lemmas, together, establish Proposition~\ref{prop:eliminating-constants}.

We now proceed with the proof of Theorem~\ref{thm:definability}.


We now give the lower bound proofs for Theorem~\ref{thm:definability}.

\begin{lemma}\label{lem:lower}
It is \NExpTime-hard to decide FO-rewritability of queries in 
($\mathcal{ALC}$,AQ) and of queries in ($\mathcal{ALC}$,BAQ). 
\end{lemma}
\begin{proof}
We prove the lower bound and employ for the reduction the same tiling problem
as in the lower bound proof of Theorem~\ref{ContNEXP}.
We also employ the ontologies constructed in the proof of Theorem~\ref{ContNEXP}.

For the reduction, let $n>0$ and $(\Tmf,H,V)$ be an instance of the $2^{n}\times 2^{n}$-tiling
problem with $\Tmf=\{T_1,\dots,T_p\}$. 
We construct a schema $\Sbf$, an $\ALC$-ontology $\Omc$ and a query $A(x)$ such that
$(\Tmf,H,V)$ has a solution if and only if $q_{\Sbf,\Omc,A(x)}$ is FO-rewritable
if and only if $q_{\Sbf,\Omc,\exists x.A(x)}$ is FO-rewritable.

We consider the ontology $\Gmc$, its extension $\Gmc^{t}$, and the schema $\Sbf_{\Gmc}$
from the proof of Theorem~\ref{ContNEXP}. To define $\Omc$, we take a fresh role name $S$
and two concept names $A$ and $F$ and set
$$
\Omc = \Gmc^{t} \cup \{\exists S.E \sqsubseteq E, E \sqcap F\sqsubseteq A\}
$$
and $\Sbf= \Sbf_\Gmc \cup \{S,F\}$.

\medskip
\noindent
\emph{Claim.} The following conditions are equivalent:
\begin{itemize}
\item $(\Tmf,H,V)$ admits no $2^{n}\times 2^{n}$-tiling;
\item $q_{\Sbf,\Omc,A(x)}$ is not FO-rewritable;
\item $q_{\Sbf,\Omc,\exists x.A(x)}$ is not FO-rewritable.
\end{itemize}

\medskip

Assume that $(\Tmf,H,V)$ admits no $2^{n}\times 2^{n}$-tiling.
$q_{\Sbf,\Omc,A(x)}$ is not FO-rewritable iff
there does not exist a finite set $\mathcal{D}$ of $\Sbf\cup\{c\}$-structures (an obstruction set)
such that the following
conditions are equivalent for every $\Sbf$-instance $\Dmf$ and $d\in \adom(\Dmf)$:
\begin{enumerate}
\item $d\in q_{\Sbf,\Omc,A(x)}(\Dmf)$.
\item there exists $\Amf\in \mathcal{D}$ such that $(\Amf,a)\rightarrow (\Dmf,d)$.
\end{enumerate}
We show that no finite obstruction set exists. To this end, we define $\Sbf$-instances 
$\Dmf_{m}$ as the union of $\Dmf_{\Gmc}$ and the facts
$$
F(a_{0}),S(a_{0},a_{1}),\ldots,S(a_{m},(0,0)).
$$
It is readily checked that
\begin{itemize}
\item $a_{0} \in q_{\Sbf,\Omc,A(x)}(\Dmf_{m})$ for all $m>0$;
\item $a_{0} \not\in q_{\Sbf,\Omc,A(x)}(\Dmf_{m}')$, where $\Dmf'_{m}$ results from $\Dmf_{m}$ by removing some fact $(a_{k},a_{k+1})$ from $\Dmf_{m}$.
\end{itemize}
It follows immediately that no finite obstruction set exists. The argument for
$q_{\Sbf,\Omc,\exists x.A(x)}$ is similar.

\medskip

Conversely, assume that $(\Tmf,H,V)$ has a $2^{n}\times 2^{n}$-tiling 
given by $f:\{0,\dots,2^{n}-1\} \times \{0,\dots,2^{n}-1\}\rightarrow \Tmf$.
We have to show that there exists an FO-formula $\varphi(x)$ over $\Sbf$ such
that for all $\Sbf$-instances $\Dmf$ and $d\in \adom(\Dmf)$, 
$(\adom(\Dmf),\Dmf)\models \varphi[d]$ iff $d\in q_{\Sbf,\Omc,A(x)}(\Dmf)$. 

Note that one can easily construct a first-order sentence $\varphi_{\Gmc}$ over $\Sbf_{\Gmc}$ such that,
for all $\Sbf_{\Gmc}$-instances $\Dmf$, the following are equivalent:
\begin{itemize}
\item $\Dmf$ is not satisfiable w.r.t.~$\Gmc$;
\item $(\adom,\Dmf)\models \varphi_{\Gmc}$.
\end{itemize}
We fix such a sentence $\varphi_{\Gmc}$ and show that the following are equivalent for
every $\Sbf$-instance $\Dmf$:
\begin{itemize}
\item $(\adom(\Dmf),\Dmf)\models \varphi_{\Gmc}$;
\item $d\in q_{\Sbf,\Omc,A(x)}(\Dmf)$.
\end{itemize}
The direction from Point~1 to Point~2 is trivial.
Conversely, assume that $(\adom(\Dmf),\Dmf)\not\models \varphi_{\Gmc}$.
Then $\Dmf$ is satisfiable w.r.t.~$\Gmc$. We define a model $(\dom,\Dmf')$ of $\Omc$ with 
$\Dmf'\supseteq \Dmf$ as follows. The domain of $\Dmf'$ coincides with $\adom(\Dmf)$.
Symbols from $\Sbf$ are defined in $\Dmf'$ in exactly the same way
as in $\Dmf$. To define the facts involving tile types $T_{k}$, associate with every 
$d\in \adom(\Dmf)$ such that ${\sf Def}$ applies to $d$, the uniquely determined pair $v(d)=(i,j)$ 
given to the values of the counters 
$X$ and $Y$ by ${\sf Def}$. Then set $T_{k}(d)\in \Dmf'$ iff $f(v(d)) = T_{k}$.
Note that $\Dmf'$ contains no facts involving $E$ or $A$.
It is readily checked that the resulting structure is a model of $\Omc$, as required.
\end{proof}

\begin{lemma}\label{datalog-rewrite-lower}
It is \NExpTime-hard to decide datalog-rewritability of queries in 
($\mathcal{ALC}$,AQ) and of queries in ($\mathcal{ALC}$,BAQ). 
\end{lemma}
\begin{proof}
The proof is based on a modification of the proof of
Lemma~\ref{lem:lower}.
For the reduction, let $n>0$ and $(\Tmf,H,V)$ be an instance of the $2^{n}\times 2^{n}$-tiling
problem with $\Tmf=\{T_1,\dots,T_p\}$. 
We construct a schema $\Sbf$, an $\ALC$-ontology $\Omc'$ and a query $A(x)$ such that
$(\Tmf,H,V)$ has a solution if and only if $q_{\Sbf,\Omc',A(x)}$ is datalog-rewritable
if and only if $q_{\Sbf,\Omc',\exists x.A(x)}$ is datalog-rewritable.

We consider the ontology $\Gmc$, its extension $\Gmc^{t}$, and the schema $\Sbf_{\Gmc}$
from the proof of Theorem~\ref{ContNEXP}. To define $\Omc'$ we take fresh role names $S$ and $H$
and fresh concept names $P_{1},P_{2},P_{3}$ and encode the 3-colorability problem as follows:
\begin{eqnarray*}
\Omc' &  = & \Gmc^{t} \cup \{\exists S.E \sqsubseteq E,\exists H.A \sqsubseteq A\}\cup \\
       &  &      \{E \sqsubseteq P_{1}\sqcup P_{2} \sqcup P_{3}\}\cup\\
       &  &      \{P_{i} \sqcap P_{j}\sqsubseteq A\mid 1\leq i< j\leq 3\}\cup\\
       &  &      \{P_{i} \sqcap \exists H.P_{i}\sqsubseteq A\mid 1\leq i \leq 3\}
\end{eqnarray*}
and $\Sbf= \Sbf_\Gmc \cup \{S,H\}$.

\medskip
\noindent
\emph{Claim.} The following conditions are equivalent:
\begin{itemize}
\item $(\Tmf,H,V)$ admits no $2^{n}\times 2^{n}$-tiling;
\item $q_{\Sbf,\Omc',A(x)}$ is not datalog-rewritable;
\item $q_{\Sbf,\Omc',\exists x.A(x)}$ is not datalog-rewritable.
\end{itemize}

\medskip

Assume that $(\Tmf,H,V)$ admits no $2^{n}\times 2^{n}$-tiling.
For any connected undirected graph $G$, 
we identify some $v$ in $G$ with $(0,0)$ and define 
a $\Sbf$-instance $\Dmf$ as the union of $\Dmf_{\Gmc}$ and the facts
$S(d,d')$ for all $d,d'$ in $G$ and $H(d,d')$ for every edge $\{d,d'\}$ in $G$. 
It is readily checked that
\begin{itemize}
\item $(0,0)\in q_{\Sbf,\Omc',A(x)}(\Dmf)$ iff $G$ is not $3$-colorable;
\item $q_{\Sbf,\Omc',\exists x.A(x)}(\Dmf)=1$ iff $G$ is not $3$-colorable.
\end{itemize}
It follows immediately that neither $q_{\Sbf,\Omc',A(x)}$ nor 
$q_{\Sbf,\Omc',\exists x.A(x)}$ are datalog-rewritable.

\medskip

Conversely, if $(\Tmf,H,V)$ admits a $2^{n}\times 2^{n}$-tiling 
then one can show datalog-rewritability using exactly the same argument as in the proof
of Lemma~\ref{lem:lower}.
\end{proof}

We now prove the undecidability results for $\mathcal{ALCF}$.
In \cite{KR12-cont,KR12-csp}, alternative definitions of query containment and FO-rewritability
are employed which consider only instances that are satisfiable w.r.t.~the ontologies
involved. We say that $(\Sbf,\Omc_{1},q_{1})$ \emph{is contained in} $(\Sbf,\Omc_{2},q_{2})$
\emph{w.r.t.~consistent instances} if $q_{(\Sbf,\Omc_{1},q_{1})}(\Dmf) \subseteq
q_{(\Sbf,\Omc_{2},q_{2})}(\Dmf)$ for all $\Sbf$-instance $\Dmf$ such that $\Dmf$ is satisfiable
w.r.t.~$\Omc_{1}$. Similarly, a query $(\Sbf,\Omc,q)$ is \emph{FO-rewritable}
w.r.t~\emph{consistent instances} if there exists an FO-query $q'$ such that
$q'(\Dmf)=q_{(\Sbf,\Omc,q)}(\Dmf)$ for all $\Sbf$-instance $\Dmf$ that are
satisfiable w.r.t.~$\Omc$.
Undecidability of query containment w.r.t.~consistent instances and of FO-rewritability 
w.r.t.~consistent instances were proven
respectively in \cite{KR12-cont} and \cite{KR12-csp}. Here we show how the proofs can be modified
to work for query containment, FO-rewritability, and datalog rewritability as defined in this paper.

\begin{theorem}\label{thm:undec0}
Query containment, FO-rewritability, and datalog-rewritability are all 
undecidable for queries in ($\mathcal{ALCF}$,AQ) and queries in ($\mathcal{ALCF}$,BAQ).
\end{theorem}
\begin{proof}
The proof is by reduction of the following finite rectangle tiling problem.
An instance of the \emph{finite rectangle tiling problem} is given by a
triple $\mathfrak{P}= (\Tmf,H,V)$ with 
\begin{itemize}
\item $\Tmf=\{T_{1},\ldots,T_{p}\}$ a non-empty, finite set of \emph{tile types} including an \emph{initial tile} $T_\mn{init}$ 
to be placed on the lower left corner, a \emph{final tile} $T_\mn{final}$ to be
placed on the upper right corner, and sets $\Umf\subseteq \Tmf$ and $\Rmf\subseteq \Tmf$ of tile types to be placed on the upper and right borders respectively, satisfying
$\Umf\cap \Rmf=\{T_{\mn{final}}\}$;
\item $H \subseteq \Tmf \times \Tmf$ a
\emph{horizontal matching relation}; and 
\item $V \subseteq \Tmf \times
\Tmf$ a \emph{vertical matching relation}. 
\end{itemize}
A \emph{tiling} for
$(\Tmf,H,V)$ is a map $f:\{0,\dots,n\} \times \{0,\dots,m\}
\rightarrow \Tmf$ such that $n,m \geq 0$, 
\begin{itemize}
\item $f(0,0)=T_\mn{init}$,
\item $f(n,m)=T_\mn{final}$, 
\item $f(n,j)\in \Rmf$ for all $0 \leq j\leq m$;
\item $f(j,i)\not\in \Rmf$ for all $j<n$ and $0 \leq i\leq m$;
\item $f(i,m)\in \Umf$ for all $0 \leq i\leq n$;
\item $f(i,j)\not\in\Umf$ for all $0 \leq i\leq n$ and $1 \leq j<m$. 
\item $(f(i,j),f(i+1,j)) \in H$ for all $0 \leq i < n$, and
\item $(f(i,j),f(i,j+1)) \in v$ for all $0 \leq i < m$.
\end{itemize}
Thus, we can assume that $H$, $V$, $\Umf$, and $\Rmf$ are such that: 
\begin{itemize}
\item if $(T_{i},T_{j})\in H$, then $T_{i}\in \Umf$ if and only if $T_{j}\in \Umf$;
\item if $T_{i}\in \Umf$, then there exists no $T_{j}$ with $(T_{i},T_{j})\in V$ or
$(T_{j},T_{i})\in V$; 
\item if $(T_{i},T_{j})\in V$, then $T_{i}\in \Rmf$ if and only if $T_{j}\in \Rmf$;
\item if $T_{i}\in \Rmf$, then there exists no $T_{j}$ with $(T_{i},T_{j})\in H$ or $(T_{j},T_{i})\in H$.
\end{itemize}
It is undecidable whether
an instance $\mathfrak{P}$ of the finite rectangle tiling problem has a tiling. 

Fix a particular $\mathfrak{P}=(\Tmf,H,V)$. 
For the data schema, we use $\Sbf = \{ T_1,\dots,T_p,x,y, x^{-}, y^{-}\}$,
where $T_1,\dots,T_p$ are treated as concept names, and $x$, $y$,
$x^{-}$, and $y^{-}$ are  role names. We use $x$ and $y$ to specify horizontal and vertical adjacency of
points in the rectangle, and the role names $x^{-}$ and $y^{-}$ to
simulate the inverses of $x$ and $y$ (note that since $x^{-}$ and $y^{-}$
are regular role names, they need not be interpreted 
as the inverses of $x$ and $y$). 
We construct an $\ALCF$-ontology $\Omc_{\mathfrak{P}}$ which asserts 
functionality of $x,y,x^{-}, y^{-}$ and contains inclusions
using additional concept names 
$U,R,Y, I_x, I_y, C, Z_{c,1}$, $Z_{c,2}$, $Z_{x,1}$, $Z_{x,2}$, $Z_{y,1}$.
The concept names $U$ and $R$ are used to mark the upper and
right border of the rectangle, $Y$ is used to mark points in the rectangle,
and the remaining concept names are used for
technical purposes explained below. 
In the following, for $e \in \{c, x, y\}$, we let
$\Bmc_{e}$ range over all Boolean combinations of the concept names
$Z_{e,1}$ and $Z_{e,2}$, i.e., over all concepts $L_1 \sqcap L_2$
where $L_i$ is a literal over $Z_{e,i}$, for $i \in \{1,2\}$.
The ontology $\Omc_{\mathfrak{P}}$ 
contains the following concept inclusions, 
where $(T_i,T_j) \in H$ and $(T_i,T_\ell) \in V$:
$$
\begin{array}{rcl}
 T_\mn{final} & \sqsubseteq & Y \sqcap U \sqcap R \\
 \exists x . (U \sqcap Y \sqcap T_j) \sqcap I_x \sqcap T_i & \sqsubseteq & U \sqcap Y  \\
 \exists y . (R \sqcap Y \sqcap T_{\ell}) \sqcap I_y \sqcap T_{i} & \sqsubseteq & R \sqcap Y \\
 \exists x . (T_{j} \sqcap Y \sqcap \exists y . Y) \hspace*{1.5cm}& & \\
\sqcap \,  \exists y . (T_{\ell} \sqcap Y \sqcap \exists x . Y) \hspace*{.5cm} & & \\
 \sqcap I_x \sqcap I_y \sqcap C \sqcap T_i & \sqsubseteq & Y  \\  
  \exists x . \exists y . \Bmc_c \sqcap \exists y . \exists x . \Bmc_c 
 &\sqsubseteq& C  \\ 
    \Bmc_x \sqcap \exists x. \exists x^{-}. \Bmc_{x} & \sqsubseteq & I_x \\
  \Bmc_y \sqcap \exists y. \exists y^{-}. \Bmc_{y} & \sqsubseteq & I_y 
\end{array}
$$
$$
\begin{array}{rcl}
 T_{i} &  \sqsubseteq & \forall y. \bot\\
 T_{j} &  \sqsubseteq & \forall x. \bot\\
 U &  \sqsubseteq & \forall x. U\\
 R &  \sqsubseteq & \forall y. R\\
 \midsqcup_{1 \leq s < t \leq p} T_s \sqcap T_t & \sqsubseteq & \bot\\
\end{array}
$$
where $T_{i}\in \Umf$ and $T_{j}\in \Rmf$. 

The first four inclusions propagate the concept $Y$ downwards 
and leftwards starting from a point marked with the 
final tile $T_\mn{final}$. Note that 
these inclusions enforce the horizontal 
and vertical matching conditions. 
The concept inclusion with right-hand side $C$ 
serves to enforce confluence,
i.e., $C$ is entailed at a constant $a$ if there is a constant
$b$ that is both an $x$-$y$-successor and a
$y$-$x$-successor of $a$.  
This is so because, intuitively, $\Bmc_{c}$ is
universally quantified: if confluence fails, 
then we can interpret $Z_{c,1}$ and $Z_{c,2}$ so that
neither of the two conjuncts on the left-hand side of
the inclusion for $C$ is satisfied. 
In a similar manner, the inclusion for $I_x$ (resp. $I_y$)
is used to ensure that $x^{-}$ (resp. $y^{-}$) 
act as the inverse of $x$ (resp. $y$)
at all points in the rectangle.

The following property can be obtained by a minor modification of 
Lemma 30 in \cite{KR10}:

\begin{lemma}\label{undec-tiling-lemma}
  $\mathfrak{P}$ admits a tiling if and only if there is a $\Sbf$-instance $\Dmf$
  which is consistent with $\Omc_{\mathfrak{P}}$ and such that 
  $q_{\Sbf, \Omc_{\mathfrak{P}}, T_\mn{init}(x) \wedge Y(x)}(\Dmf)\neq \emptyset$.
\end{lemma}

Let $\varphi_{\mathfrak{P}}$ be the first-order 
translation of the conjunction of all $T_{i} \sqsubseteq \forall y. \bot$, $T_{i}\in \Umf$,
$T_{j} \sqsubseteq \forall x. \bot$, $T_{j}\in \Rmf$,
and of $\midsqcup_{1 \leq s < t \leq p} T_s \sqcap T_t  \sqsubseteq  \bot$.
The following is readily checked:

\medskip

\noindent
\emph{Claim.}
For all $\Sbf$-instances $\Dmf$, $(\adom(\Dmf),\Dmf)\models \varphi_{\mathfrak{P}}$
iff $\Dmf$ is satisfiable w.r.t.~$\Omc_{\mathfrak{P}}$. 

\medskip

We now prove undecidability of query containment.
Let $E$ be a fresh concept name and let
$$
\Omc_{2}= \Omc_{\mathfrak{P}}\cup\{E\sqsubseteq E\}, \quad
\Omc_{1}= \Omc_{\mathfrak{P}}\cup \{ Y \sqcap T_\mn{init} \sqsubseteq E\}
$$ 
Now one can prove that the following conditions are equivalent:

\begin{itemize}
\item $\mathfrak{P}$ admits a tiling;
\item $(\Sbf,\Omc_{1},E(x))$ is not contained in $(\Sbf,\Omc_{2},E(x))$;
\item $(\Sbf,\Omc_{1},\exists x.E(x))$ is not contained in $(\Sbf,\Omc_{2},\exists x.E(x))$
\end{itemize}
Assume first that $\mathfrak{P}$ admits a tiling. 
Then by Lemma \ref{undec-tiling-lemma}, there is a $\Sbf$-instance $\Dmf$
which is consistent with $\Omc_{\mathfrak{P}}$ and such that 
$q_{\Sbf, \Omc_{\mathfrak{P}}, T_\mn{init}(x) \wedge Y(x)}(\Dmc)\neq \emptyset$.
It follows immediately that 
$q_{\Sbf, \Omc_{1}, E(x)}(\Dmf) \neq \emptyset$
and
$q_{\Sbf,\Omc_{1},\exists x.E(x)}(\Dmf)=1$. 
On the other hand, since $\Dmf$ is consistent with $\Omc_{2}$, 
and $E$ appears only trivially in $\Omc_{2}$,
we have $q_{\Sbf,\Omc_{2},E(x)}(\Dmf)=\emptyset$ and 
$q_{\Sbf,\Omc_{2},\exists x.E(x)}(\Dmf)=0$. 

Next suppose that $\mathfrak{P}$ does not admit a tiling, and 
let $\Dmf$ be an $\Sbf$-instance which is consistent with $\Omc_{1}$. 
By Lemma \ref{undec-tiling-lemma}, 
$q_{\Sbf, \Omc_{\mathfrak{P}}, T_\mn{init}(x) \wedge Y(x)}(\Dmc)=\emptyset$,
and hence $q_{\Sbf, \Omc_{1}, \exists x. E(x)}(\Dmf)=0$.
The desired containments trivially follow.  

\medskip

To prove undecidability of FO-rewritability, we expand $\Omc_{1}$ to a new ontology
$\Omc_{3}$.
To define $\Omc_{3}$ we take a fresh role name $S$
and two concept names $A$ and $F$ and set
$$
\Omc_{3} = \Omc_{1} \cup \{\exists S.E \sqsubseteq E, E \sqcap F\sqsubseteq A\}
$$
and $\Sbf_{3}= \Sbf \cup \{S,F\}$.

\medskip
\noindent
\emph{Claim.} The following conditions are equivalent:
\begin{itemize}
\item $\mathfrak{P}$ admits a tiling;
\item $q_{\Sbf_{3},\Omc_{3},A(x)}$ is not FO-rewritable;
\item $q_{\Sbf_{3},\Omc_{3},\exists x.A(x)}$ is not FO-rewritable.
\end{itemize}
Assume first that $\mathfrak{P}$ admits a tiling. By Lemma \ref{undec-tiling-lemma}, 
we can find an $\Sbf$-instance $\Dmf_{\mathfrak{P}}$
which is consistent with $\Omc_{\mathfrak{P}}$ and $b \in \adom(\Dmf_{\mathfrak{P}})$ such that 
$b \in q_{\Sbf, \Omc_{\mathfrak{P}}, T_\mn{init}(x) \wedge Y(x)}(\Dmf_{\mathfrak{P}})$, 
and hence $b \in q_{\Sbf, \Omc_{1}, E(x)}(\Dmf_{\mathfrak{P}})$.
We can use essentially the same argument as in Lemma \ref{lem:lower}
to show that $q_{\Sbf, \Omc_{1}, E(x)}$ and $q_{\Sbf, \Omc_{1}, E(x)}$ are not FO-rewritable. 
Specifically, we
construct $\Sbf$-instances 
$\Dmf_{m}$ by taking the union of $\Dmf_{\mathfrak{P}}$ and the facts 
$$
F(a_{0}),S(a_{0},a_{1}),\ldots,S(a_{m},b).
$$
It is readily checked that
\begin{itemize}
\item $a_{0} \in q_{\Sbf_{3},\Omc_{3},A(x)}(\Dmf_{m})$ for all $m>0$;
\item $a_{0} \not\in q_{\Sbf_{3},\Omc_{3},A(x)}(\Dmf_{m}')$, where $\Dmf'_{m}$ results from $\Dmf_{m}$ by removing some fact $(a_{k},a_{k+1})$ from $\Dmf_{m}$.
\end{itemize}
It follows that no finite obstruction set exists, and hence that 
$q_{\Sbf, \Omc_{1}, A(x)}$ is not FO-rewritable. 
We can proceed similarly for $q_{\Sbf, \Omc_{1}, \exists x. A(x)}$.

Assume now that $\mathfrak{P}$ does not admit a tiling.
Then for every $\Sbf$-instance $\Dmf$, $\Dmf$ is satisfiable w.r.t.~$\Omc_{\mathfrak{P}}$ if and only if
$q_{\Sbf,\Omc_{3},\exists x.A(x)}(\Dmf)= 0$. Thus, the query defined by $\neg\varphi_{\mathfrak{P}}$
is equivalent to $q_{\Sbf,\Omc_{3},\exists x.A(x)}$,
 and the query defined by $(x=x) \wedge \neg\varphi_{\mathfrak{P}}$
is equivalent to $q_{\Sbf,\Omc_{3},A(x)}$.

\medskip

To prove undecidability of datalog-rewritability, we expand $\Omc_{1}$ to a new ontology
$\Omc_{4}$.
To define $\Omc_{4}$, we take fresh role names $S$ and $H$
and fresh concept names $P_{1},P_{2},P_{3}$ and encode the 3-colorability problem as follows:
\begin{eqnarray*}
\Omc_{4} &  = & \Gmc_{1} \cup \{\exists S.E \sqsubseteq E,\exists H.A \sqsubseteq A\}\cup \\
       &  &      \{E \sqsubseteq P_{1}\sqcup P_{2} \sqcup P_{3}\}\cup\\
       &  &      \{P_{i} \sqcap P_{j}\sqsubseteq A\mid 1\leq i< j\leq 3\}\cup\\
       &  &      \{P_{i} \sqcap \exists H.P_{i}\sqsubseteq A\mid 1\leq i \leq 3\}
\end{eqnarray*}
We use the schema $\Sbf_{4}= \Sbf \cup \{S,H\}$.

\medskip
\noindent
\emph{Claim.} The following conditions are equivalent:
\begin{itemize}
\item $\mathfrak{P}$ admits a tiling;
\item $q_{\Sbf_{4},\Omc_{4},A(x)}$ is not datalog-rewritable;
\item $q_{\Sbf_{4},\Omc_{4},\exists x.A(x)}$ is not datalog-rewritable.
\end{itemize}

\noindent First suppose that $\mathfrak{P}$ admits a tiling. 
We have seen previously that this implies the existence of 
an $\Sbf$-instance $\Dmf_{\mathfrak{P}}$
which is consistent with $\Omc_{\mathfrak{P}}$ and contains 
$b \in \adom(\Dmf_{\mathfrak{P}})$ such that
$b \in q_{\Sbf, \Omc_{1}, E(x)}(\Dmf_{\mathfrak{P}})$.
We proceed similarly to Lemma \ref{datalog-rewrite-lower}.
Given a connected undirected graph $G$, 
we define an $\Sbf$-instance $\Dmf$ as the union of $\Dmf_{\mathfrak{P}}$ and the facts
$S(d,d')$ for all $d,d'$ in $G$ and $H(d,d')$ for every edge $\{d,d'\}$ in $G$. 
It is readily checked that
\begin{itemize}
\item $b \in q_{\Sbf_{4},\Omc_{4},A(x)}$ iff $G$ is not $3$-colorable;
\item $q_{\Sbf_{4},\Omc_{4},\exists x. A(x)}(\Dmf)=1$ iff $G$ is not $3$-colorable.
\end{itemize}
It follows directly that neither $q_{\Sbf,\Omc',A(x)}$ nor 
$q_{\Sbf,\Omc',\exists x.A(x)}$ are datalog-rewritable.

Next suppose that $\mathfrak{P}$ does not admit a tiling.
Then for every $\Sbf$-instance $\Dmf$, we have that
$\Dmf$ is satisfiable w.r.t.~$\Omc_{\mathfrak{P}}$ if and only if
$q_{\Sbf,\Omc_{4},\exists x.A(x)}(\Dmf)= 0$. We can then simply reuse 
the FO-rewritings $\neg\varphi_{\mathfrak{P}}$
and $(x=x) \wedge \neg\varphi_{\mathfrak{P}}$ from above, 
since these can be equivalently expressed as datalog queries. 
\end{proof}




 

\end{document}